\documentclass{article}

\usepackage{amssymb, amsmath, amsthm}

\usepackage[english]{babel}
\usepackage[utf8]{inputenc}
\usepackage{textcomp}
\usepackage{hyperref}
\usepackage[T1]{fontenc}
\usepackage{comment}
\usepackage{todonotes}
\usepackage[capitalise]{cleveref}
\usepackage{fullpage}
\usepackage{mathtools}
\usepackage{caption}
\usepackage{multirow}

\title{Fast Hashing with Strong Concentration Bounds\footnote{An extended abstract appeared at the 52nd Annual ACM Symposium on Theory of Computing (STOC'20).}}
\author{Anders Aamand\footnote{Basic Algorithms Research Copenhagen (BARC), University of Copenhagen.} \and Jakob B.\ T.\ Knudsen$^\dagger$ \and Mathias B.\ T.\ Knudsen\footnote{SupWiz ApS.} \and Peter M. R. Rasmussen$^\dagger$ \and Mikkel Thorup$^\dagger$}
\date{\today}

\newcommand{\abs}[1]{\left\vert #1\right\vert}

\newcommand{\eps}{\varepsilon}

\newcommand{\R}{\mathbb{R}}
\newcommand{\N}{\mathbb{N}}
\newcommand{\Z}{\mathbb{Z}}

\newcommand{\C}{\mathcal C}

\newcommand\cU{{\mathcal U}}

\newcommand\fct{\rightarrow}

\newcommand{\F}{\mathcal F}
\newcommand{\FF}{\mathbb F}

\newcommand{\PR}[1]{\Pr\left[ #1 \right]}
\newcommand{\PRC}[2]{\Pr\left[ #1 \, \middle\vert \, #2\right]}

\newcommand\drop[1]{}

\newcommand{\req}[1]{\eqref{#1}}

\newcommand{\E}[1]{\mathbb E\left[#1\right]}
\newcommand{\EC}[2]{\E{\left. #1 \, \right\vert \, #2}}
\newcommand{\Var}[1]{\mathrm{Var}\left[ #1 \right]}
\newcommand{\VarC}[2]{\Var{\left. #1 \, \right\vert \, #2 }}
\newcommand{\Covar}[2]{\mathrm{Cov}\left(#1, #2 \right)}
\newcommand{\CovarC}[3]{\Covar{#1}{#2 \; \middle\vert \; #3}}
\newcommand{\indicator}[1]{\left[ \, #1 \, \right]}
\newcommand{\set}[1]{\left\{ #1 \right\}}
\newcommand{\setbuilder}[2]{\set{ #1 \; \middle\vert \; #2 }}

\newcommand{\chernoff}[1]{\C\left( #1 \right)}

\newcommand{\ceil}[1]{\lceil #1 \rceil}

\newcommand{\cupdot}{\mathbin{\mathaccent\cdot\cup}}

\newcommand\errorterm{added error probability }
\newcommand\Errorterm{Added error probability}

\makeatletter
\newtheorem*{rep@theorem}{\rep@title}
\newcommand{\newreptheorem}[2]{%
\newenvironment{rep#1}[1]{%
 \def\rep@title{#2 \ref{##1}}%
 \begin{rep@theorem}}%
 {\end{rep@theorem}}}
\makeatother

\newtheorem{theorem}{Theorem}
\newreptheorem{theorem}{Theorem}

\newtheorem{lemma}[theorem]{Lemma}
\newreptheorem{lemma}{Lemma}
\newtheorem{corollary}[theorem]{Corollary}

\theoremstyle{definition}
\newtheorem{definition}{Definition}
\theoremstyle{remark}
\newtheorem*{remark}{Remark}

\graphicspath{{figures/}}

\begin{document}
\pagenumbering{roman}
\maketitle

\begin{abstract}
    Previous work on tabulation hashing by P{\v a}tra{\c s}cu and Thorup
from STOC'11 on simple tabulation and from SODA'13 on twisted tabulation
offered Chernoff-style concentration bounds on hash based sums, e.g., the number of balls/keys hashing to a given bin, but
under some quite severe restrictions on the expected values of these
sums. The basic idea in tabulation hashing is to view
a key as consisting of $c=O(1)$ characters, e.g., a 64-bit key as
$c=8$ characters of 8-bits. The character domain $\Sigma$ should be
small enough that character tables of size $|\Sigma|$ fit in fast
cache. The schemes then use $O(1)$ tables of this size,
so the space of tabulation hashing is $O(|\Sigma|)$. However, the concentration bounds by P{\v a}tra{\c s}cu and Thorup only apply if the expected sums are $\ll |\Sigma|$.

To see the problem, consider the very simple case where we use
tabulation hashing to throw $n$ balls into $m$ bins and want to analyse the number of balls in a given bin. With their concentration bounds, we are fine if $n=m$, for then the expected value is $1$.
However, if $m=2$, as when tossing $n$
unbiased coins, the expected value $n/2$ is $\gg |\Sigma|$ for large data sets, e.g.,
data sets that do not fit in fast cache.

To handle expectations that go beyond the limits of our small
space, we need a much more advanced analysis of simple tabulation,
plus a new tabulation technique that we call
\emph{tabulation-permutation} hashing which is at most twice as
slow as simple tabulation. No other hashing scheme of comparable speed offers
similar Chernoff-style concentration bounds.
\end{abstract}

\newpage
\tableofcontents
\newpage
\pagenumbering{arabic}
\def\cC{\mathcal{C}}
\def\deltaerr{\varepsilon}
\def\xor{\oplus}

\section{Introduction}
Chernoff's concentration bounds~\cite{Che52:chernoff} date back to the
1950s but bounds of this types go even further back to Bernstein in the 1920s~\cite{Bernstein1924}. Originating from the area of statistics they are now one of the
most basic tools of randomized algorithms \cite{motwani95book}. A
canonical form considers the sum $X=\sum_{i=1}^nX_i$ of independent
random variables $X_1, \dots, X_n\in [0, 1]$. Writing $\mu = \E X$ it 
holds for every $\deltaerr\geq0$ that
\begin{align}\label{eq:classic-chernoff+}
\Pr[X\ge (1+\deltaerr)\mu]& \le \exp(-\mu\;\cC(\deltaerr))& \left[\;\leq\ 
\exp(-\deltaerr^2\mu/3)\textnormal{ for }\deltaerr\leq 1\right],\\
\label{eq:classic-chernoff-}
\Pr[X\le (1-\deltaerr)\mu]& \le \exp(-\mu\;\cC(-\deltaerr))& \left[\;\leq\ 
\exp(-\deltaerr^2\mu/2)\textnormal{ for }\deltaerr\leq 1\right].
\end{align}
Here $\mathcal{C}: (-1,\infty) \to [0,\infty)$ is given by 
$\cC(x)=(x+1)\ln (x+1)-x$, so $\exp(-\C(x))=
   \frac{e^x}{(1+x)^{(1+x)}}$. Textbook proofs
of \req{eq:classic-chernoff+} and \req{eq:classic-chernoff-} can be
found in \cite[\S 4]{motwani95book}\footnote{The bounds in \cite[\S
4]{motwani95book} are stated as working only for $X_i\in\{0,1\}$, but
the proofs can easily handle any $X_i\in[0,1]$.}. Writing $\sigma^2=\Var X$, a more general bound is
\begin{align}
    \label{eq:var-chernoff}
\Pr[|X-\mu|\geq t]& \le 2\exp(-\sigma^2\cC(t/\sigma^2)) & \left[\;\leq\ 
2\exp(-(t/\sigma)^2/3)\textnormal{ for }t\leq \sigma^2\right].
\end{align}
Since $\sigma^2\leq\mu$ and $\cC(-\deltaerr)\leq 1.5\,\cC(\deltaerr)$ for $\deltaerr\leq 1$, \eqref{eq:var-chernoff} is at least as good as \eqref{eq:classic-chernoff+} and \eqref{eq:classic-chernoff-}, up to constant factors, and often better. In this work, we state our results in relation to \req{eq:var-chernoff}, known as Bennett's inequality
\cite{BennetInequality}.

Hashing is another fundamental tool of randomized algorithms dating
back to the 1950s \cite{Dum56}. A random hash function, $h:U\fct R$,
assigns a hash value, $h(x)\in R$, to every key $x\in U$. Here
both $U$ and $R$ are typically bounded integer ranges. The original
application was hash tables with chaining where $x$ is placed in bin
$h(x)$, but today, hash functions are ubiquitous in randomized algorithms. For instance, 
 they play a fundamental role in streaming
and distributed settings where a system uses a hash function to
coordinate the random choices for a given key. In most applications, we require concentration bounds for one of the following cases of increasing generality.
\begin{enumerate}
    \item Let $S\subseteq U$ be a set of balls and assign to each ball, $x\in S$, a weight, $w_x\in [0, 1]$. We wish to distribute the balls of $S$ into a set of bins $R=[m] = \{0, 1, \dots, m-1\}$. For a bin, $y\in [m]$, $X=\sum_{x\in S}w_x\cdot [h(x)=y]$ is then the total weight of the balls landing in bin $y$.
    \item We may instead be interested in the total weight of the balls with hash values in the interval $[y_1, y_2)$ for some $y_1, y_2\in [m]$, that is, $X=\sum_{x\in S}w_x\cdot [y_1\leq h(x)< y_2]$.    
    \item More generally, we may consider a fixed \emph{value function} $v\colon U\times R\to[0, 1]$. For each key $x\in U$, we define the random variable $X_x=v(x,h(x))$, where the randomness of $X_x$ stems from that of $h(x)$. We write $X=\sum_{x\in U}v(x, h(x))$ for the sum of these values. 
    %More generally, we may consider a fixed \emph{value function} $v\colon U\times R\to[0, 1]$. Each key $x\in U$ is now weighted according to its hash value $h(x)\in R$. We write $X=\sum_{x\in U}v(x, h(x))$ for the sum of these weights. 
\end{enumerate}
To exemplify applications, the first case is common when trying to allocate resources; the second case arises in streaming algorithms; and the
 third case handles the computation of a complicated statistic, $X$, on incoming data.
In each case, we wish the variable $X$ to be concentrated around its mean, $\mu = \E X$, according to the Chernoff-style bound of \eqref{eq:var-chernoff}.
If we had fully random hashing, this would indeed be the case. However, storing a fully random hash function is infeasible. The goal of this paper is to obtain such concentration with a practical constant-time hash function. More specifically, we shall construct hash functions that satisfy the following definition when $X$ is a random variable as in one of the three cases above.%\todo{With this definition, we should spend more time talking about query balls and the choice of value function}
\begin{definition}[Strong Concentration]\label{def:strongly-concentrated}
	Let $h\colon [u]\to [m]$ be a hash function, $S\subseteq [u]$ be a set of hash keys of size $n = \abs S$, and $X=X(h, S)$ be a random variable, which is completely determined by $h$ and $S$. Denote by $\mu = \E X$ and $\sigma^2=\Var{X}$ the expectation and variance of $X$. We say that $X$ is \emph{strongly concentrated with \errorterm $f(u, n, m)$} if for every $t>0$,
	\begin{align}
		\Pr\left[\abs{X-\mu}\geq t\right]\leq O\left(\exp\left(-\Omega(\sigma^2\C(t/\sigma^2)\right)\right)+f(u, n, m).
	\end{align}
\end{definition}

Throughout the paper we shall prove properties of random variables that are determined by some hash function. In many cases, we would like these properties to continue to hold while conditioning the hash function on its value on some hash key.
\begin{definition}[Query Invariant]\label{def:query-invariance}
	Let $h\colon [u]\to [m]$ be a hash function, let $X=X(h)$ be a random variable determined by the outcome of $h$, and suppose that some property $T$ is true of $X$. We say that the property is \emph{query invariant} if whenever we choose $x\in [u]$ and $y\in [m]$ and consider the hash function $h'= (h|h(x)=y)$, i.e.,\ $h$ conditioned on $h(x)=y$, property $T$ is true of $X'=X(h')$.
\end{definition}
\begin{remark}
For example, consider the case (1) from above. We are interested in the random variable $X=\sum_{x\in S}w_x\cdot [h(x)=y]$. Suppose that for every choice of weights, $(w_x)_{x\in S}$, $X$ is strongly concentrated and that this concentration is query invariant. Let $x_0\in [u]$ be a distinguished query key. Then since for every $y_0\in [m]$, the hash function $h' = (h|h(x_0)=y_0)$ satisfies that $X'=\sum_{x\in S}w_x\cdot [h'(x)=y_0]$ is strongly concentrated, it follows that $X'' =\sum_{x\in S}w_x\cdot [h(x)=h(x_0)]$ is strongly concentrated. Thus, $h$ allows us to get Chernoff-style concentration on the weight of the balls landing in the same bin as $x_0$. 

This may be generalized such that in the third case from above, the weight function may be chosen as a function of $h(x_0)$. Thus, the property of being query invariant is very powerful. It is worth noting that the constants of the asymptotics may change when conditioning on a query. Furthermore, the expected value and variance of $X'$ may differ from that of $X$, but this is included in the definition.
\end{remark}

One way to achieve Chernoff-style bounds in all of the above cases is through the classic
$k$-independent hashing framework of Wegman and
Carter \cite{wegman81kwise}. The random hash function $h:U\fct R$ is
$k$-independent if for any $k$ distinct keys $x_1,\ldots,x_k\in U$,
$(h(x_1),\ldots,h(x_k))$ is uniformly distributed in $R^k$. Schmidt
and Siegel \cite{schmidt95chernoff} have shown that with $k$-independence,
the above Chernoff
bounds hold with an \errorterm decreasing exponentially in
$k$. Unfortunately, a lower bound by Siegel \cite{siegel04hash} implies that evaluating a $k$-independent hash function takes
$\Omega(k)$ time unless we use a lot of space (to be detailed later).

P\v{a}tra\c{s}cu and Thorup have shown that Chernoff-style bounds can
be achieved in constant time with tabulation based hashing methods;
namely simple tabulation \cite{patrascu12charhash} for the first case described above and twisted
tabulation \cite{PT13:twist} for all cases. However, their results suffer from some
severe restrictions on the expected value, $\mu$, of the sum. More
precisely, the speed of these methods relies on using space small
enough to fit in fast cache, and the Chernoff-style
bounds \cite{patrascu12charhash,PT13:twist} all require that $\mu$ is
much smaller than the space used. For larger values of $\mu$,
P\v{a}tra\c{s}cu and Thorup \cite{patrascu12charhash,PT13:twist}
offered some weaker bounds with a deviation that was off by several logarithmic
factors. It can be shown that some of these limitations are inherent to simple and twisted tabulation. For instance, they cannot even reliably distribute balls into $m=2$ bins, as described in the first case above, if the expected number of balls in each bin exceeds the space used.

In this paper, we construct and analyse a new family of fast hash functions \emph{tabulation-permutation} hashing that has Chernoff-style concentration bounds like~\eqref{eq:var-chernoff} without any restrictions  
on $\mu$. This generality is important if building a general online system with no knowledge of future input. Later, we shall give concrete examples from
streaming where $\mu$ is in fact large.
Our bounds hold for all of the cases described above and all possible inputs. Furthermore, tabulation-permutation hashing is an order of magnitude faster than any other known hash function with similar concentration bounds, and almost as fast as simple and twisted tabulation. We demonstrate this both theoretically and experimentally.
Stepping back, our main theoretical contribution lies in the field of analysis of algorithms, and is in the spirit of Knuth's analysis of linear probing~\cite{knuth63linprobe}, which shows strong theoretical guarantees for a very practical algorithm. We show that tabulation-permutation hashing has strong theoretical Chernoff-style concentration bounds. Moreover, on the practical side, we perform experiments, summarized in \Cref{tab:speed}, demonstrating that it is comparable in speed to some of the fastest hash functions in use, none of which provide similar concentration bounds.

When talking about hashing in constant time, the actual size of the constant is of crucial importance. First, hash functions typically execute the same instructions on all keys, in which case we always incur the worst-case running time. Second, hashing is often an inner-loop bottle-neck of data processing. Third, hash functions are often applied in time-critical settings. Thus, even speedups by a multiplicative constant are very impactful. 
 As an example from the Internet,
suppose we want to process packets passing through a high-end
Internet router. Each application only gets very limited time to look
at the packet before it is forwarded.  If it is not done in time,
the information is lost. Since processors and routers use some of the
same technology, we never expect to have more than a few instructions
available. Slowing down the Internet is typically not an option.
The papers of Krishnamurthy et al.\ \cite{KSZC03} and Thorup and Zhang \cite{thorup12kwise} explain in more detail how high
speed hashing is necessary for their Internet traffic analysis.
Incidentally, our hash function is a bit faster than the ones
from \cite{KSZC03,thorup12kwise}, which do not provide Chernoff-style
concentration bounds.

Concrete examples of the utility of our new hash-family may be found in \cite{NoRepetitions}. In \cite{NoRepetitions} it is shown that some classic streaming algorithms enjoy very substantial speed-ups when implemented using tabulation-permutation hashing; namely
the original similarity estimation of Broder~\cite{Broder97onthe}
and the estimation of distinct elements of Bar-Yossef et
al.~\cite{BJKST02}. The strong concentration bounds makes the use of independent repetitions unnecessary, allowing the implementations of the algorithms to be both simpler and faster. We stress that in high-volume streaming algorithms, speed is of critical importance.

Tabulation-permutation hashing builds on top of simple tabulation hashing, and to analyse it, we require a new and better understanding of the behaviour and inherent limitations of simple tabulation, which we proceed to describe. Afterwards we break these limitations by introducing our new powerful tabulation-permutation hashing scheme.

\subsection{Simple Tabulation Hashing}
\emph{Simple tabulation} hashing dates back to 
Zobrist \cite{zobrist70hashing}. In simple tabulation hashing, we consider the key domain $U$ to be of the form $U=\Sigma^c$ for some character alphabet $\Sigma$ and $c=O(1)$, such that each key consists of $c$ characters of $\Sigma$. 
Let $m=2^\ell$ be given and identify $[m]=\{0, 1, \dots, m-1\}$ with $[2]^\ell$.
A simple tabulation hash function, $h\colon \Sigma^c\to [m]$, is then
defined as follows. 
For each $j\in \{1, \dots, c\}$ store a fully random character table $h_j\colon \Sigma\to [m]$ mapping
characters of the alphabet $\Sigma$ to $\ell$-bit hash values.  To evaluate $h$ on a key
$x=(x_1, \dots, x_c)\in \Sigma^c$, we compute 
$h(x) =
h_1(x_1)\oplus \cdots \oplus h_c(x_c),$ where $\oplus$ denotes bitwise
XOR -- an extremely fast operation. With character tables in cache,
this scheme is the fastest known 3-independent hashing
scheme \cite{patrascu12charhash}. We will denote by $u=\abs U$ the size of the key domain, identify $U=\Sigma^c$ with $[u]$, and always assume the size of the alphabet, $\abs \Sigma$, to be a power of two. For instance, we could consider 32-bit keys consisting of four 8-bit characters. For a given computer, the best choice of $c$ in terms of speed is easily determined experimentally once and for all, and is independent of the problems considered.

Let $S\subseteq U$ and consider hashing $n=\abs S$ weighted balls or keys into $m=2^\ell$ bins using a simple tabulation function, $h\colon [u]\to [m]$, in line with the first case mentioned above. We shall prove the theorem below.

\begin{theorem}\label{thm:intro-simple-tab}
    Let $h\colon [u]\to [m]$ be a simple tabulation hash function
with $[u]=\Sigma^c$, $c=O(1)$. Let $S\subseteq [u]$ be given of size $n=\abs{S}$ and assign to each key/ball $x\in S$ a weight $w_x \in [0,1]$. Let $y\in [m]$, and define $X=\sum_{x\in S}w_x\cdot [h(x) = y]$ to be the total weight of the balls hashing to bin $y$. Then for any constant $\gamma>0$, $X$ is strongly concentrated with \errorterm $n/m^\gamma$, where the constants of the asymptotics are determined solely by $c$ and $\gamma$. Furthermore, this concentration is query invariant.
\end{theorem}

In \cref{thm:intro-simple-tab}, we note that the expectation, $\mu = \E{X}$, and the variance, $\sigma^2 = \Var{X}$, are the same as if $h$ were a fully random hash function since $h$ is 3-independent. This is true even when conditioning on the hash value of a query key having a specific value.
The bound provided by~\Cref{thm:intro-simple-tab} is therefore the same as the variance based Chernoff
bound \req{eq:var-chernoff} except for a constant delay in the exponential decrease and an
\errorterm of $n/m^\gamma$. Since
$\sigma^2\leq \mu$,~\Cref{thm:intro-simple-tab} also implies the classic
one-sided Chernoff bounds \req{eq:classic-chernoff+}
and \req{eq:classic-chernoff-}, again with the constant delay and the \errorterm
as above, and a leading factor of 2.

P{\v a}tra{\c s}cu and Thorup \cite{patrascu12charhash} proved
an equivalent probability bound, but without weights, and, more importantly, 
with the restriction that the number of bins $m\geq n^{1-1/(2c)}$. In particular, this implies the restriction $\mu\leq |\Sigma|^{1/2}$. Our new
bound gives Chernoff-style concentration with high probability in $n$ for any $m\geq n^{\eps}$, $\eps=\Omega(1)$. Indeed, letting 
$\gamma'=(\gamma+1)/\eps$, the \errorterm becomes $n/m^{\gamma'}\leq 1/n^{\gamma}$. 

However, for small $m$ the error probability $n/m^{\gamma}$ is
prohibitive. For instance, unbiased coin tossing, corresponding to the
case $m=2$, has an \errorterm of $n/2^\gamma$ which is
useless.  In \cref{sec:few-bin-problem}, we will show that it is inherently impossible to get good
concentration bounds using simple tabulation hashing when the number of bins $m$ is
small. To handle all instances, including those with few bins, and to
support much more general Chernoff bounds, we introduce a new hash
function: tabulation-permutation hashing.

\subsection{Tabulation-Permutation Hashing}
We start by defining \emph{tabulation-permutation hashing} from
$\Sigma^c$ to $\Sigma^d$ with $c, d=O(1)$. 
A tabulation-permutation hash function $h\colon \Sigma^c\to \Sigma^d$
is given as a composition, $h=\tau\circ g$, of a simple tabulation
hash function $g\colon \Sigma^c\to \Sigma^d$ and a permutation
$\tau\colon \Sigma^d\to \Sigma^d$. The permutation is a coordinate-wise fully random
permutation: for each $j\in \{1,\ldots,d\}$, pick a uniformly random character
permutation $\tau_j: \Sigma \to \Sigma$. Now, $\tau =
(\tau_1, \dots, \tau_d)$ in the sense that for $z=(z_1, \dots,
z_d)\in \Sigma^d$, $\tau(z)
= \left(\tau_1\left(z_1\right), \dots, \tau_d\left(z_d\right)\right)$. In words, a tabulation-permutation hash function hashes $c$ characters to $d$ characters using simple tabulation, and then randomly permutes each of the $d$ output characters. As is, tabulation-permutation hash functions yield values in $\Sigma^d$, but we will soon see how we can hash to $[m]$ for any $m\in \N$.

If we precompute tables
$T_i\colon \Sigma\to \Sigma^d$, where
\[T_i(z_i) = \left(\overbrace{0, \dots, 0}^{i-1}, \tau_i(z_i), \overbrace{0, \dots, 0}^{d-i}\right), \quad z_i\in \Sigma\textnormal,\] 
then $\tau(z_1, \dots, z_d) = T_1(z_1)\oplus\cdots \oplus
T_d(z_d)$. Thus, $\tau$ admits the same implementation as simple
tabulation, but with a special  distribution on the character tables. If in particular $d\leq c$, the permutation step can be executed at least as fast as
the simple tabulation step.

Our main result is that with tabulation-permutation
hashing, we get high probability Chernoff-style bounds for the third and most general case described in the beginning of the introduction.
\begin{theorem}\label{thm:intro-tab-perm}
     Let $h\colon [u]\to [r]$ be a tabulation-permutation hash
     function with $[u]=\Sigma^c$ and $[r]=\Sigma^d$, $c,d=O(1)$. Let 
      $v\colon [u]\times [r]\to [0, 1]$ be a fixed value
     function that to each key $x\in [u]$ assigns a value $X_x=v(x,
     h(x))\in[0,1]$ depending on the hash value $h(x)$ and define $X=\sum_{x\in
     [u]} X_x$. For any constant $\gamma>0$, $X$ is strongly concentrated with \errorterm $1/u^\gamma$, where the constants of the asymptotics are determined solely by $c$, $d$, and $\gamma$. Furthermore, this concentration is query invariant. %\todo{Why are codomains of value functions $[0, 1]$ in intro/techniques and $[-1, 1]$ in paper?}
\end{theorem}

Tabulation-permutation hashing inherits the 3-independence
of simple tabulation, so as in~\Cref{thm:intro-simple-tab},
 $\mu = \E{X}$ and $\sigma^2 = \Var{X}$ have
exactly the same values as if $h$ were a fully-random hash function. Again, this is true even
when conditioning on the hash value of a query key having a specific value.

Tabulation-permutation hashing allows us to hash into $m$ bins for any $m\in \N$ (not necessarily a power of two) preserving the strong concentration from~\cref{thm:intro-tab-perm}. To do so, simply define the hash function $h^m\colon [u]\to [m]$ by $h^m(x)=h(x)\bmod m$. 
Relating back to~\Cref{thm:intro-simple-tab}, consider a set
$S\subseteq U$ of $n$ balls where each ball $x\in S$ has a weight $w_x\in[0,1]$ and balls $x$ outside $S$ are defined to have weight $w_x=0$.
To measure the total weight of the balls landing in a given bin $y \in [m]$, we define the value
function $v(x,z)=w_x \cdot[z \bmod m=y]$. Then
\[X= \sum_{x\in [u]} v(x,h(x))=\sum_{x\in S} w_x \cdot [h^m(x)=y]\]
is exactly the desired quantity and we get the concentration bound from~\Cref{thm:intro-tab-perm}.
Then the big advantage of tabulation-permutation hashing over simple tabulation hashing is that it reduces the \errorterm from
$n/m^\gamma$ of~\Cref{thm:intro-simple-tab} to the $1/u^\gamma$ of~\Cref{thm:intro-tab-perm}, where $u$ is the 
size of the key universe. Thus, with tabulation-permutation hashing, we actually get Chernoff bounds with 
high probability regardless of the number of bins.

P{\v a}tra{\c s}cu and Thorup \cite{PT13:twist} introduced twisted
tabulation that like our tabulation-permutation achieved Chernoff-style
concentration bounds with a general value function $v$. Their bounds are
equivalent to those of Theorem~\ref{thm:intro-tab-perm}, but only under the 
restriction $\mu\leq \abs\Sigma^{1-\Omega(1)}$.  To understand
how serious this restriction is, consider again tossing an unbiased coin for each key $x$ in a set $S\subseteq [u]$, corresponding to the case $m=2$ and $\mu=|S|/2$.
With the restriction from  \cite{PT13:twist}, we can only
handle $|S|\leq 2\abs\Sigma^{1-\Omega(1)}$, but recall that $\Sigma$ is chosen
small enough for character tables to fit in fast cache, so this
rules out any moderately large data set. We are going to show
that for certain sets $S$, twisted tabulation has the same problems as simple tabulation when hashing to few bins. This implies that the restrictions
from  \cite{PT13:twist} cannot be lifted with a better analysis.

P{\v a}tra{\c s}cu and Thorup \cite{PT13:twist} were acutely aware
of how prohibitive the restriction $\mu\leq \abs\Sigma^{1-\Omega(1)}$ is. For
unbounded $\mu$, they proved a weaker bound; namely that with twisted tabulation
hashing, $X=\mu\pm O(\sigma (\log u)^{c+1})$ 
with probability $1-u^{-\gamma}$ for any $\gamma=O(1)$. 
With our tabulation-permutation hashing, we
get $X=\mu\pm O(\sigma(\log u)^{1/2})$ with the same high probability, $1-u^{-\gamma}$.
Within a constant factor on the deviation, our high probability bound is as
good as with fully-random hashing. 

More related work, including Siegel's~\cite{siegel04hash} and Thorup's~\cite{Tho13:simple-simple} highly independent hashing will be discussed in~\Cref{sec:relatedwork}. 

\subsection{Tabulation-1Permutation}
Above we introduced tabulation-permutation hashing which yields Chernoff-style
bounds with an arbitrary value function. This is the same general
scenario as was studied for twisted tabulation in \cite{PT13:twist}. However, for almost all applications we are aware of, we only need the generality of the second case presented at the beginning of the introduction. Recall that in this case we are only interested in the total weight of the balls hashing to a certain interval. As it turns out, a significant simplification of tabulation-permutation hashing suffices to achieve strong concentration bounds.  We call this simplification \emph{tabulation-1permutation}. Tabulation-permutation hashing randomly permutes each of the $d$ output characters of a simple tabulation function $g\colon \Sigma^c\to\Sigma^d$. Instead, tabulation-1permutation only permutes the most significant character. 
 
More precisely,
a tabulation-1permutation hash function $h\colon \Sigma^c\to \Sigma^d$ is a composition, $h=\tau\circ g$, of a simple tabulation function, $g\colon \Sigma^c\to\Sigma^d$, and a random permutation, $\tau\colon \Sigma^d\to\Sigma^d$, of the most significant character, $\tau(z_1, \dots, z_d) = (\tau_1(z_1), z_2, \dots, z_d)$ for a random character permutation $\tau_1\colon \Sigma\to\Sigma$.

To simplify the implementation of the hash function and speed up its evaluation, we can
precompute a table
$T\colon \Sigma\to \Sigma^d$ such that for $z_1\in \Sigma$,
\[T(z_1) = \left(z_1\xor \tau_1(z_1), \overbrace{0, \dots, 0}^{d-1}\right).\]
Then if $g(x) = z = (z_1, \dots, z_d)$, $h(x)= z\xor T(z_1)$.

This simplified scheme, needing only $c+1$ character lookups, is
powerful enough for concentration within an arbitrary interval. %Because we use a general reduction by P{\v a}tra{\c s}cu and Thorup~\cite{patrascu12charhash}, our bound is based on the expectation as in~\eqref{eq:classic-chernoff+} and~\eqref{eq:classic-chernoff-}, which is still powerful enough for most applications % including those we will discuss for streaming in~\Cref{sec:streaming}.

\begin{theorem}\label{thm:intro-tab-1perm}
     Let $h\colon [u]\to [r]$ be a tabulation-1permutation hash
     function with $[u]=\Sigma^c$ and $[r]=\Sigma^d$, $c,d=O(1)$. Consider
a key/ball set $S\subseteq [u]$ of size $n=\abs{S}$ where each ball $x\in S$ 
is assigned a weight $w_x \in [0,1]$.  Choose arbitrary hash values $y_1, y_2\in [r]$ with
$y_1\leq y_2$.  Define $X=\sum_{x\in
S}w_x\cdot [y_1\leq h(x)< y_2]$ to be the total weight of balls hashing to
the
interval $[y_1, y_2)$. Then for any constant
$\gamma>0$, $X$ is strongly concentrated with \errorterm $1/u^\gamma$, where the constants of the asymptotics are determined solely by $c$, $d$, and $\gamma$. Furthermore, this concentration is query invariant. 
\end{theorem}
One application of Theorem~\ref{thm:intro-tab-1perm} is in
the following sampling scenario: We set $y_1=0$, and sample all keys
with $h(x)<y_2$. Each key is then sampled with probability $y_2/r$,
and~\Cref{thm:intro-tab-1perm} gives concentration on the number of samples.
In \cite{NoRepetitions} this is used for more efficient implementations of streaming algorithms.

Another application is efficiently hashing into
an arbitrary number $m\leq r$ of bins. We previously discussed using
hash values modulo $m$, but a general mod-operation is often quite slow.
Instead we can think of hash values as fractions $h(x)/r\in [0,1)$. Multiplying
by $m$, we get a value in $[0,m)$, and the bin index is then obtained by rounding down to the nearest integer. This implementation is very efficient because $r$ is
a power of two, $r=2^b$, so the rounding is obtained by a right-shift by $b$ bits. To hash a key
$x$ to $[m]$, we simply compute $h^m(x)=(h(x)*m)\,\texttt{>\/>} \, b$. 
Then $x$ hashes to bin $d\in[m]$ if and only if 
$d\in[y_1,y_2)\subseteq[r]$ where $y_1=\lfloor rd/m\rfloor$ and $y_2=\lfloor r(d+1)/m\rfloor$, so the number of
keys hashing to a bin is concentrated as in Theorem~\ref{thm:intro-tab-1perm}. Moreover,
$h^m$ uses only $c+1$ character lookups and a single
multiplication in addition to some very fast shifts and bit-wise Boolean operations.

\subsection{Subpolynomial Error Probabilities}
In Theorem~\ref{thm:intro-tab-perm} and~\ref{thm:intro-tab-1perm}, we have $\Pr[|X-\mu|\geq t] =O(\exp(-\Omega(\sigma^2\cC(t/\sigma^2))))+1/u^\gamma$
which holds for any fixed $\gamma$. The value of $\gamma$ affects the
constant hidden in the $\Omega$-notation delaying the exponential decrease. In~\Cref{sec:few-bin-problem}, we will show that the
same bound does not hold if $\gamma$ is replaced by any slow-growing
but unbounded function. Nevertheless, it follows from our analysis that for  every $\alpha(u) = \omega(1)$ there exists $\beta(u)= \omega(1)$ such that whenever $\exp(-\sigma^2\cC(t/\sigma^2))<1/u^{\alpha(u)}$, $\Pr[|X-\mu|\geq
t]\leq 1/u^{\beta(u)}$.

\subsection{Generic Remarks on Universe Reduction and Amount of Randomness}\label{sec:unired} The following observations are fairly standard in the literature. 
Suppose we wish to hash a set of keys $S$ belonging to some universe $\cU$. The universe may be so large compared to $S$ that it is not efficient to directly implement a theoretically powerful hashing scheme like tabulation-permutation hashing. A standard first step is to perform a {\em universe reduction}, mapping
$\cU$ randomly to ``signatures'' in $[u]=\{0,1,\ldots,u-1\}$, where
$u=n^{O(1)}$, e.g.\ $u=n^3$, so that no two keys from $S$ are expected
to get the same signature \cite{carter77universal}. As the only theoretical property required for the universe reduction is a low collision probability, this step can be implemented using very simple hash functions as described in~\cite{Thorup2015HighSH}. In this paper, we
generally assume that this universe reduction has already been done, if
needed, hence that we only need to deal with keys from a
universe $[u]$ of size polynomial in $n$. For any small constant $\eps>0$ we may thus pick $c=O(1/\eps)$ such that the space used for our
hash tables, $\Theta(\abs\Sigma)$, is $O(n^\eps)$. 
Practically speaking, this justifies focusing on the hashing of $32$- and $64$-bit keys. 

When we defined simple tabulation above, we said the character tables 
were fully random. However, for the all the bounds in this paper,
it would suffice if they were populated with a $O(\log
u)$-independent pseudo-random number generator (PNG), so we only need
a seed of $O(\log u)$ random words to be shared among all applications
who want to use the same simple tabulation hash function.  Then, as a
preprocesing for fast hashing, each application can locally fill
the character tables in $O(|\Sigma|)$ time
\cite{christiani14prg}. Likewise, for our tabulation permutation
hashing, our bounds only require a
$O(\log u)$-independent PNG to generate the permutations. The
basic point here is that tabulation based hashing does
not need a lot of randomness to fill the tables, but only space
to store the tables as needed for the fast computation of hash values.

\subsection{Techniques}

\label{sec:introduction_techniques}The paper relies on three main technical insights to establish the concentration inequality for tabulation-permutation hashing of \cref{thm:intro-tab-perm}. We shall here describe each of these ideas and argue that each is in fact necessary towards an efficient hash function with strong concentration bounds.

\subsubsection{Improved Analysis of Simple Tabulation}\label{sec:techsimple}
The first step towards proving \cref{thm:intro-tab-perm} is to better understand the distribution of simple tabulation hashing. We describe below how an extensive combinatorial analysis makes it possible to prove a  generalised version of \cref{thm:intro-simple-tab}.

To describe the main idea of this technical contribution, we must first introduce some ideas from previous work in the area. This will also serve to highlight the inherent limitations of previous approaches. A simplified account is the following. Let $h\colon \Sigma^c\to[m]$ be a simple tabulation hash function, let $y\in [m]$ be given, and for some subset of keys $S\subseteq \Sigma^c$, let $X=\sum_{x\in S}[h(x)=y]$ be the random variable denoting the number of elements $x\in S$ that have hash value $h(x)=y$. Our goal is to bound the deviation of $X$ from its mean $\mu=\abs S/m$. 
We first note that picking a random simple tabulation hash function $h: \Sigma^c \to[m]$ amounts to filling the $c$ character tables, each of size $\Sigma$, with uniformly random hash values. Thus, picking a simple tabulation hash function $h:\Sigma^c \to [m]$ corresponds to picking a uniformly random hash function $h\colon [c] \times \Sigma \to [m]$. We call $[c] \times \Sigma$ the set of \emph{position characters}. Viewing a key $x=(x_1,\dots,x_c)\in \Sigma^c$ as a set of position characters, $x=\{(1,x_1),\dots,(c,x_c)\}$, and slightly abusing notation, it then holds that $h(x)=\bigoplus_{\alpha \in x}h(\alpha)$.
Now let $\alpha_1, \dots, \alpha_{r}$ be a (for the sake of the proof) well-chosen ordering of the position characters. For each $k \in [r+1]$, we define the random variable $X_k=\EC{X}{h(\alpha_1), \dots, h(\alpha_k)}$, where $h(\alpha_i)$ is the value of the entry of the lookup table of $h$ corresponding to $\alpha_i$. The process $(X_k)_{k=0}^{r}$ is then a martingale. We can view this as revealing the lookup table of $h$ one entry at a time and adjusting our expectation of the outcome of $X$ accordingly. Defining the martingale difference $Y_k = X_k-X_{k-1}$, we can express $X$ as a sum $X= \mu + \sum_{k=1}^{c\cdot \abs \Sigma}Y_k$. Previous work has then bounded the sum using a Chernoff inequality for martingales as follows. Due to the nature of the ordering of $\{\alpha_i\}_{i=1}^r$, we can find $M>0$ such that with high probability, $\abs {Y_k}\leq M$ for every $k$. Then conditioned on each of the $Y_k$s being bounded, $X$ satisfies the Chernoff bounds of \eqref{eq:classic-chernoff+} and \eqref{eq:classic-chernoff-} except the exponent is divided by $M$. As long as the expectation, $\mu$, satisfies $\mu = O( \abs\Sigma)$, it is possible\footnote{In \cite{patrascu12charhash}, the actual analysis of simple tabulation using this approach achieves $\mu=O(\sqrt{\abs\Sigma})$.} that $M=O(1)$, yielding Chernoff bounds with a constant in the delay of the exponential decrease.
 However, since there are only $c\cdot \abs \Sigma$ variables, $Y_k$, it is clear that $M\geq \mu/(c\cdot \abs\Sigma)$. Thus, whenever $\mu = \omega(\abs\Sigma)$, the delay of the exponential decrease is super-constant, meaning that we do not get asymptotically tight Chernoff-style bounds. This obstacle has been an inherent issue with the previous techniques in analysing both simple tabulation \cite{patrascu12charhash} as well as twisted tabulation \cite{PT13:twist}. Being unable to bound anything beyond the absolute deviation of each variable $Y_k$, it is impossible to get good concentration bounds for large expectations, $\mu$.

Going beyond the above limitation, we dispense with the idea of bounding absolute deviations and instead  
 bound the sum of variances, $\sigma^2=\sum_{k=1}^{c\cdot \abs \Sigma}\Var {Y_k}$. This sum has a combinatorial interpretation relating to the number of collisions of hash keys, i.e., the number of pairs $y_1, y_2\in \Sigma^c$ with $h(y_1)=h(y_2)$. 

An extensive combinatorial analysis of simple tabulation hashing yields high-probability bounds on the sum of variances that is tight up to constant factors. This is key in establishing an induction that allows us to prove \cref{thm:intro-simple-tab}. 
Complementing our improved bounds, we will show that simple tabulation hashing inherently does not support Chernoff-style concentration bounds for small $m$.

\subsubsection{Permuting the Hash Range}\label{sec:techperm}
 Our next step is to consider the hash function $h = \tau \circ g\colon \Sigma^c\to \Sigma$ where $g\colon \Sigma^c\to \Sigma$ is a simple tabulation hash function and $\tau\colon \Sigma \to \Sigma$ is a uniformly random permutation. Our goal is to show that $h$ provides good concentration bounds for any possible value function. To showcase our approach, we consider the example of hashing to some small set, $[m]$, of bins, e.g., with $m=2$ as in our coin tossing example. This can be done using the hash function $h^m \colon \Sigma^c \to [m]$ defined by $h^m(x)=(h(x) \mod{m})$. For simplicity we assume that $m$ is a power of two, or equivalently, that $m$ divides $|\Sigma|$. We note that the case of small $m$ was exactly the case that could not be handled with simple tabulation hashing alone.

Let us look at the individual steps of $h^m$. First, we use simple tabulation mapping into the ``character bins'', $\Sigma$. The number of balls in any given character bin is nicely concentrated, but only because $|\Sigma|$ is large. Next, perform a permutation followed by the mod $m$ operation. The last two steps correspond to the way we would deal a deck of $|\Sigma|$ cards into $m$ hands. The cards are shuffled by a random permutation, then dealt to the $m$ players one card at a time in cyclic order.
The end result is that each of the final $m$ bins is assigned exactly $|\Sigma|/m$ random character bins. An important point is now that because the partitioning is \emph{exact}, the error in the number of balls in a final bin stems solely from the errors in the $|\Sigma|/m$ character bins, and because the partitioning is \emph{random}, we expect the positive and negative errors to cancel out nicely. 
The analysis, which is far from trivial, requires much more than these properties. For example, we also need the bound described in~\cref{sec:techsimple} on the sum of variances. This bound ensures that not only is the number of balls in the individual character bins nicely concentrated around the mean, but moreover, there is only a small number of character bins for which the error is large.
That these things combine to yield strong concentration, not only in the  specific example above, but for general value functions as in \cref{thm:intro-tab-perm}, is quite magical.
\\ \\
We finish the discussion by mentioning two approaches that do not work and highlight how a permutation solves the issues of these strategies.

First, one may ask why we need the permutation at all. After all, the mod $m$ operation also partitions the $|\Sigma|$ character bins into groups of the same size, $|\Sigma|/m$. The issue is that while a simple tabulation hash function, $g:\Sigma^c \to \Sigma$, has good concentration in each of the individual character bins, the $|\Sigma|/m$ character bins being picked out by the mod $m$ operation constitute a very structured subset of $\Sigma$, and the errors from this set of bins could be highly correlated. We indeed show that the structure of simple tabulation causes this to happen for certain sets of keys, both theoretically (\cref{sec:few-bin-problem}) and experimentally (\cref{sec:experiments}).

Second, the reader may wonder why we use a permutation, $\tau \colon \Sigma \to \Sigma$, instead of a random hash function 
as in double tabulation \cite{Tho13:simple-simple}. In terms of the card dealing analogy, this would correspond to throwing the $|\Sigma|$ cards at the $m$ astonished card players one at a time with a random choice for each card, not guaranteeing that the players each get the same number of cards. And this is exactly the issue. Using a fully random hash function $\tau'$, we incur an extra error stemming from $\tau'$ distributing the $|\Sigma|$ character bins unevenly into the final bins. This is manifested in the variance of the number of balls hashing to a specific bin: Take again the coin tossing example with $n\geq |\Sigma|$ balls being distributed into $m=2$ bins. With a permutation $\tau$ the hash function becomes $2$-independent, so the variance is the same as in the fully random setting, $n/4$. Now even if the simple tabulation hash function, $g$, distributes the $n$ keys into the character bins evenly, with exactly $n/\Sigma$ keys in each, with a fully random hash function, $\tau'$, the variance becomes $(n/|\Sigma|)^2\cdot|\Sigma|/4=n^2/(4|\Sigma|)$, a factor of $n/|\Sigma|$ higher.

\subsubsection{Squaring the Hash Range}\label{sec:techsquaring}
The last piece of the puzzle is a trick to extend the range of a hash function satisfying Chernoff-style bounds. We wish to construct a hash function $h\colon \Sigma^c\to [m]$ satisfying Chernoff-style bounds for $m$ arbitrarily large as in \cref{thm:intro-tab-perm}. At first sight, the trick of the previous subsection would appear to suffice for the purpose. However, if we let $g=\tau\circ h $ be the composition of a simple tabulation hash function $h\colon \Sigma^c\to [m]$ and $\tau$ a random permutation of $[m]$, we run into trouble if for instance $[m]=\Sigma^c$. In this case, a random permutation of $[m]$ would require space equal to that of a fully random function $f\colon \Sigma^c\to [m]$, but the whole point of hashing is to use less space. Hence, we instead prove the following. Let $a\colon C\to D$ and $b\colon C\to D$ be two independent hash functions satisfying Chernoff-style bounds for general value functions. Then this property is preserved up to constant factors under ``concatenation'', i.e., if we let $c\colon C\to D^2$ be given by $c(x)=(a(x),b(x))$, then $c$ is also a hash function satisfying Chernoff-style bounds for general value functions, albeit with a slightly worse constant delay in the exponential decrease than $a$ and $b$. Thus, this technique allows us to ``square'' the range of a hash function.

With this at hand, let $h_1, h_2\colon \Sigma^c\to \Sigma$ be defined as $h_1=\tau_1\circ g_1$ and $h_2 = \tau_2\circ g_2$, where $g_1, g_2\colon \Sigma^c\to\Sigma$ are simple tabulation hash functions and $\tau_1, \tau_2\colon \Sigma\to\Sigma$ are random permutations. Then the concatenation $h\colon \Sigma^c\to\Sigma^2$ of $h_1$ and $h_2$ can be considered a composition of a simple tabulation function $g\colon \Sigma^c\to \Sigma^2$ given by $g(x)=(g_1(x), g_2(x))$ and a coordinate-wise permutation $\tau=(\tau_1, \tau_2)\colon \Sigma^2\to\Sigma^2$, where the latter is given by $\tau(x_1, x_2) = (\tau_1(x_1), \tau_2(x_2)), x_1, x_2\in \Sigma$. Applying our composition result, gives that $g$ also satisfies Chernoff-style bounds. Repeating this procedure $\lceil\log(d) \rceil=O(1)$ times, yields the desired concentration bound for tabulation-permutation hashing $h\colon\Sigma^c\to\Sigma^d$ described in \cref{thm:intro-tab-perm}.

\subsection{Related Work -- Theoretical and Experimental Comparisons}\label{sec:relatedwork}
In this section, we shall compare the performance of tabulation-permutation and tabulation-1permutation hashing with other related results. Our comparisons are both theoretical and empirical.
Our goal in this paper is fast constant-time hashing having strong concentration 
bounds with high probability, i.e., bounds of the form
$$
\Pr[|X-\mu|\geq t] \le 2\exp(-\Omega(\sigma^2\cC(t/\sigma^2)))+u^{-\gamma},
$$
%Here $X=\sum_{x \in S}X_x$ is a sum of random variables where $X_x$ depends only on the hash value $h(x)$ of the key $x$, $\mu=\E{X}$, and $\sigma^2=\Var{X}$.  
as in~\Cref{def:strongly-concentrated,thm:intro-tab-perm,thm:intro-tab-1perm}, or possibly with $\sigma^2$ replaced by $\mu\geq \sigma^2$.
Theoretically, we will only compare with other hashing schemes that are relevant to this goal.
In doing so, we distinguish between the hash functions that achieve Chernoff-style bounds with restrictions on the expected value and those that, like our new hash functions, do so without such restrictions, which is what we want for all possible input.
Empirically, we shall compare the practical evaluation time of tabulation-permutation and permutation-1permutation to the fastest commonly used hash functions and to hash functions with similar theoretical guarantees.
 A major goal of algorithmic analysis is to understand the theoretical behavior of simple algorithms that work well in practice, providing them with good theoretical guarantees such as worst-case behavior. For instance, one may recall Knuth's analysis of linear probing \cite{knuth63linprobe}, showing that this very practical algorithm has strong theoretical guarantees. In a similar vein, we not only show that the hashing schemes of tabulation-permutation and tabulation-1permutation have strong theoretical guarantees, we also perform experiments, summarized in \Cref{tab:speed}, demonstrating that in practice they are comparable in speed to some of the most efficient hash functions in use, none of which have similar concentration guarantees.  Thus, with our new hash functions, hashing with strong theoretical concentration guarantees is suddenly feasible for time-critical applications. 

 \begin{table}
\begin{center}
    \begin{tabular}{|l | r | r | r | r |}
        \cline{2-5}
        \multicolumn{1}{l|}{} & \multicolumn{4}{|c|}{Running time (ms)} \\
        \cline{2-5}
        \multicolumn{1}{l|}{} \multirow{2}{*}{}&\multicolumn{2}{|c|}{Computer 1} & \multicolumn{2}{|c|}{Computer 2} \\
        \hline Hash function                  & 32 bits & 64 bits & 32 bits & 64 bits \\
        \hline
        \emph{Multiply-Shift}                 & 4.2     &  7.5   & 23.0     & 36.5    \\
        \emph{2-Independent PolyHash}         & 14.8    & 20.0   & 72.2     & 107.3   \\
        \emph{Simple Tabulation}              & 13.7    & 17.8   & 53.1     & 55.9    \\
        \emph{Twisted Tabulation}             & 17.2    & 26.1   & 65.6     & 92.5    \\ 
        \emph{Mixed Tabulation}               & 28.6    & 68.1   & 120.1    & 236.6   \\
        \hline
        \textbf{Tabulation-1Permutation}      & 16.0    & 19.3   & 63.8     & 67.7    \\
        \textbf{Tabulation-Permutation}       & 27.3    & 43.2   & 118.1    & 123.6   \\ \hline
        Double Tabulation                     & 1130.1  & --     & 3704.1   & --      \\
        ``Random'' (100-Independent PolyHash) & 2436.9  & 3356.8 & 7416.8   & 11352.6 \\
        \hline
    \end{tabular}
    \caption{The time for different hash functions to hash $10^7$ keys of  
    length 32 bits and 64 bits, respectively, to ranges of size 32 bits and 64 bits. The experiment was carried out on two computers. The hash functions written in italics are those without general Chernoff-style bounds. Hash functions written in bold are the contributions of this paper. The hash functions in regular font are known to provide Chernoff-style bounds. Note that we were unable to implement double tabulation from 64 bits to 64 bits since the hash tables were too large to fit in memory. \label{tab:speed}}
\end{center}
\end{table}

 \begin{table}
\begin{center}
    \begin{tabular}{|l | l | l | l | l |}
        \hline Hash function                  & Time & Space & Concentration Guarantee & Restriction \\
        \hline
        {Multiply-Shift}                 & $O(1)$   &  $O(1)$   & Chebyshev's inequality & None \\
        {$k$-Independent PolyHash}         & $O(k)$ & $O(k)$   & Chernoff-style bounds & 
        \begin{tabular}{ l }Requires $k=\Omega(\log u)$ for\\ \errorterm $O(1/u^\gamma)$ 
        \end{tabular}\\
        Simple Tabulation & $O(c)$ & $O(u^{1/c})$ & Chernoff-style bounds & \Errorterm: $O(n/m^\gamma)$\\
        Twisted Tabulation & $O(c)$ & $O(u^{1/c})$ & Chernoff-style bounds & Requires: $\mu\leq \abs{\Sigma}^{1-\Omega(1)}$\\
        Mixed Tabulation &$O(c)$& $O(u^{1/c})$ & Chernoff-style bounds & Requires: $\mu = o(\abs\Sigma)$ \\
        \textbf{Tabulation-Permutation}       & $O(c)$ & $O(u^{1/c})$ & Chernoff-style bounds & \Errorterm: $O(1/u^\gamma)$ \\ 
        Double Tabulation                     & $O(c^2)$  & $O(u^{1/c})$ & Chernoff-style bounds & \Errorterm: $O(1/u^\gamma)$     \\        \hline
    \end{tabular}
    \caption{Theoretical time and space consumption of some of the hash functions discussed.   \label{tab:theoretical_performance}}
\end{center}
\end{table}

\subsubsection{High Independence and Tabulation}\label{sec:independence-and-tab}
Before this paper, the only known way to obtain unrestricted Chernoff-style concentration bounds with hash functions that can be evaluated in constant time was through $k$-independent hashing. Recall that a hash function $h:U \to R$ is $k$-independent if the distribution of $(h(x_1),\dots,h(x_k))$ is uniform in $R^k$ for every choice of distinct keys $x_1,\dots,x_k\in U$.
Schmidt, Siegel, and
Srinivasan \cite{schmidt95chernoff} have shown that with $k$-independent hashing, we have Chernoff-style concentration bounds in all three cases mentioned at the beginning of the introduction up to an \errorterm decreasing exponentially in
$k$. With $k=\Theta(\gamma\log u)$, this means Chernoff-style concentration with an
\errorterm of $1/u^\gamma$ like in Theorem~\ref{thm:intro-tab-perm} and~\ref{thm:intro-tab-1perm}. 
However, evaluating any $k$-independent hash function takes time
$\Omega(k)$ unless we use a lot of space.
Indeed, a cell probe lower bound by
Siegel~\cite{siegel04hash} states that evaluating a $k$-independent
hash function over a key domain $[u]$ using $t<k$ probes,  requires us to use at least
$u^{1/t}$ cells to represent the hash function.  
Thus, aiming for Chernoff concentration through $k$-independence with $k=\Omega(\log u)$ and with constant evaluation time, we would have to use $u^{\Omega(1)}$ space like our tabulation-permutation. 
Here it should be mentioned that $k$-independent PolyHash modulo a prime $p$
can be evaluated at $k$ points in total time $O(k \log^2 k)$ using
multipoint evaluation methods. Then the average evaluation time is $O(\log^2 k)$, but it requires that the hashing can be done to batches of $k$ keys at
a time. We can no longer hash one key at a time, continuing
with other code before we hash the next key. This
could be a problem for some applications. A bigger practical issue is
that it is no longer a black box implementation of a hash function. To
understand the issue, think of Google's codebase where thousands of
programs are making library calls to hash functions. A change to
multipoint evaluation would require rewriting all of the calling
programs, checking in each case that batch hashing suffices ---
a huge task that likely would create many errors.
A final point is that multipoint evaluation is complicated to implement yet still not as fast as our tabulation-permutation hashing. 
Turning to upper bounds, Siegel designed a $u^{\Omega(1/c^2)}$-independent hash function that can be represented in tables of size $u^{1/c}$ and evaluated in $c^{O(c)}$ time. With $c=O(1)$, this suffices for Chernoff-style concentration bounds by the argument above. However, as Siegel states, the hashing scheme is ``far too slow for any practical
application''.

In the setting of Siegel, Thorup's double
tabulation \cite{Tho13:simple-simple} is a simpler and more
efficient construction of highly independent hashing. It is the main constant-time competitor of our new tabulation-permutation hashing, and yet it is 30 times slower in our experiments. In the following, we describe the theoretical guarantees of double tabulation hashing and discuss its concrete parameters in terms of speed and use of space towards comparing it with tabulation-permutation hashing.

A \emph{double tabulation} hash function, $h:\Sigma^c\to\Sigma^c$ is the composition of two independent simple tabulation
hash functions $h_1:\Sigma^c\fct\Sigma^d$ and $h_2:\Sigma^d\fct\Sigma^c$, $h=h_2\circ h_1$. Evaluating the function thus requires $c+d$ character lookups. Assuming that each memory unit stores an element
from $[u]=\Sigma^c$ and $d\geq c$, the space used for the character
tables is $(c(d/c)+d)u^{1/c}=2d
u^{1/c}$.
Thorup \cite{Tho13:simple-simple} has shown that if $d\geq
6c$, then with probability $1-o(\Sigma^{2-d/(2c)})$ over the choice of
$h_1$, the double tabulation hash function $h$ is $k$-independent for
$k=|\Sigma|^{1/(5c)}=u^{\Omega(1/c^2)}$. More precisely, with this probability, the output keys $(h_1(x))_{x \in \Sigma^c}$ are distinct, and $h_2$ is $k$-independent when restricted to this set of keys. If we are lucky to pick such an $h_1$,
this means that we get the same high indepence as
Siegel~\cite{siegel04hash}. With $d=6c$, the space used is
$12cu^{1/c}=O(c u^{1/c})$ and the number of character lookups to
compute a hash value is $7c=O(c)$. Tabulation-permutation hashing is very comparable to Thorup's double tabulation. As previously noted, it can be implemented in the same way, except that we fill the character tables of $h_2$ with permutations and padded zeros instead of random hash values. 
To compare, a tabulation-permutation hash function $h\colon \Sigma^c\to\Sigma^c$ requires $2c$ lookups and uses space $2cu^{1/c}$, which may not seem a big difference. 
However, in the following, we demonstrate how restrictions on double tabulation cost an order of magnitude in speed and space compared with tabulation-permutation hashing when used with any realistic parameters.

With Thorup's double tabulation, for $(\log u)$-independence, we need $\log u\leq |\Sigma|^{1/(5c)}=u^{1/(5c^2)}$. 
In choosing values for $u$ and $c$ that work in practice, this inequality is very restrictive. Indeed, even for $c=2$, $\log u \leq u^{1/20}$, which roughly implies that $\log u \geq 140$. Combined with the fact that the character tables use space $12c|\Sigma|$, and that $|\Sigma|\geq (\log u)^{5c}$, this is an intimidating amount of space.
Another problem is the error probability over $h_1$ of $1-o(\Sigma^{2-d/(2c)})$. If we want this to be $O(1/u)$, like in the error bounds from Theorem~\ref{thm:intro-tab-perm} and~\ref{thm:intro-tab-1perm}, we need $d\geq 2(c^2+2c)$.  
Thus, while things work well asymptotically, these constraints make it hard to implement highly independent double tabulation on any normal computer. 
However, based on a more careful analysis of the case with 32-bit keys, Thorup shows that using $c=2$ characters of 16 bits, and $d=20$ derived characters, gives a 100-independent hash function with probability $1-1.5\times 10^{-42}$. According to~\cite{Tho13:simple-simple} we cannot use significantly fewer resources even if we just want 4-independence. 
For hashing 32-bit keys, this means making $22$ lookups for each query and using tables of total size $40\cdot 2^{16}$.  
In contrast, if we hash 32-bit keys with tabulation-permutation hashing, we may use 8-bit characters with $d=c=4$, thus making only $8$ lookups in tables of total size $8\cdot 2^{8}$. 
For this setting of parameters, our experiments (summarized in Table~\ref{tab:speed}) show that double tabulation is approximately 30 times slower than tabulation-permutation hashing. 
For 64-bit keys, Thorup~\cite{Tho13:simple-simple} suggests implementing double tabulation with $c=3$ characters of 22 bits and $d=24$. 
This would require $26$ lookups in tables of total size $48\cdot 2^{22}$. We were not able to implement this on a regular laptop due to the space requirement.

We finally mention that Christani et al.~\cite{christiani15indep} have presented a hash family which obtains the even higher independence $u^{\Omega(1/c)}$. The scheme is, however, more complicated with a slower evaluation time of $\Theta( c \log c)$.

\subsubsection{Space Bounded Independence and Chernoff Bounds}\label{sec:space-bounded-ind}
One of the earliest attempts of obtaining strong concentration bounds via hashing is a simple and elegant construction by Dietzfelbinger and Meyer auf der Heide~\cite{Dietzfelbinger1992}. For some parameters $m,s,d$, their hash family maps to $[m]$, can be represented with $O(s)$ space, and uses a $(d+1)$-independent hash function as a subroutine, where $d=O(1)$, e.g., a degree-$d$ polynomial. In terms of our main goal of Chernoff-style bounds, their result can be cast as follows: Considering the number of balls from a fixed, but unknown, subset $S \subseteq U$, with $|S|=n$, that hashes to a specific bin, their result yields Chernoff bounds like ours with a constant delay in the exponential decrease and with an \errorterm of $n\left(\frac{n}{ms}\right)^d$. The expected number of balls in a given bin is $\mu=n/m$, so the \errorterm is $n(\mu/s)^d$. To compare with tabulation-permutation, suppose we insist on using space $O(|\Sigma|)$ and that we moreover want the \errorterm to be $u^{-\gamma}=|\Sigma|^{-c\gamma}$ like in~\Cref{thm:intro-tab-perm,thm:intro-tab-1perm}. With the hashing scheme from~\cite{Dietzfelbinger1992}, we then need $\mu =O(|\Sigma|^{1-\gamma c/d})$. If we want to be able to handle expectations of order, e.g. $|\Sigma|^{1/2}$, we thus need $d\geq 2c\gamma$. For $64$-bit key, $c=8$, and $\gamma=1$, say, this means that we need to evaluate a $16$-independent hash function.
In general, we see that the concentration bound above suffers from the same issues as those provided by P\v{a}tra\c{s}cu and Thorup for simple and twisted tabulation~\cite{patrascu12charhash,PT13:twist}, namely that we only have Chernoff-style concentration if the expected value is much smaller than the space used.

Going in a different direction, Dietzfelbinger and Rink~\cite{dietzfel09splitting} use universe
splitting to create a hash function that is highly independent (building on
previous works~\cite{Dietz2007,dietzfel03tabhash,Fotakis2005,hagerup01perfhash})
but, contrasting double tabulation as described above,
only within a fixed set $S$, not the entire universe. The construction requires an
upper bound $n$ on the size of $S$, and a polynomial 
error probability of $n^{-\gamma}$ is tolerated. Here $\gamma=O(1)$ is part of
the construction and affects the evaluation time. Assuming no such error has occurred, which is not checked, the hash function is highly independent on
$S$. As with Siegel's and Thorup's highly independent hashing discussed above, this implies Chernoff bounds without the constant delay in the exponential decrease, but this time only within the fixed set $S$. In the same setting, Pagh and Pagh \cite{PP08} have presented
a hash function that uses $(1+o(1))n$ space and which is fully independent on any given set $S$  of size at most $n$ with 
high probability. This result is very useful, e.g., as part of
solving a static problem of size $n$ using linear space, since, with high probability, we may assume fully-random hashing as
a subroutine.  However, from a
Chernoff bound perspective, the fixed polynomial error probability implies
that we do not benefit from any independence above $O(\log n)$, using the aforementioned results from~\cite{schmidt95chernoff}. 
More importantly, we do not
want to impose any limitations to the size of the sets we wish to hash in this paper. Consider
for example the classic problem of counting distinct elements
in a huge data stream. The size of the data stream might be very large, but regardless,
the hashing schemes of this paper will only use space $O(u^{1/c})$ with $c$ chosen 
small enough for hash tables to fit in fast cache.

Finally, Dahlgaard et al.\ \cite{DKRT15:k-part} have shown that on a given set $S$ of size $|S|\leq |\Sigma|/2$ a double
tabulation hash function, $h=h_2\circ h_1$ as described above, is fully random with probability
$1-|\Sigma|^{1-\lfloor d/2\rfloor}$ over the choice of $h_1$. For an error probability of $1/u$, we set $d=(2c+2)$ yielding a hash function that can be evaluated with $3c+2$ character lookups and using $(4c+4)|\Sigma|$ space.  This
can be used to simplify the above construction by Pagh and Pagh \cite{PP08}. Dahlgaard et
al. \cite{DKRT15:k-part} also propose mixed tabulation hashing which
they use for statistics over $k$-partitions. Their analysis is easily
modified to yield Chernoff-style bounds for intervals similar to our bounds for
tabulation-1permuation hashing presented in Theorem \ref{thm:intro-tab-1perm},
but with the restriction that the expectation $\mu$ is at most
$|\Sigma|/\log^{2c} |\Sigma|$. This restriction is better than the earlier mentioned restictions $\mu\leq |\Sigma|^{1/2}$ for simple
tabulation \cite{patrascu12charhash} and $\mu\leq
|\Sigma|^{1-\Omega(1)}$ for twisted tabulation \cite{PT13:twist}.
For mixed tabulation hashing, Dahlgaard et al.\ use $3c+2$ lookups and
$(5c+4)|\Sigma|$ space. In comparison, tabulation-1permutation hashing, which has no restriction on $\mu$, uses only
 $c+1$ lookups and $(c+1)|\Sigma|$ space.

\subsubsection{Small Space Alternatives in Superconstant Time}\label{sec:superconstant-time}
Finally, there have been various interesting developments regarding hash functions with
small representation space that, for example, can hash $n$ balls to $n$ bins
such that the maximal number of balls in any bin is $O(\log n/\log\log n)$,
corresponding to a classic Chernoff bound.  Accomplishing this through independence of
the hash function, this requires $O(\log n/\log\log n)$-independence and evaluation time
unless we switch to hash functions using a lot of space as described above. However,
\cite{CRSW11:balance,MRRR14:balance} construct hash families taking a random seed of
$O(\log\log n)$ words and which can be evaluated using $O((\log\log n)^2)$ operations,
still obtaining a maximal load in any bin of $O(\log n/\log\log n)$ with high probability.
This is impressive as it only uses a small amount of space and a short random seed, though
it does require some slightly non-standard operations when evaluating the hash functions.
The running time however, is not constant, which is what we aim for in this paper. 

A different result is by~\cite{GopalanKM18} who construct hash families which hash $n$
balls to $2$ bins. They construct hash families that taking a random seed of $O((\log\log n)^2)$
words get Chernoff bounds with an \errorterm of $n^{-\gamma}$ for some constant $\gamma$,
which is similar to our bounds. Nothing is said about the running time of the hash function of~\cite{GopalanKM18}. Since one of our primary goals is to design hash functions with constant running time, this makes the two results somewhat incomparable.

\subsubsection{Experiments and Comparisons}\label{sec:intro-exp}
To better understand the real-world performance of our new hash functions in comparison with well-known and comparable alternatives, we performed some simple experiments on regular laptops, as presented in~\Cref{tab:speed}. We did two types of experiments. 
\begin{itemize}
\item On the one hand we compared with potentially faster hash functions with weaker or restricted concentration bounds to see how much we lose in speed with our theoretically strong tabulation-permutation hashing. We shall see that our tabulation-permutation is very competitive in speed.
\item On the other hand we compared with the fastest previously known hashing schemes with strong concentration bounds like ours. Here we will see that we gain a factor of 30 in speed.
\end{itemize}

Concerning weaker, but potentially faster, hashing schemes we have chosen two types of hash functions for the comparison. First, we have the fast 2-independent hash functions multiply-shift (with addition) and 2-independent PolyHash. They are among the fastest hash functions in use and are commonly used in streaming algorithms. It should be noted that when we use 2-independent hash functions, the variance is the same as with full randomness, and it may hence suffice for applications with constant error probability. Furthermore, for data sets with sufficient entropy, Chung, Mitzenmacher, and Vadhan \cite{chung2013simple} show that 2-independent hashing suffices. However, as previously mentioned, we want provable Chernoff-style concentration bounds of our hash functions, equivalent up to constant factors to the behavior of a fully random hash function, for any possible input.
Second, we have simple tabulation, twisted tabulation, and mixed tabulation, which are tabulation based hashing schemes similar to tabulation-1permutation and tabulation-permutation hashing, but with only restricted concentration bounds. It is worth noting that Dahlgaard, Knudsen, and Thorup \cite{Dahlgaard2017Practical} performed experiments showing that the popular hash functions MurmurHash3 \cite{murmur3} and CityHash \cite{cityhash} along with the cryptographic hash function Blake2 \cite{Aumasson2013Blake} all are slower than mixed tabulation hashing, which we shall see is even slower than permutation-tabulation hashing. These hash functions are used in practice, but given that our experiments show mixed tabulation to be slightly slower than tabulation-permutation hashing, these can now be replaced with our faster alternatives that additionally provide theoretical guarantees as to their effectiveness.

Concerning hashing schemes with previous known strong concentration bounds, we compared with double tabulation and 100-independent PolyHash, which are the strongest competitors that we are aware of using simple portable code.

The experiment measures the time taken by various hash functions to hash a large set of keys. Since the hash functions considered all run the same instructions for all keys, the worst- and best-case running times are the same, and hence choosing random input keys suffices for timing purposes. 
Further technical details of the experiments are covered in \cref{sec:experiments}. We considered both hashing 32-bit keys to 32-bit hash values and 64-bit keys to 64-bits hash values. We did not consider larger key domains as we assume that a universe reduction, as described in~\Cref{sec:unired}, has been made if needed. The results are presented in \cref{tab:speed}. Below, we comment on the outcome of the experiment for each scheme considered.

\paragraph{Multiply-Shift.} The fastest scheme of the comparison is Dietzfelbinger's 2-independent
Multiply-Shift~\cite{dietzfel96universal}. For 32-bit keys it uses
one 64-bit multiplication and a shift. For 64-bit keys it uses 
one 128-bit multiplication and a shift. As expected, this very simple hash function was the fastest in the experiment. 

\paragraph{2-Independent PolyHash.} We compare twice with the classic $k$-independent
PolyHash~\cite{wegman81kwise}. Once with $k=2$ and again with $k=100$. $k$-independent PolyHash is based on evaluating a random degree $(k-1)$-polynomial over a prime field, using Mersenne primes to make it fast:
$2^{61}-1$ for 32-bit keys and $2^{89}-1$ for 64-bit keys. 
The 2-independent version was 2-3 times slower in experiments than multiply-shift. It is possible that implementing PolyHash with a specialized carry-less multiplication \cite{Lemire2016} would provide some speedup. However, we do not expect it to become faster than multiply-shift.

\paragraph{Simple Tabulation.} The baseline for comparison of our tabulation-based schemes is simple tabulation hashing. Recall that we hash using $c$ characters from $\Sigma=[u^{1/c}]$ (in this experiment we considered $u=2^{32}$ and $u=2^{64}$). This implies $c$ lookups from the character tables, which have total size $c\abs \Sigma$. For each lookup, we carry out a few simple AC$^0$ operations, extracting the characters for the lookup and applying an XOR. Since the size of the character alphabet influences the lookup times, it is not immediately clear, which choice of $c$ will be the fastest in practice. This is, however, easily checkable on any computer by simple experiments. In our case, both computers were fastest with 8-bit characters, hence with all character tables fitting in fast cache.

Theoretically, tabulation-based hashing methods are incomparable in speed to multiply-shift and 2-independent PolyHash, since the latter methods use constant space but multiplication which has circuit complexity $\Theta(\log w/\log\log w)$ for $w$-bit words \cite{CSV84}. Our tabulation-based schemes use
only AC$^0$ operations, but larger space. This is an inherent difference, as 2-independence is only possible with AC$^0$ operations using a large amount of space \cite{AMRT96,MNT93,Mil96}.
As is evident from \cref{tab:speed}, our experiments show that simple
tabulation is 2-3 slower than multiply-shift, but as fast or faster
than 2-independent PolyHash. Essentially, this can be ascribed to the cache of the two computers used being comparable in speed to arithmetic instructions. This is not surprising as most computation in the world involves data and hence cache. It is therefore expected that most computers have cache as fast as arithmetic instructions. In fact, since fast multiplication circuits are complex and expensive, and a lot of data processing does not involve multiplication, one could imagine computers with much faster cache than multiplication \cite{HP12:hardware}.

\paragraph{Twisted Tabulation.}
Carrying out a bit more work than simple tabulation, twisted tabulation performs $c$ lookups of entries that are twice the size, as well as executing a few extra AC$^0$ operations. It hence performs a little worse than simple tabulation hashing. 

\paragraph{Mixed Tabulation.} We implemented mixed tabulation hashing with the same parameters ($c=d$) as in~\cite{Dahlgaard2017Practical}. With these parameters the scheme uses $2c$ lookups from $2c$ character tables, where $c$ of the lookups are to table entries that are double as long as the output, which may explain its worse performance with 64-bit domains. In our experiments, mixed tabulation performs slightly worse than tabulation-permutation hashing. Recall from above that mixed tabulation is faster than many popular hash functions without theoretical guarantees, hence so is our tabulation-permutation.

\paragraph{Tabulation-1Permutation.} Also only slightly more involved than simple tabulation, tabulation-1permutation performs $c+1$ lookups using $c+1$ character tables.
In our experiments, tabulation-1permutation turns out to be a little bit 
faster than twisted tabulation, at most 30\% slower than simple tabulation,
and at most 4 times slower than multiply-shift. Recall that tabulation-1permutation is our hash function of choice for streaming applications where speed is critical.

\paragraph{Tabulation-Permutation.} Tabulation-permutation hashing performs $2c$ lookups from $2c$ character
tables. In our experiments, it is slightly more than twice as slow as simple
tabulation, and at most 8 times slower than multiply-shift. It is also worth noting that it performs better than mixed tabulation.

\paragraph{Double Tabulation.}
Recall that among the schemes discussed so far, only tabulation-permutation and tabulation-1permutation
hashing offer unrestricted Chernoff-style concentration with high probability. Double tabulation is the first alternative with similar guarantees and in our experiments it is 30 times slower for 32-bit keys. For 64-bit keys, we were unable to run it on the computers at our disposal due to the large amount of space required for the hash tables. As already discussed, theoretically, double tabulation needs more space and lookups. The 32-bit version performed $26$ lookups in tables of total size $48\cdot 2^{22}$, while tabulation-permutation only needs $8$ lookups using $8\cdot 2^{8}$
space. It is not surprising that double tabulation lags so far behind.

\paragraph{100-Independent PolyHash.}
Running the experiment with 100-independent PolyHash, it turned out that for 32-bit
keys, it is slower than 100-independent double tabulation.
A bit surprisingly, 100-independent PolyHash ran nearly 200 times slower
than the 2-independent PolyHash, even though it essentially just runs
the same code 99 times. An explanation could be that the 2-independent scheme just keeps two coefficients in registers while the 100-independent
scheme would loop through all the coefficients.
We remark that the number $100$ is somewhat arbitrary. We need $k=\Theta(\log u)$, but we do not know the exact constants in the Chernoff bounds with $k$-independent hashing. The running times are, however, easily scalable and for $k$-independent PolyHash, we would expect the evaluation time to change by a factor of roughly $k/100$.

\begin{figure}
    \centering
    \def\svgwidth{\textwidth}
    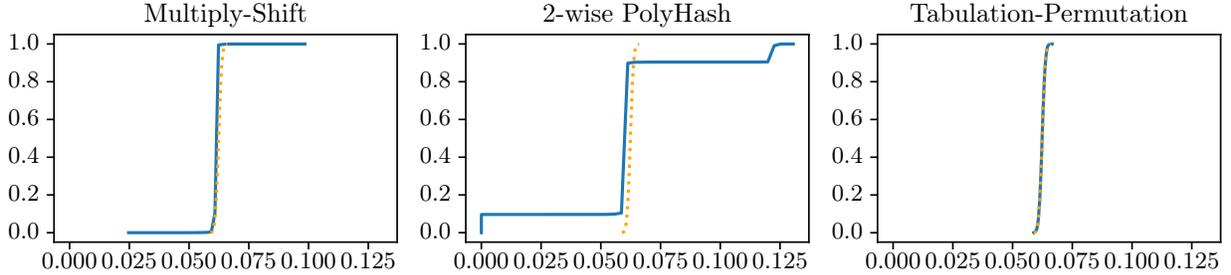
    \caption{Hashing the arithmetic progression $\setbuilder{a \cdot i}{i \in [50000]}$
    to $16$ bins for a random integer $a$. The dotted line is a 100-independent PolyHash.}
    \label{fig:polynomials_intro}
\end{figure}

\paragraph{Bad instances for Multiply-Shift and 2-wise PolyHash} We finally present experiments demonstrating concrete bad instances for the hash functions Multiply-Shift~\cite{dietzfel96universal} and 2-wise PolyHash, underscoring what it means for them to not support Chernoff-style concentration bounds. 
In each case, we compare
with our new tabulation-permutation hash function as well as 100-independent PolyHash,
which is our approximation to an ideal fully random hash function.
We refer the reader to~\Cref{sec:experiments} for bad instances for simple-tabulation~\cite{zobrist70hashing} and twisted tabulation~\cite{PT13:twist} as well as a more thorough discussion of our experiments.
%We note that all schemes considered are 2-independent, so they all have
%exactly the same variance as fully-random hashing. From $2$-independence, it also follows that the schemes work perfectly on sufficiently random input~\cite{mitzenmacher08hash}. Our concern is therefore concrete inputs making them fail in the tail.

Bad instances for Multiply-Shift and 2-independent
PolyHash are analyzed in detail in~\cite[Appendix
B]{patrascu16kwise-lb}.  The specific instance we consider is that of hashing the arithmetic
progression $A=\setbuilder{a \cdot i}{i \in [50000]}$ into $16$ bins, where we are interested in the number of keys from $A$ that hashes to a specific bin. We performed this experiment $5000$ times, with independently chosen hash functions. The cumulative distribution functions on the number of keys from $A$ hashing to a specific bin is presented in~\Cref{fig:polynomials_intro}. We see that most of the
time 2-independent PolyHash and Multiply-Shift distribute the keys perfectly with exactly $1/16$ of
the keys in the bin. By $2$-independence, the variance is the same as with fully random hashing, and this should suggest a  much heavier tail,
which is indeed what our experiments show. For contrast, we see that the cumulative distribution function with our
tabulation-permutation hash function is almost indistinguishable from that of 100-independent Poly-Hash.
We note that no amount of experimentation can prove that tabulation-permutation (or any other hash function) works well for all possible inputs. However, given the mathematical concentration guarantee of Theorem 2, the strong performance of tabulation-permutation in the experiment is no surprise.

\section{Technical Theorems and how they Combine}
We now formally state our main technical results, in their full generality, and show how they combine to yield Theorems~\ref{thm:intro-simple-tab},~\ref{thm:intro-tab-perm}, and~\ref{thm:intro-tab-1perm}. A fair warning should be given to the reader. The theorems to follow are intricate and arguably somewhat inaccessible at first read. Rather than trying to understand everything at once, we suggest that the reader use this section as a roadmap for the main body of the paper. We will, however, do our best to explain the contents of the results as well as disentangling the various assumptions in the theorems. 

As noted in \cref{sec:introduction_techniques}, the exposition is subdivided into three parts, each yielding theorems that we believe to be of independent interest. First, we provide an improved analysis of simple tabulation (\Cref{sec:simple}). We then show how permuting the output of a simple tabulation hash function yields a hash function having Chernoff bounds for arbitrary value functions (\Cref{sec:valuefunctions}). Finally, we show that concatenating the output of two independent hash functions preserves the property of having Chernoff bounds for arbitrary value functions
(\Cref{sec:codomainext}).

It turns out that the proofs of our results become a little cleaner when we assume that value functions take values in $[-1,1]$, so from here on we state our results in relation to such value functions. Theorems~\ref{thm:intro-simple-tab},~\ref{thm:intro-tab-perm}, and~\ref{thm:intro-tab-1perm} will still follow, as the value functions in these theorems can also be viewed as having range $[-1,1]$.

\subsection{Improved Analysis of Simple Tabulation}\label{sec:combine_simple} Our new and improved result on simple tabulation is the subject of~\Cref{sec:simple}. It is stated as follows.
\begin{theorem}\label{thm:simpleConcentration}
	Let $h\colon \Sigma^c \to [m]$ be a simple tabulation hash function and
    $S \subseteq \Sigma^c$ be a set of keys of size $n=\abs S$. Let $v\colon \Sigma^c \times [m] \to [-1, 1]$ be
    a value function such that the set
    $Q = \setbuilder{i \in [m]}{\exists x \in \Sigma^c : v(x, i) \neq 0}$ satisfies  $\abs{Q} \le m^{\eps}$,
    where $\eps < \frac{1}{4}$ is a constant.
    \begin{enumerate}
    	\item For any constant $\gamma\geq 1$, the random variable $V=\sum_{x\in S}v(x, h(x))$ is strongly concentrated with \errorterm $O_{\gamma, \eps, c}(n/m^\gamma)$, where the constants of the asymptotics are determined by $c$ and $\gamma$. Furthermore, this concentration is query invariant.
    	\item For $j\in [m]$ define the random variable $V_j = \sum_{x\in S}v(x, h(x)\xor j)$ and let $\mu = \E{V_j}$, noting that this is independent of $j$. For any $\gamma\geq 1$,
    \begin{align}
    	\Pr \left[\sum_{j\in [m]}(V_j-\mu)^2>D_{\gamma, c} \sum_{x\in S}\sum_{k\in [m]} v(x, k)^2\right] = O_{\gamma, \eps, c}(n/m^\gamma)\label{eq:squares}
    \end{align}
    for some constant $D_{\gamma, c}$ and this bound is query invariant up to constant factors.
    \end{enumerate}
\end{theorem}
The technical assumption involving $Q$ states that the value function has \emph{bounded support} in the hash range: The value $v(x,h(x))$ can only possibly be non-zero if $h(x)$ lies in the relatively small set $Q$ of size at most $m^{\eps}$. In fact, when proving~\Cref{thm:intro-simple-tab} it suffices to assume that $|Q|= 1$, as we shall see below, but for our analysis of tabulation-permutation hashing we need the more general result above. Another nice illustration of the power of~\Cref{thm:simpleConcentration} holding with value functions of \emph{any} bounded support will appear when we prove~\Cref{thm:intro-tab-1perm} in~\Cref{sec:concinintervals}. 

To see that \cref{thm:intro-simple-tab} is implied by \cref{thm:simpleConcentration}, one may observe that the latter is a generalization of the former. Let $y\in [m]$ be the bin and $(w_x)_{x\in S}$ be the weights of the balls from $S$ in the setting of \cref{thm:intro-simple-tab}. Then defining the value function $v\colon \Sigma^c\times [m]\to [0, 1]$, 
\begin{align*}
	v(x, y') = 
	\begin{cases}
		w_x\cdot [y'=y], & x\in S,\\
		0, & x\not\in S,
	\end{cases}
\end{align*}
we find that $X=\sum_{x\in S}w_x\cdot [h(x)=y] = \sum_{x\in S}v(x, h(x))$ is strongly concentrated by part 1 of \cref{thm:simpleConcentration} and the concentration is query invariant.

Finally, the bound~\eqref{eq:squares} requires some explaining. For this, we consider the toy example of \cref{thm:intro-simple-tab}. Suppose we have a set $S\subseteq[u]$ of balls with weights $(w_x)_{x\in S}$ and we throw them into the bins of $[m]$ using a simple tabulation hash function. We focus on the total weight of balls landing in bin $0$, defining the value function by $v(x,y)=w_x$ for $x \in S$ and $y=0$, and $v(x,y)=0$ otherwise. In this case, $ \mu=\frac{1}{m}\sum_{x\in S} w_x$ denotes the expected total weight in any single bin and $V_j=\sum_{x\in S}w_x \cdot [h(x)=j]$ denotes the total weight in bin $j\in [m]$. Then~\eqref{eq:squares} states that $\sum_{j \in [m]}(V_j-\mu)^2=O(\|w\|_2^2)$ with high probability in $m$. This is exactly a bound on the \emph{variance} of the weight of balls landing in one of the bins when each of the hash values of the keys of $S$ are shifted by an XOR with a uniformly random element of $[m]$. Note that this  example corresponds to the case where $\abs Q=1$. In its full generality, i.e., for general value functions of bounded support,~\eqref{eq:squares} is similarly a bound on the variance of the value obtained from the keys of $S$ when their hash values are each shifted by a uniformly random XOR.
This variance bound turns out to be an important ingredient in our proof of the strong concentration in the first part of~\cref{thm:simpleConcentration}. As described in Section~\ref{sec:techsimple} the proof proceeds by fixing the hash values of the position characters $[c] \times \Sigma$ in a carefully chosen order, $\alpha_1 \prec \cdots \prec \alpha_r$. Defining $G_i$ to be those keys that contain $\alpha_i$ as a position character but no $\alpha_j$ with $j>i$, the internal clustering of the keys of $G_i$ is determined solely by $(h(\alpha_j))_{j<i}$ and fixing $h(\alpha_i)$ ``shifts'' each of these keys by an XOR with $h(\alpha_i)$. Now~\eqref{eq:squares}, applied with $S=G_i$, exactly yields a bound on the \emph{variance} of the total value obtained from the keys from $G_i$ when fixing the random XOR $h(\alpha_i)$. Thus,~\eqref{eq:squares} conveniently bounds the variance of the martingale described in Section~\ref{sec:techsimple}. 
As such,~\eqref{eq:squares} is merely a technical tool, but we have a more important reason for including the bound in the theorem. As it turns out, for \emph{any} hash function satisfying the conclusion of~\Cref{thm:simpleConcentration}, composing with a uniformly random permutation yields a hash family having Chernoff-style concentration bounds for any value function as we describe next.

\subsection{Permuting the Hash Range}
Our next step in proving~\Cref{thm:intro-tab-perm} is to show that, given a hash
function with concentration bounds like
in~\Cref{thm:simpleConcentration}, composing with a uniformly random permutation
of the entire range yields a hash function with Chernoff-style concentration
for general value functions. The main theorem, proved in~\Cref{sec:valuefunctions}, is as follows.	\begin{theorem}\label{thm:valuefunctions}
		Let $\eps\in (0, 1]$ and $m\geq 2$ be given. Let $g\colon [u]\to [m]$ be a 3-independent hash function satisfying the following. For every $\gamma>0$, and for every value function $v\colon [u]\times [m]\to [-1, 1]$ such that the set $Q=\setbuilder{i\in [m]}{\exists x\in [u]\colon v(x, i)\neq 0}$ is of size $\abs Q\leq m^\eps$, the two conclusions of \cref{thm:simpleConcentration} holds with respect to $g$.
		
		Let $v'\colon [u]\to [-1, 1]$ be any value function, $\tau\colon [m]\to [m]$ be a uniformly random permutation independent of $g$, and $\gamma>0$. Then the for the hash function $h=\tau\circ g$, the sum $\sum_{x\in [u]}v'(x, h(x))$ is strongly concentrated with \errorterm $O_{\gamma, \eps}(u/m^{\gamma})$, where the constants of the asymptotics are determined solely by $\eps$ and $\gamma$. Furthermore, this concentration is query invariant.
	\end{theorem}
We believe the theorem to be of independent interest. From a hash function that only performs well for value functions supported on an asymptotically small subset of the bins we can construct a hash function performing well for any value function -- simply by composing with a random permutation.~\Cref{thm:simpleConcentration} shows that simple tabulation satisfies the two conditions in the theorem above. It follows that if $m=|U|^{\Omega(1)}$, e.g., if $m=|\Sigma|$, then composing a simple tabulation hash function $g:\Sigma^c \to [m]$ with a uniformly random permutation $\tau:[m] \to [m]$ yields a hash function $h=\tau \circ g$ having Chernoff-style bounds for general value functions asymptotically matching those from the fully random setting up to an \errorterm inversely polynomial in the size of the universe.
In particular these bounds hold for tabulation-permutation hashing from $\Sigma^c$ to $\Sigma$, that is, using just a single permutation, which yields the result of~\Cref{thm:intro-tab-perm} in the case $d=1$. If we desire a range of size $m\gg |\Sigma|$ the permutation $\tau$ becomes too expensive to store. Recall that in tabulation-permutation hashing from $\Sigma^c$ to $\Sigma^d$ we instead use $d$ permutations $\tau_1,\dots, \tau_d:\Sigma \to \Sigma$, hashing
$$
\Sigma^c \xrightarrow[]{g^{\text{simple}}} \Sigma^{d} \xrightarrow[]{(\tau_1,\dots,\tau_d)} \Sigma^d.
$$
Towards proving that this is sensible, the last step in the proof of~\Cref{thm:intro-tab-perm} is to show that concatenating the outputs of independent hash functions preserves the property of having Chernoff-style concentration for general value functions.

\subsection{Squaring the Hash Range}
The third and final step towards proving~\Cref{thm:intro-tab-perm} is showing that
concatenating the hash values of two independent hash functions each with Chernoff-style
bounds for general value functions yields a new hash function with
similar Chernoff-style bounds up to constant factors.
In particular it will follow that tabulation-permutation hashing has Chernoff-style
bounds for general value functions. However, as with~\Cref{thm:valuefunctions}, the result is of more general interest. Since it uses the input hash functions in a black box manner, it is a general tool towards constructing new hash functions with Chernoff-style bounds. The main theorem, proved in~\Cref{sec:codomainext}, is the following.
\begin{theorem}\label{thm:extendingCodomain}
	Let $h_1\colon A\to B_1$ and $h_2\colon A\to B_2$ be 2-wise independent hash functions with a common domain such that for every pair of value functions, $v_1\colon A\times B_1\to [-1, 1]$ and $v_2\colon A\times B_2\to [-1, 1]$, the random variables $X_1 = \sum_{a\in A}v_1(a, h_1(a))$ and $X_2 = \sum_{a\in A}v_2(a, h_2(a))$ are strongly concentrated with \errorterm $f_1$ and $f_2$, respectively, and the concentration is query invariant. Suppose further that $h_1$ and $h_2$ are independent. Then the hash function $h = (h_1, h_2)\colon A\to B_1\times B_2$, which is the concatenation of $h_1$ and $h_2$, satisfies that for every value function $v\colon A\times(B_1\times B_2)\to [-1, 1]$, the random variable $X=\sum_{a\in A} v(a, h(a)) = \sum_{a\in A}v(a, h_1(a), h_2(a))$ is strongly concentrated with additive error $O(f_1+f_2)$ and the concentration is query invariant.
\end{theorem}
We argue that~\Cref{thm:extendingCodomain}, combined with the previous results, leads to~\Cref{thm:intro-tab-perm}. 
\begin{proof}[Proof of~\Cref{thm:intro-tab-perm}]

We proceed by induction on $d$. For $d=1$ the result follows from~\Cref{thm:simpleConcentration} and~\ref{thm:valuefunctions} as described in the previous subsection. Now suppose $d>1$ and that the result holds for smaller values of $d$. Let $\gamma=O(1)$ be given.  Let $d_1=\lfloor d/2 \rfloor$ and $d_2=\lceil d/2 \rceil$. A tabulation-permutation hash function $h: \Sigma^c \to \Sigma^d$ is the concatenation of two independent tabulation-permutation hash functions $h_1:\Sigma^c \to \Sigma^{d_1}$ and $h_2:\Sigma^c \to \Sigma^{d_2}$. Letting $A=\Sigma^c$, $B_1=\Sigma^{d_1}$, $B_2=\Sigma^{d_2}$, the induction hypothesis gives that the conditions of \cref{thm:extendingCodomain} are satisfied and the conclusion follows. Note that since $d=O(1)$, the induction is only applied a constant number of times. Hence, the constants hidden in the asymptotics of \cref{def:strongly-concentrated} are still constant.
\end{proof}

\subsection{Concentration in Arbitrary Intervals.}\label{sec:concinintervals}
We will now show how we can use our main result,~\Cref{thm:intro-tab-perm}, together with our improved understanding of simple tabulation~\Cref{thm:simpleConcentration} to obtain~\Cref{thm:intro-tab-1perm} which shows that the extra efficient tabulation-1permutation hashing provides Chernoff-style concentration for the special case of weighted balls and intervals. This section also serves as an illustration of how our previous results play in tandem, and it illustrates the importance of~\Cref{thm:simpleConcentration} holding, not just for single bins, but for any value function of bounded support.

\begin{proof}[Proof of~\Cref{thm:intro-tab-1perm}]
Let $S\subseteq[u]$ be a set of keys, with each key $x\in S$ having a weight 
$w_x\in [0,1]$. Let $h=\tau\circ g\colon \Sigma^c \to \Sigma^d=[r]$ be a tabulation-1permutation hash function, with $g:\Sigma^c \to  \Sigma^d$ a simple tabulation hash function and $\tau\colon \Sigma^d \to \Sigma^d$ a random permutation of the most significant character, $\tau(z_1,\dots,z_d)=(\tau_1(z_1),z_2,\dots,z_d)$ for a uniformly random permutation $\tau_1 \colon \Sigma \to \Sigma$. Let $y_1,y_2\in\Sigma^d$ and $X$ be defined as in~\Cref{thm:intro-tab-1perm}, $X=\sum_{x\in
S}w_x\cdot [y_1\leq h(x)< y_2]$. Set $\mu=\E{X}$, and $\sigma^2=\Var{X}$. 
For simplicity we assume that $|I|\geq r/2$.  Otherwise, we just apply the argument below with $I$ replaced by $[r]\setminus I=[0,y_1)\cup[y_2,r)$, which we view as an interval in the cyclic ordering of $[r]$. We will partition $I=[y_1,y_2)$ into a constant number of intervals in such a way that our previous results yield Chernoff style concentration bound on the total weight of keys landing within each of these intervals. The desired result will follow.

To be precise, let $t>0$ and $\gamma=O(1)$ be given. Let $P_1=\{x \in \Sigma\mid \forall y \in \Sigma^{d-1}:(x,y)\in I\}$ and $I_1=\{(x_1,\dots,x_d)\in \Sigma^d \mid x_1 \in P_1\}$. Whether or not $h(x)\in I_1$ for a key $x\in \Sigma^c$ depends solely on the most significant character of $h(x)$. With $X_1=\sum_{x\in
S}w_x\cdot [h(x)\in I_1]$, $\mu_1=\E{X_1}$, and $\sigma_1^2=\Var{X_1}$, we can therefore apply~\Cref{thm:intro-tab-perm} to obtain that for any $t'>0$ and $\gamma'=O(1)$,
\begin{align}\label{eq:Cherforfirstint}
\Pr[|X_1-\mu_1|\geq t'] \le C\exp(-\Omega(\sigma_1^2\cC(t'/\sigma_1^2)))+1/u^{\gamma'}\leq C\exp(-\Omega(\sigma^2\cC(t'/\sigma^2)))+1/u^{\gamma'},
\end{align}
for some constant $C$.
Here we used that $\sigma_1^2\leq \sigma^2$ as $|I_1|\leq |I|\leq |\Sigma^d|/2$.
Next, let $d_1=\lg |\Sigma|$ and $d_2,\dots, d_\ell\in \N$ be such that for $2\leq i \leq \ell$, it holds that $2^{d_i}\leq (2^{d_1+d_2+\cdots+d_i})^{1/4}$, and further $2^{d_1+d_2+\cdots+d_{\ell}}=|\Sigma|^{d}$. We may assume that $u$ and hence $|\Sigma|$ is larger than some constant as otherwise the bound in~\Cref{thm:intro-tab-1perm} is trivial. It is then easy to check that we may choose $\ell$ and the $(d_i)_{2\leq i \leq \ell}$ such that $\ell =O(\log d)=O(1)$. We will from now on consider elements of $\Sigma^d$ as numbers written in binary or, equivalently, bit strings of length $d':=d_1+\cdots+d_\ell$. For $i=1,\dots,\ell$ we define a map $\rho_i:\Sigma^d \to [2]^{d_1+\cdots+d_i}$ as follows. If $x=b_1\dots b_{d'}\in [2]^{d'}$, then $\rho_i(x)$ is the length $d_1+\cdots+d_i$ bit string $b_{1}\dots b_{d_1+\cdots+d_{i}}$
 Set $J_1=I$. For $i=2,\dots,\ell$ we define $J_i\subseteq I$ and $I_i\subseteq I$ recursively as follows. First, we let $J_i=J_{i-1}\setminus I_{i-1}$. Second, we define $I_i$ to consist of those elements of $x\in J_i$ such that if $y\in \Sigma^c$ has $\rho_i(y)=\rho_i(x)$, then $y\in J_i$. In other words, $I_i$ consists of those elements of $ J_i$ that remain in $J_i$ when the least significant $d_{i+1}+\dots+d_\ell$ bits of $x$ are changed in an arbitrary manner. It is readily checked that for $i=1,\dots, \ell$, $I_i$ is a disjoint union of two (potentially empty) intervals $I_i=I_i^{(1)}\cup I_i^{(2)}$ such that for each $j\in \{1,2\}$ and $x,y\in I_i^{(j)}$, $\rho_i(x)=\rho_i(y)$. Moreover, the sets $(I_i)_{i=1}^\ell$ are pairwise disjoint and $I=\bigcup_{i=1}^\ell I_i$.

We already saw in~\eqref{eq:Cherforfirstint} that we have Chernoff-style concentration for the total weight of balls landing in $I_1$. We now show that the same is true for $I_i^{(j)}$ for each $i=2,\dots,\ell$ and $j\in\{0,1\}$. So let such an $i$ and $j$ be fixed. Note that whether or not $h(x)\in I_i^{(j)}$, for a key $x\in \Sigma^c$, depends solely on the most significant $d_1+\cdots+d_i$ bits of $h(x)$.  
Let $h':\Sigma^c\to [2]^{d_1+\cdots+d_i}$ be defined by $h'(x)=\rho_i(h(x))$. Then $h'$ is itself a simple tabulation hash function and $h'(x)$ is obtained by removing the $d_{i+1}+\cdots+d_\ell$ least significant bits of $h(x)$. Letting $I'=\rho_i(I_i^{(j)})$, it thus holds that $h(x)\in I_i^{(j)}$ if and only if $h'(x)\in I'$. Let now $X_i^{(j)}=\sum_{x\in
S}w_x\cdot [h(x)\in I_i^{(j)}]$, $\mu_i^{(j)}=\E{X_i^{(j)}}$, and $\sigma_1^2=\Var{X_i^{(j)}}\leq \sigma^2$. As $|I'|\leq 2^{d_i}\leq (2^{d_1+\cdots+d_i})^{1/4}$, we can apply~\Cref{thm:simpleConcentration} to conclude that for $t'>0$ and $\gamma'=O(1)$,
\begin{align}\label{eq:Cherforlastint}
\Pr[|X_i^{(j)}-\mu_i^{(j)}|\geq t'] \le C\exp(-\Omega(\sigma_1^2\cC(t'/\sigma_1^2)))+1/u^{\gamma'}\leq C\exp(-\Omega(\sigma^2\cC(t'/\sigma^2)))+1/u^{\gamma'}.
\end{align} 
Now applying~\eqref{eq:Cherforfirstint} and~\eqref{eq:Cherforlastint} with $t'=t/(2\ell-1)$ and $\gamma'=\gamma +\frac{\log (2\ell)}{\log u}=O(1)$, it follows that
\begin{align*}
\Pr[|X-\mu|\geq  t] 
\le&
\Pr[|X_1-\mu_1|\geq t']+\sum_{i=2}^\ell\sum_{j=1}^2 \Pr[|X_i^{(j)}-\mu_i^{(j)}|\geq t']
\le 
2C\ell\exp(-\Omega(\sigma^2\cC(t'/\sigma^2)))+2\ell/u^{\gamma'} \\
=&
O(\exp(-\Omega(\sigma^2 \; \cC(t/\sigma^2))))+1/u^\gamma,
\end{align*}
as desired.
\end{proof}

\section{Preliminaries}
Before proceeding, we establish basic definitions and describe results from the literature which we will use.
\subsection{Notation}
    Throughout the paper, we use the following general notation. 
    \begin{itemize}
        \item We let $[n]$ denote the set $\{0, 1, \dots, n-1\}$.
        \item For a statement or event $Q$ we let $[Q]$ be the indicator variable on $Q$, i.e.,\
    \begin{align*}
        [Q] = 
        \begin{cases}
            1, & Q\text{ occurred or is true},\\
            0, & \text{otherwise}.
        \end{cases}    
    \end{align*}
    \item Whenever  $Y_0, \dots, Y_{n-1}\in \R$ are variables and $i\in [n+1]$, we shall denote by $Y_{<i}$ the sum $\sum_{j<i}Y_j$. Likewise, whenever $A_0, \dots, A_{n-1}$ are sets and $i\in [n+1]$, we shall denote by $A_{<i}$ the set $\bigcup_{j<i}A_j$.
    \item Suppose we have a hash function $h\colon A\to B$ with domain $A$ and range $B$. We shall often associate weight and value functions with $h$ as follows.
		\begin{itemize}
			\item A function $w\colon A\to \R$ is called a \emph{weight function}, corresponding to the idea that every ball or key $x\in A$ has an associated weight, $w(x)\in \R$. Occasionally, we shall write $w_x$ for $w(x)$.
			\item A function $v\colon A\times B\to \R$ is called a \emph{value function}, with the interpretation that a key $x\in A$ yields a value $v(x, h(x))$ depending on the bin/hash value $h(x)\in B$.
		\end{itemize}
		  For weight functions $w:A \to \R$, a subset of balls, $S\subset A$, and a bin $y_0\in B$, we will be interested in sums of the form $W=\sum_{x\in S}w(x)[h(x)=y_0]$, i.e., the total weight of the balls in $S$ that are hashed to bin $y_0$. Defining the value function $v: A\times B \to \R$ by $v(x,y)=w(x)[y=y_0]$, $W$ is exactly equal to $\sum_{x\in S} v(x,h(x))$, i.e., the total value obtained by the balls in $S$. From this perspective, value functions are more general objects than weight functions.
    \end{itemize}

\subsection{Probability Theory and Martingales}
In the following, we introduce the necessary notions of probability theory.
 A note of caution is in order. The paper at hand relies on results from the theory of martingales to arrive at its conclusion. Working with martingales, we shall require probability theoretic notions of a fairly general and abstract character. For an introduction to measure and probability theory, see, for instance, \cite{schilling2005}.

For the most basic notation, let $(\Omega, \mathcal{F}, \Pr)$ be a probability space. 
\begin{itemize}
	\item Let $X_1, \dots, X_n:\Omega \to \R$ be $\mathcal{F}$-measurable random variables. We denote by $\mathcal G = \sigma(X_1, \dots, X_n)\subset \mathcal{F}$ the smallest $\sigma$-algebra such that $X_1, \dots, X_n$ are all $\mathcal G$-measurable. We say that $\mathcal G$ is the \emph{sigma algebra generated by} $X_1, \dots, X_n$. Intuitively, $\sigma(X_1, \dots, X_n)$ represents the collective information regarding the outcome of the joint distribution $(X_1, \dots, X_n)$.
	\item Let $X: \Omega \to \R$ be an $\mathcal{F}$-measurable random variable, and let $\mathcal{G}$ be a $\sigma$-algebra with $\mathcal{G}\subset \mathcal{F}$. If $\E{|X|}<\infty$, we may define the random variable $\EC X{\mathcal G}$ to be the \emph{conditional expectation} of $X$ given $\mathcal G$. It is important to note that $\EC{X}{\mathcal G}$ is $\mathcal G$-measurable. In the context of the above notation, $\EC X{\sigma(X_1, \dots, X_n)} = \EC X{X_1, \dots, X_n}$ is the expectation of $X$ as a function of the outcomes of $X_1, \dots, X_n$.
\end{itemize}

We proceed to discuss martingales and martingale differences. For convenience we shall assume all random variables to be bounded, i.e., whenever $X$ is a random variable, we assume that there exists a constant $M\geq 0$ such that $|X|\leq M$ almost surely. 

\begin{definition}[Filtration]
	Let $(\Omega,\mathcal{P}(\Omega),\Pr)$ be a finite measure space. A sequence of $\sigma$-algebras, $(\mathcal{F}_i)_{i=0}^r$, is a \emph{filtration} of $(\Omega,\mathcal{P}(\Omega),\Pr)$ if  $\{\emptyset,\Omega\}=\mathcal{F}_0\subseteq \mathcal{F}_1 \subseteq \cdots \subseteq \mathcal{F}_r=\mathcal{P}(\Omega)$. We shall usually omit explicit reference to the background space.
\end{definition}
\begin{definition}[Adapted Sequence]
	Let $(\mathcal F_i)_{i=0}^r$ be a filtration. A sequence of random variables $(X_i)_{i=0}^r$ is  \emph{adapted} to $(\mathcal F_i)_{i=0}^r$ if for every $i\in [r+1]$, $X_i$ is $\mathcal F_i$-measurable. In that case, we say that $(X_i, \mathcal F_i)$ is an \emph{adapted sequence}.
\end{definition}
\begin{definition}[Martingale]
	A \emph{martingale} is an adapted sequence, $(X_i, \mathcal F_i)$, satisfying that for every $i\in \{1, \dots, r\}$, $\EC{X_i}{\mathcal{F}_{i-1}}=X_{i-1}$.
\end{definition}
\begin{definition}[Martingale Difference]
	A \emph{martingale difference} is a an adapted sequence,  $(Y_i,\mathcal{F}_i)_{0=1}^r$, such that $Y_0=0$ almost surely and for every $i\in \{1, \dots, r\}$, $\EC{Y_i}{\mathcal{F}_{i-1}}=0$.
\end{definition} 
If $(X_i,\mathcal{F}_i)_{i=0}^r$ is a martingale, we may define the sequence of random variables $(Y_i)_{i=0}^r$ by $Y_0=0$ and $Y_i=X_i-X_{i-1}$ for $i=1,\dots,r$. Then $(Y_i,\mathcal{F}_i)_{i=0}^r$ is a martingale difference. Conversely, if $(Y_i,\mathcal{F}_i)_{i=0}^r$ is a martingale difference, a martingale $(X_i,\mathcal{F}_i)_{i=0}^r$ can be constructed by letting $X_i=Y_{<i+1}=\sum_{j\leq i}Y_j$. Under this correspondence, martingales and martingale differences are in a sense two sides of the same coin. 

 Concluding the section, we describe canonical constructions of a martingale and a martingale difference, respectively, that we shall use later on.
\begin{itemize}
	\item Let $X$ be a random variable and consider a filtration $(\mathcal{F}_i)_{i=0}^r$. We may define a martingale from $X$ with respect to $(\mathcal F_i)_{i=0}^r$ by defining the sequence  of random variables $(X_i)_{i=0}^r$ by $X_i=\EC{X}{\mathcal{F}_i}$ for each $i\in [r+1]$. Clearly, $\EC{X_i}{\mathcal{F}_{i-1}}=\EC{X}{\mathcal{F}_{i-1}} = X_{i-1}$, so $(X_i, \mathcal F_i)_{i=0}^r$ is indeed a martingale.

We shall apply this construction in the following situation. Suppose we have random variables $Z_1, \dots, Z_r$ taking values in the measure spaces $A_1, \dots, A_r$ and denote by $Z$ the joint distribution $(Z_1, \dots, Z_r)$. For some function $f\colon A_1\times \dots A_r\to \R$, we wish to assess the value of $f(Z)$. We may then define the filtration $\mathcal F_i=\sigma(Z_1, \dots, Z_i)$ for $i\in [r+1]$ and set $X_i=\EC{f(Z)}{\mathcal{F}_i}$. This yields a martingale $(X_i, \mathcal F_i)_{i=0}^r$ with $X_0=\E{f(Z)}$ and $X_r=f(Z)$. This is known as a Doob martingale and the construction will be used in~\Cref{sec:valuefunctions} to prove~\Cref{thm:valuefunctions}.
	\item Let $(Z_i, \mathcal{F}_i)_{i=0}^r$ be an adapted sequence and define $Y_0=0$ and $Y_i=Z_i-\EC{Z_i}{\mathcal{F}_{i-1}}$ for $i\in \{1, \dots, r\}$. Then $(Y_i,\mathcal{F}_i)_{i=1}^r$ is a martingale difference. This construction is applied in~\Cref{sec:simple} to prove~\Cref{thm:simpleConcentration}.
\end{itemize}

\subsection{Martingale Concentration Inequalities}
In applications of probability theory, we often consider a sequence of random variables $X_0, \dots, X_r$. If we are lucky, the random variables are independent, pair-wise independent, or a derivative thereof. It is unfortunately often the case, however, that there is no such independence notion that apply to $X_0, \dots, X_r$. One reason that martingales have been as successful as they are, is that frequently, one may instead impose a martingale structure on the variables, and martingales satisfy many of the same theorems that independent variables do. In this exposition, we shall consider sums of the form $X=\sum_{i=0}^rX_i$ where the $X_i$ are far from independent, yet we would like $X$ to satisfy Chernoff-style bounds. 

To this end, we state a martingale version of Bennett's inequality due to Fan et al~\cite{FanGramaLiu}. The reader may note the similarity to \cref{eq:var-chernoff}.
\begin{definition}
	We denote by $\mathcal{C}: (-1,\infty) \to [0,\infty)$ the function given by $\cC(x)=(x+1)\ln (x+1)-x$.
\end{definition}
\begin{theorem}[Fan et al.~\cite{FanGramaLiu}]\label{martingalebennettthm}
    Let $\sigma>0$ be given. Let $(X_i,\mathcal{F}_{i})_{i=0}^r$ be a martingale difference such that almost surely $|X_i|\leq 1$ for all $i \in \{1, \dots, r\}$ and $\sum_{i=1}^r \E {X_i^2 \mid \mathcal{F}_{i-1}}\leq \sigma^2$. Writing $X=\sum_{i=1 }^r X_i$, it holds for any $t\geq 0$  that
    $$
    	\Pr\left[X\geq t\right]\leq e^t\cdot \left( \frac{\sigma^2}{\sigma^2+t}\right)^{\sigma^2+t}.
	$$
\end{theorem}
Simple calculations yield the following corollary.
\begin{corollary}\label{martingalebennettcor}
    Suppose that $(X_i,\mathcal{F}_{i})_{i=0}^r$ is a martingale difference and there exist $M,\sigma\geq 0$ such that $|X_i|\leq M$ for all $i\in \{1, \dots, r\}$ and $\sum_{i=1}^r \E {X_i^2 \mid \mathcal{F}_{i-1}}\leq \sigma^2$. Define $X=\sum_{i=1 }^r X_i$. For any $t\geq 0$ it holds that 
    $$
    	\Pr\left[X\geq t\right]\leq  \exp\left(-\frac{\sigma^2}{M^2} \mathcal{C}\left(\frac{tM}{\sigma^2} \right)\right),
    $$
    where $\mathcal{C}(x)=(x+1)\ln (x+1)-x$.
\end{corollary}

Finally, we present three lemmas describing the asymptotic behavior of $\mathcal{C}$. We omit the proofs of the first two since the results are standard and follow by elementary calculus.
\begin{lemma}\label{rateofgrowth}
    For any $x\geq 0$
    \[
        \frac{1}{2}x \ln(x+1)\leq \mathcal{C}(x)\leq x \ln(x+1) \; .
    \]
    For any $x\in [0,1]$
    \[
        \frac{1}{3}x^2 \leq \mathcal{C}(x) \leq \frac{1}{2}x^2 \; ,
    \]
    where the right hand inequality holds for all $x \ge 0$.
\end{lemma} 
\begin{lemma}\label{lem:Benn-func-consts}
    For any $a \ge 0$. If $b \ge 1$ then
    \[
        b\C(a) \le \C(ab) \le b^2\C(a) \; .
    \]
    If $0 \le b \le 1$ then
    \[
        b^2\C(a) \le \C(ab) \le b\C(a) \; .
    \]
\end{lemma}
Note that as a corollary, if $b=\Theta(1)$ and $a\geq 0$, then $\mathcal{C}(ba)=\Theta(\C(a))$. The final lemma shows that the bound of~\Cref{martingalebennettcor} only gets worse when $\sigma^2$ or $M$ is replaced by some larger number.
\begin{lemma}\label{lem:Benn-func-var-and-max}
    Let $a \ge 0$ be given. On $\R^+$, the following two functions are decreasing
    \begin{align*}
        x &\mapsto x\C\left(\frac{a}{x}\right) \; , \\
        x &\mapsto \frac{\C(ax)}{x^2} \; .
    \end{align*}
\end{lemma}
\begin{proof}
Let $0<x\leq y$ be given. We then observe that the first function is indeed decreasing since by the first bound of \cref{lem:Benn-func-consts}, $x\C(a/x)=x\C\left((a/y)\cdot(y/x)\right)\geq y\C\left(a/y\right)$. That the second function is decreasing follows from a similar argument.\qedhere
\end{proof}

%!TEX root = ../Tabulation.tex

\section{Analysis of Simple Tabulation}\label{sec:simple}
In this section, we analyze the simple tabulation hashing scheme. The section is divided in three parts. First, there will be an introductory section regarding simple tabulation hashing and associated notation. Second, we shall prove the sum of squares result (\cref{eq:squares}). The final section presents a proof of \cref{thm:simpleConcentration}. In order to make the exposition slightly simpler and more accessible, we postpone the argument that our concentration bounds are query invariant to~\Cref{sec:queryinvariance}.

    \subsection{Simple Tabulation Basics}
    Simple tabulation hashing as introduced by Zobrist~\cite{zobrist70hashing} is defined as follows. 
    \begin{definition}[Simple Tabulation Hashing]
    	Let $\Sigma$ be an alphabet, $c\geq 1$ an integer, and $m=2^k, k>0$, a power of two. A simple tabulation hash function, $h\colon \Sigma^c\to [m]$, is a random variable taking values in the set of functions from $\Sigma^c$ to $[m]$, chosen with respect to the following distribution. For each $j\in \{1, \dots, c\}$, let $h_j\colon \Sigma\to [m]$ be a fully random hash function, in other words, a uniformly random function from $\Sigma$ to $[m]$. We evaluate $h$ on the key $x=(x_1, \dots, x_c)\in \Sigma^c$ by computing $h(x) = \bigoplus_{j=1}^ch_j(x_j)$, where $\oplus$ denotes bitwise XOR.
    \end{definition}
    Now, towards analyzing simple tabulation hashing, we add the following notation.
    \begin{definition}[Position Character]
    	Let $\Sigma$ be an alphabet and $c\geq 1$ an integer. We call an element $\alpha=(a, y)\in \{1, \dots, c\}\times \Sigma$ a \emph{position character} of $\Sigma^c$.
    \end{definition}
    Let $h\colon \Sigma^c \to [m]$ be a simple tabulation hash function. We may consider a key $x=(x_1, \dots, x_c)\in \Sigma^c$ as a set of $c$ position characters, $\{(1, x_1), \dots, (c, x_c)\}\subseteq \{1, \dots, c\}\times \Sigma$. Recall that $h(x)=\bigoplus_{i=1}^c h_i(x_i)$ for uniformly random functions $h_i\colon \Sigma\to [m]$. For a position character $\alpha=(a, y)\in \{1, \dots, c\}\times \Sigma$, we may overload notation and write $h(\alpha)=h_a(y)$. Extending this, for a set of position characters $A=\{\alpha_1, \dots, \alpha_n\}\subseteq \{1, \dots, c\}\times \Sigma^c$, $h(A)=\bigoplus_{i=1}^nh(\alpha_i)$. Note that this agrees with the correspondence between keys of $\Sigma^c$ and sets of position characters mentioned before, since for $x=(x_1, \dots, x_c)\in \Sigma^c$, $h(x) = h(\{(1, x_1), \dots, (c, x_c)\})$. If finally $A,B\subset \{1,\dots,c\}\times \Sigma$ are sets of position characters we write $A\oplus B$ for the symmetric difference between $A$ and $B$, i.e., $A \oplus B=(A\setminus B) \cup (B \setminus A)$. We note that for a simple tabulation hash function $h$, $h(A \oplus B)=h(A) \oplus h(B)$.

    %For a multiset of position characters, $A\subseteq \{1, \dots, c\}\times \Sigma$, we occasionally view $A$ as a vector of $\{0, 1\}^{\{1, \dots, c\}\times \Sigma}$ with a 1 at entry $(a, y)\in \{1, \dots, c\}\times \Sigma$ if and only if $(a, y)$ occurs an odd number of times in $A$. From the definition of $h$ it follows that, if $A$ and $A'$ are multisets of position characters where each position character $\alpha\in \{1, \dots, c\}\times \Sigma$ occurs with the same parity, then $h(A)=h(A')$, such that the above identification makes sense. In this setting, for two multisets $A=\{\alpha_1,\dots,\alpha_r\}$ and $B=\{\beta_1,\dots,\beta_s\}$, with $A,B\subseteq \{1, \dots, c\}\times \Sigma$, we may define $A\oplus B\in \{0, 1\}^{j}$ to be the vector corresponding to the multiset $\{\alpha_1,\dots,\alpha_r,\beta_1,\dots,\beta_s\}$. We shall sometimes write $A=\emptyset$, if the vector corresponding to $A$ is the zero-vector. \todo{Are we using this identification? Isn't it better to think of $\oplus$ as symmetric difference $\Delta$ when keys are though of as sets of  position characters?}

    \begin{definition}[Projection Onto an Index]
    	Let $c\geq 1$ be an integer and $i\in \{1, \dots, c\}$ be given. We denote by $\pi_i\colon \Sigma^c\to \{1, \dots, c\}\times \Sigma$ the projection onto the $i$th coordinate given by $\pi_i(x_1, \dots, x_c) = (i, x_i)$, i.e., projecting a key $x$ to its $i$th position character. We extend this to sets of keys, such that for $S\subseteq \Sigma^c$, $\pi_i(S) = \setbuilder{\pi_i(x)}{x\in S}$.
    \end{definition}
    
    The following lemma by Thorup and Zhang~\cite{thorup12kwise} describes the independence of sets of position characters of $\Sigma^c$ under a simple tabulation function $h\colon \Sigma^c \to [2^r]$. We provide a proof for completeness.
            \begin{lemma}[Thorup and Zhang~\cite{thorup12kwise}]\label{lem:LinIndep}
        Let $h\colon \Sigma^c\to [2^r]$ be a simple tabulation hash function. For each $i\in \{1, \dots, t\}$, let $s_i\subseteq \{1, \dots, c\}\times \Sigma$ be a set of position characters of $\Sigma^c$. Let $j\in \{1, \dots, t\}$. If every subset of indices $B\subseteq\{1, \dots, t\}$ containing $j$ satisfies $\bigoplus_{i\in B}s_i \neq \emptyset$, then the distribution of $h(s_j)$ is independent of the joint distribution $(h(s_i))_{i\neq j}$. 
            \end{lemma}
    
    \begin{proof}
    Let $\FF_2$ be the field $\Z/2\Z$ and $V$ the $\FF_2$-vector space $\FF_2^{\{1,\dots,c\}\times \Sigma}$. For a set of position characters $A$, we define $v_A\in V$ as follows: For $(a,y)\in \{1,\dots,c\} \times \Sigma$ we let $v_A(a,y)=1$ if and only if $(a,y)\in A$, and $v_A(a,y)=0$ otherwise. Picking a random simple tabulation hash function $h:\Sigma^c \to [2^r]$ is equivalent to picking a random \emph{linear} function $h': V \to [2^r]$. Here $[2^r]$ is identified with the $\FF_2$-vector space $\FF_2^r$. Indeed, $(v_{\{\alpha\}})_{\alpha \in \{1,\dots,c\}\times \Sigma}$ forms a basis for $V$, and choosing a random linear map $h':V \to[2^r]$ can be done by picking independent and uniformly random values for $h'$ on the basis elements, and extending by linearity. To define $h$ from $h'$, we simply put $h(x)=\bigoplus_{\alpha \in x} h'(v_{\{\alpha\}})$ for a key $x\in \Sigma$ viewed as a set of position characters. Conversely, a simple tabulation hash function $h:\Sigma^c \to [2^r]$ uniquely extends to a linear map $h': V \to [2^r]$. Now under this identification, the condition in the lemma is equivalent to $v_{s_j}$ being linearly independent of the vectors $(v_{s_i})_{i \neq j}$. As $h'$ is a random linear map, it follows by elementary linear algebra that $h'(v_{s_j})=h(s_j)$ is independent of the joint distribution $(h'(v_{s_i}))_{i \neq j}=(h(s_i))_{i \neq j}$, as desired.

        %Let $\FF$ be the field $\Z/2\Z$ and $V$ the vector space $[m]=(\Z/2\Z)^{\log m}$ over $\FF$. As hinted in the preliminaries, we can view each $s_i$ as a vector $v_i\in \FF^{c\cdot \abs \Sigma}$ with a 1 at position $(a, \alpha)$ if and only if $(a, \alpha)$ is contained in $s_i$ an odd number of times. Further, we consider $h$ to be the uniform distribution over $V^{c\cdot \Sigma}$. Now, considering $v_i$ as a linear transformation $V^{c\cdot \Sigma}\to V$ by viewing it as a matrix, $v_i\in \F^{1\times (c\cdot \Sigma)}$, evaluation of $h$ on $s_i$ is performed as $v_i h$. It is easy to see that the distribution of $v_ih$ is indeed the same as the distribution of $h(s_i)$.

        %The lemma is now equivalent to the claim that if for some index $j\in \{1, \dots, t\}$, $v_j$ is independent as a vector in $\FF^{c\cdot \abs\Sigma}$ of the vectors $(v_i)_{i\neq j}$ then the distribution of $v_jh$ is independent of the joint distribution $(v_ih)_{i\neq j}$. However, this is an elementary fact of linear algebra and the conclusion follows.
    \end{proof}
    \subsection{Bounding the Sum of Squared Deviations}
In the following section we shall prove the bound \eqref{eq:squares} of~\Cref{thm:simpleConcentration} from~\cref{sec:combine_simple}, stated independently here as \cref{thm:sum-squares}. It is a technical, albeit crucial, step on the way to proving \cref{thm:simpleConcentration} itself. The foundation of the proof of \cref{thm:sum-squares} is a series of combinatorial observations regarding simple tabulation hashing. 

 Recall from \cref{sec:techsimple} our general proof strategy when proving concentration bounds for simple tabulation hashing. For a set of keys $S\subseteq \Sigma^c$ to be hashed, we fix an ordering of the position characters of $\Sigma^c$. We then fix the hash table entries corresponding to the position characters one at a time according to this ordering. Crucial to the success of this strategy is fixing an ordering where each position character ``decides'' only a small part of the final outcome.
\begin{definition}[Group of Keys]
	Let $S\subseteq \Sigma^c$ be a set of keys and $A=\setbuilder{\alpha\in x}{x\in S}$ be the set of position characters of the keys of $S$. For an enumeration or ordering of the position characters of $A$ as $\{\alpha_1, \dots, \alpha_r\}=A$, we denote by $G_i\subseteq S$ the $i$th \emph{group of keys} with respect to $S$ and the ordering of the position characters. The set is given by $G_i = \setbuilder{x\in S}{\{\alpha_i\}\subseteq x\subseteq \{\alpha_1, \dots, \alpha_i\}}$.
\end{definition}
Put in other words, let $\prec$ denote the ordering on $A$, let $x$ be a key of $S$, and let $\beta_1, \dots, \beta_c$ be the position characters of $x$ such that $\beta_1\prec \beta_2\prec\dots \prec \beta_c$, i.e., $\beta_c$ is last in the ordering of $A$. Then $x\in G_i$ if and only if $\alpha_i=\beta_c$. In relation to simple tabulation, this has the following meaning. In the proof, we shall fix the values $h(\alpha_j)$ one at a time starting at $j=1$ and ending at $j=r$. For every $x\in G_i$, the value of $h(x)$ is then undecided before $h(\alpha_i)$ is known, but is known once $h(\alpha_1), \dots, h(\alpha_i)$ are all fixed.
In analyzing the contribution of each group to the final outcome of the process, we start by proving a generalization of a result from \cite{patrascu12charhash}. It says that if we assign each key a weight, it is always possible to choose the ordering of the position characters such that the total weight of each group is relatively small. The original lemma simply assigned weight 1 to every key.
\begin{lemma}\label{lem:groups}
    Let $S \subseteq \Sigma^c$ be given and let $A = \setbuilder{\alpha \in x}{x \in S}$
    be the position characters of the keys of $S$. Let
    $w \colon \Sigma^c \to \R_{\ge 0}$ be a weight function. Then there exists
    an ordering of the position characters,
    $\{\alpha_1, \ldots, \alpha_r\}= A$ such that for every $i\in \{1, \dots, r\}$,
    the group
    $
        G_i = \setbuilder{x \in S}{\set{\alpha_i} \subseteq x \subseteq \set{\alpha_1, \ldots \alpha_i}}
    $
    satisfies
    \[
        \sum_{x \in G_i} w(x)
            \le \left(\max_{x \in S} w(x) \right)^{1/c} \left(\sum_{x \in S} w(x)\right)^{1 - 1/c} \; .
    \]
\end{lemma}
\begin{proof}
    We define the ordering recursively and backwards as $\alpha_r, \dots, \alpha_1$.
    Let $T_i = A \setminus \set{\alpha_{i+1}, \dots, \alpha_r}$ and
    $S_i = \setbuilder{x \in S}{x \subseteq T_i}$.
    We prove that we can find an $\alpha_i \in T_i$ such that
    \[
        G_{i} = \setbuilder{x \in S_i}{\alpha_i \in x} \; ,
    \]
    satisfies
    \[
        \sum_{x \in G_i} w(x)
            \le \left(\max_{x \in S_i} w(x) \right)^{1/c} \left(\sum_{x \in S_i} w(x)\right)^{1 - 1/c} \; ,
    \]
    which will establish the claim. Let $B_k$ be the set of position characters
    at position $k$ contained in $T_i$, i.e., $B_k = \set{(k, y) \in T_i}=\pi_k(T_i)$.
    Then as $\prod_{k=1}^c \abs{B_k} \geq \abs{S_i}$, we have $\abs{B_k} \geq \abs{S_i}^{1/c}$ for some $k$.

    Since each key of $S_i$ contains at most one position character from $B_k$,
    we can choose $\alpha_i$ such that
    \[
        \sum_{x \in G_i} w(x)
            \le \frac{\sum_{x \in S_i} w(x)}{\abs{B_k}}
            \le \frac{\sum_{x \in S_i} w(x)}{\abs{S_i}^{1/c}}
            % \le \left(\frac{\sum_{x \in S_i}}{\abs{S_i}}\right)^{1/c} \left(\sum_{x \in S_i} w(x)\right)^{1 - 1/c}
            \le \left(\max_{x \in S_i} w(x) \right)^{1/c} \left(\sum_{x \in S_i} w(x)\right)^{1 - 1/c} \; .
    \]
\end{proof}

    % \begin{lemma}\label{lem:Groups}
    %     Let $S\subseteq \Sigma^c$ be given and let $A=\setbuilder{\alpha\in x}{x\in S}$ be the position characters of the keys of $S$. Then there exists an ordering of the position characters of $A$, $\alpha_1, \dots, \alpha_{r} \in A, r=\abs{A}$ such that for every $i$, the set
    %     $$G_{i}=\{ x\in S\mid \{\alpha_i\}\subseteq x \subseteq \{\alpha_1, \dots, \alpha_i\} \}$$
    %     satisfies $\norm{G_i}_2\leq \norm{S}_\infty\cdot  \norm{S}_2^{1-1/c}$.
    % \end{lemma}
    % \begin{proof}
    %     We define the ordering recursively and backwards as $\alpha_r, \dots, \alpha_1$. Let
    %     $T_i= A\setminus \{\alpha_{i+1}, \dots, \alpha_r\}$.
    %     We prove that we can find an $\alpha_i\in T_i$ such that
    %     $$G_{i}=\{ x\in S\mid \{\alpha_i\}\subseteq x \subseteq \{\alpha_1, \dots, \alpha_i\} \}$$
    %     satisfies $\norm{G_i}_2\leq \norm{S}_\infty\cdot  \norm{S}_2^{1-1/c}$, which will establish the claim. Let $B_k$ be the set of position characters at position $k$ contained in $T_i$, i.e.\ $B_k = \{(k, y)\in T_i\}$. Further, let $X'\subseteq S$ be the keys of $S$ with position characters only from $T_i$, i.e.\ $X'=\setbuilder{x\in S}{x\subseteq T_i}$. Then as $\prod_{k=1}^c\abs{B_k}\geq \abs X'$, we have $\abs{B_k}\geq \abs{X}^{1/c}$ for some $k$.

    %     Since each key of $X'$ contains at most one position character from $B_k$, we can choose $\alpha_i$ such that $$\norm{G_i}_2^2\leq \norm{X'}_2^2/\abs{B_k}\leq \norm{X'}_2^2/\abs{X}^{1/c}\leq \norm{X'}_2^{2(1-1/c)}\cdot \norm{X'}_\infty^{2/c}\leq \left(\norm{S}_\infty\cdot \norm{S}_2^{1-1/c}\right)^2.$$
    % \end{proof}

 Suppose we have keys $x_1, \dots, x_t\in \Sigma^c$. It follows as a corollary of  \cref{lem:LinIndep} that with a simple tabulation hash function $h\colon \Sigma^c\to [m]$, the values $h(x_1), \dots, h(x_t)$ are completely independent if and only if there does not exist a subset of indices $B\subseteq \{1, \dots, t\}$ with $\bigoplus_{i\in B} x_i=\emptyset$. In this vein, it turns out to be natural, given sets of keys $A_1, \dots, A_{\ell}\subseteq \Sigma^c$, to bound the number of tuples $x_1\in A_1, \dots, x_{\ell}\in A_{\ell}$ with $\bigoplus_{i=1}^{\ell}x_i = \emptyset$. This is the content of Lemma 3 of \cite{DKRT15:k-part}. We prove the following generalization of this result, which deals with weighted keys. Note that the statements would be identical if each key was assigned the weight 1.
    \begin{lemma}\label{lem:CollBound}
            Let $\ell\in \N$ be even, $w_1,\dots,w_\ell\colon \Sigma^c\to \R$ be weight functions, and $A_1, \dots, A_\ell\subseteq \Sigma^c$ be sets of keys. Then
            \begin{align*}
                \sum_{\substack{x_1\in A_1, \dots, x_\ell\in A_\ell \\ \bigoplus_{k=1}^\ell x_k = \emptyset}} \prod_{k=1}^\ell w_k(x_k)\leq ((\ell-1)!!)^c\cdot \prod_{k=1}^\ell \sqrt{\sum_{x\in A_k}w_k(x)^2}.
            \end{align*}
    \end{lemma}
    \begin{proof}
        For every $(x_1, \dots, x_\ell)\in A_1\times\dots \times A_\ell$ satisfying $\bigoplus_{k=1}^\ell x_k = \emptyset$ we have  $\bigoplus_{k=1}^\ell \{\pi(x_k, c)\}=\emptyset$. This implies that each character in the $c$-th position occurs an even number of times in $(x_1, \dots, x_\ell)$. Thus, for any such tuple we can partition the indices $1, \dots, \ell$ into pairs $(i_1, j_1), \dots, (i_{\ell/2}, j_{\ell/2})$ satisfying $\pi(x_{i_k}, c) = \pi(x_{j_k}, c)$ for every $k\in \{1, \dots, \ell\}$. Fix such a partition and let $X\subseteq A_1\times\dots \times A_\ell$ be the set
        \begin{align*}
            X = \{(x_1, \dots, x_\ell)\in A_1\times\ldots \times A_\ell\mid \forall k\in \{1, \dots, \ell/2\}\colon \pi(x_{i_k}, c) = \pi(x_{j_k}, c)\}.
        \end{align*}
        We proceed by induction on $c$.

        For $c=1$, $\pi(x_{i_k}, c) = \pi(x_{j_k}, c)$ implies $x_{i_k}=x_{j_k}$ such that
        \begin{align*}
            X = \{(x_1, \dots, x_\ell)\in A_1\times\ldots \times A_\ell\mid \forall k\in \{1, \dots, \ell/2\}\colon x_{i_k} = x_{j_k}\}.
        \end{align*}
        Thus, by the Cauchy-Schwartz inequality,
        \begin{align*}
            \sum_{(x_1, \dots, x_\ell)\in X}\prod_{k=1}^\ell w_k(x_k) =& \prod_{k=1}^{\ell/2} \sum_{x\in A_{i_k}\cap A_{j_k}}w_{i_k}(x)w_{j_k}(x) \\
            \leq &\prod_{k=1}^{\ell/2} \left(\sqrt{\sum_{x\in A_{i_k}}w_{i_k}(x)^2}\cdot \sqrt{\sum_{x\in A_{j_k}}w_{j_k}(x)^2}\right) \\
            \leq & \prod_{k=1}^\ell \sqrt{\sum_{x\in A_k}w_k(x)^2}.
        \end{align*}
    Since this is true for any partition into pairs, $(i_1, j_1), \dots, (i_{\ell/2}, j_{\ell/2})$, there are exactly $(\ell-1)!!$ such partitions, and every term in the original sum is counted by some partition, we get the desired bound for $c=1$.

        Let $c>1$ and assume that the statement holds when each key has $<c$ characters. 
        For each $a\in \Sigma$ and $k\in \{1, \dots, \ell\}$ define the set
        \begin{align*}
            A_k[a] = \{ x\in A_k\mid \pi(x, c) = a \}.
        \end{align*} 
        Fixing the last character of each pair in our partition by picking $a_1, \dots, a_{\ell/2}\in \Sigma$ and considering the sets $A_{i_k}[a_k]$ and $A_{j_k}[a_k]$, we can consider the keys of $\prod_{k=1}^{\ell/2} A_{i_k}[a_k]\times A_{j_k}[a_k]$ as only having $c-1$ characters, which allows us to apply the induction hypothesis. This yields
        \begin{align*}
                \sum_{\substack{(x_1, \dots, x_\ell)\in X\\ \bigoplus_{k=1}^\ell x_k = \emptyset}} \prod_{k=1}^\ell w_k(x_k) &=
                \sum_{(a_k)_{k=1}^{\ell/2}\in \Sigma^{\ell/2}}
                    \left(
                        \sum_{\substack{(x_{i_k}, x_{j_k})_{k=1}^{\ell/2}\in \prod_{k=1}^{\ell/2} A_{i_k}[a_k]\times A_{j_k}[a_k]\\ \bigoplus_{k=1}^\ell x_k = \emptyset}}
                            \prod_{k=1}^{\ell/2} w_{i_k}(x_{i_k})w_{j_k}(x_{j_k})
                    \right)\\
                    &\leq \sum_{(a_k)_{k=1}^{\ell/2}\in \Sigma^{\ell/2}}
                        \left( 
                            ((\ell-1)!!)^{c-1}\cdot 
                            \prod_{k=1}^{\ell/2} 
                                \left(\sqrt{\sum_{x\in A_{i_k}[a_k]} w_{i_k}(x)^2}\cdot \sqrt{\sum_{x\in A_{j_k}[a_k]} w_{j_k}(x)^2}\right)
                        \right)\\
                    & = ((\ell-1)!!)^{c-1}\cdot\prod_{k=1}^{\ell/2}\left(\sum_{a\in \Sigma} 
                                \left(\sqrt{\sum_{x\in A_{i_k}[a]} w_{i_k}(x)^2}\cdot \sqrt{\sum_{x\in A_{j_k}[a]} w_{j_k}(x)^2}\right)\right)\\
                    &\leq ((\ell-1)!!)^{c-1}\cdot \prod_{k=1}^{\ell/2}\left( \sqrt{\sum_{a\in \Sigma} \sum_{x\in A_{i_k}[a]} w_{i_k}(x)^2}\cdot \sqrt{\sum_{a\in \Sigma}\sum_{x\in A_{j_k}[a]} w_{j_k}(x)^2} \right)\\
                    &=((\ell-1)!!)^{c-1}\cdot\prod_{k=1}^{\ell/2} 
                                \left(\sqrt{\sum_{x\in A_{i_k}} w_{i_k}(x)^2}\cdot \sqrt{\sum_{x\in A_{j_k}} w_{j_k}(x)^2}\right),
            \end{align*}
            where the last inequality follows from the Cauchy-Schwartz inequality.
            Since the indices can be partitioned into pairs in $(\ell-1)!!$ ways, the same argument as in the induction start yields
            \begin{align*}
                \sum_{\substack{x_1\in A_1, \dots, x_\ell\in A_\ell \\ \bigoplus_{k=1}^\ell x_k = \emptyset}} \prod_{k=1}^\ell w_k(x_k)\leq ((\ell-1)!!)^c\cdot \prod_{k=1}^\ell \sqrt{\sum_{x\in A_k}w_k(x)^2},
            \end{align*}
            which was the desired conclusion.
    \end{proof}
The following rather technical lemma bounds the moments of collisions between sets of keys. However, we shall dwell on it for a moment as it reflects considerations that will come up repeatedly going forward. Consider a simple tabulation function $h\colon \Sigma^c\to [m]$ and a value function $v\colon \Sigma^c\times [m]\to \R$. Hashing the keys of some subsets $A_1, \dots, A_n\subseteq \Sigma^c$ into $[m]$ using $h$, we are interested in the sums $X_i=\sum_{x\in A_i}v(x, h(x))$ for $1\leq i\leq n$ and, in particular, in properties of the joint distribution $(X_1, \dots, X_n)$. Here, the actual values of $X_i$ are not as important as how much $X_i$ deviates form its mean. For $1\leq i\leq n$, We thus consider the variables
$$ Y_i = X_i-\E{X_i} = \sum_{x\in A_i}\sum_{b\in [m]}v(x, b)\left( [h(x) = b]-\frac{1}{m} \right),$$
and for a level of generality required for proving the main theorems of this section, we consider the \emph{shifted} variables 
$$Y_i^{(j)} = \sum_{x\in A_i}\sum_{b\in [m]}v(x, b)\left( [h(x) = j\xor b]-\frac{1}{m} \right),$$
for $j\in [m]$, corresponding to shifting the hash function $h$ by $j\in [m]$.
\begin{lemma}\label{lem:moment}
    Let $h\colon \Sigma^c\to [m]$ be a simple tabulation hash function and
    $v\colon \Sigma^c \times [m] \to \R$ a value function. Let
    $Q = \setbuilder{i \in [m]}{\exists x \in \Sigma^c : v(x, i) \neq 0}$ be the
    support of $v$ and write $\ell = \abs{Q}$. Let $n \in \N$ and $A_1, \ldots, A_n \subseteq \Sigma^c$.
    For every $i \in \{1, \dots, n\}$ and $j \in [m]$ define the random variable
    \[
        Y_i^{(j)} = \sum_{x \in A_i} \sum_{b \in Q} v(x, b) \left( [h(x) = j \oplus b] - \frac{1}{m} \right) \; ,
    \]
    and set
    \[
        T = \sum_{\substack{j_1, \ldots, j_n \in [m]\\ \bigoplus_{j =1}^n j_k = 0}} \prod_{k =1}^n Y_k^{(j_k)} \; .
    \]
    Then for every constant $t \in \N$,
    \[
        \abs{\E{T^t}} = O_{t,n,c}\left( \ell^{tn/2} \prod_{k =1}^n \left(
            \sum_{x \in A_k} \sum_{b \in Q} v(x, b)^2
        \right)^{t/2} \right) \; .
    \]
\end{lemma}
\begin{proof}
    We rewrite $T$ as follows
    \begin{align*}
        T
            &= \sum_{\substack{j_1, \ldots, j_n \in [m]\\ \bigoplus_{k =1}^n j_k = 0}}
                \prod_{k =1}^n Y_k^{(j_k)}
            \\&= \sum_{\substack{j_1, \ldots, j_n \in [m]\\ \bigoplus_{k =1}^n j_k = 0}}
                \left(\prod_{k =1}^n \sum_{x \in A_k} \sum_{b \in Q} v(x, b) \left(
                        [h(x) = j_k \oplus b] - \frac{1}{m}
                    \right)
                \right)
            \\&= \sum_{(x_1, \ldots, x_n)\in A_1 \times \ldots \times A_n}
                \sum_{b_1, \ldots b_n \in Q} \left(
                    \sum_{\substack{j_1, \ldots, j_n \in [m]\\ \bigoplus_{k =1}^n j_k = 0}} \left(
                        \prod_{k =1}^n \left(
                            v(x_k, b_k) \left([h(x_k) = j_k \oplus b_k] - \frac{1}{m}\right)
                        \right)
                    \right)
                \right)
            \\&= \sum_{(x_1, \ldots, x_n)\in A_1 \times \ldots \times A_n}
                \sum_{b_1, \ldots b_n \in Q} \left(
                    \left(\prod_{k =1}^n v(x_k, b_k) \right)
                    \cdot \left(
                        \sum_{\substack{j_1, \ldots, j_n \in [m]\\ \bigoplus_{k =1}^n j_k = 0}} \left(
                            \prod_{k =1}^n \left(
                                [h(x_k) = j_k \oplus b_k] - \frac{1}{m}
                            \right)
                        \right)
                    \right)
                \right)
            \\&= \sum_{(x_1, \ldots, x_n)\in A_1 \times \ldots \times A_n}
                \sum_{b_1, \ldots b_n \in Q} \left(
                    \left(\prod_{k =1}^n v(x_k, b_k) \right)
                    \cdot \left(
                        \left[\bigoplus_{k =1}^n h(x_k) = \bigoplus_{k =1}^n b_k\right] - \frac{1}{m}
                    \right)
                \right)
        %     \\&= \sum_{(x_1, \dots, x_\ell)\in A_1\times\dots\times A_k} \left( \left(  \prod_{k=1}^\ell w_{x_k} \right) \cdot \sum_{\substack{j_1, \dots, j_\ell\in [m]\\ \bigoplus_{k=1}^\ell j_k = 0}} \left( \prod_{k=1}^\ell \left([h(x_k)=j_k]-\frac1m\right) \right)\right)\\
        % &= \sum_{(x_1, \dots, x_\ell)\in A_1\times\dots\times A_k} \left( \left(  \prod_{k=1}^\ell w_{x_k} \right) \cdot  \left(\left[\bigoplus_{k=1}^\ell h(x_k)=0\right]-\frac1m\right)\right).
    \end{align*}
    Here the last equality is derived by observing that for fixed
    $(x_1, \ldots, x_n) \in A_n \times \ldots \times A_n$ and fixed
    $b_1, \ldots b_n \in Q$,
    \begin{align*}
        \sum_{\substack{j_1, \ldots, j_n \in [m]\\ \bigoplus_{k =1}^n j_k = 0}} \left(
            \prod_{k =1}^n \left(
                [h(x_k) = j_k \oplus b_k] - \frac{1}{m}
            \right)
        \right)
            = \sum_{\substack{j_1, \ldots, j_n \in [m]\\ \bigoplus_{k =1}^n j_k = 0}} \left(
                    \sum_{B \subseteq \{1, \dots, n\}} (-m)^{-(n - \abs{B})} \prod_{k \in B} [h(x_k) = j_k \oplus b_k]
            \right)
        % \sum_{\substack{j_1, \ldots, j_n \in [m]\\ \bigoplus_{k \in [n]} j_k = 0}} \left(
        %         \prod_{k=1}^\ell \left([h(x_k)=j_k]-\frac1m\right)\right) &= \sum_{\substack{j_1, \dots, j_\ell\in [m]\\ \bigoplus_{k=1}^\ell j_k = 0}} \left( \sum_{B\subseteq [m]} (-m)^{-(\ell-\abs B)}
        % \prod_{k\in B} [h(x_k)=j_k]\right)
    \end{align*}
    and since for $\emptyset \subseteq B \subsetneq \{1, \dots, n\}$ there are exactly
    $m^{n - \abs{B}-1}$ tuples $(j_1, \dots, j_n) \in [m]^n$ satisfying $j_k \oplus b_k = h(x_k)$
    for every $k \in B$ and $\bigoplus_{k=1}^n j_k = 0$, we get
    \begin{align*}
        \sum_{\substack{j_1, \ldots, j_n \in [m]\\ \bigoplus_{k =1}^n j_k = 0}} \left(
            \prod_{k =1}^n \left(
                [h(x_k) = j_k \oplus b_k] - \frac{1}{m}
            \right)
        \right)
        &= \sum_{\substack{j_1, \ldots, j_n \in [m]\\ \bigoplus_{k =1}^n j_k = 0}}\left( \prod_{k =1}^n\left[ h(x_k) = j_k\oplus b_k \right] \right) + \frac{1}{m}\sum_{B \subsetneq \{1, \dots, n\}} (-1)^{n - \abs{B}} \\
            &= \left[\bigoplus_{k =1}^n h(x_k) = \bigoplus_{k =1}^n b_k\right]
                + \frac{1}{m}\sum_{B \subsetneq \{1, \dots, n\}} (-1)^{n - \abs{B}}.
    \end{align*}
    By the principle of inclusion-exclusion, the last term is $-\frac{1}{m}$, which concludes the rearrangement.

    Write $S= A_1 \times \ldots \times A_n$ and let $f\colon S \to \R$ be the function
    \[
        f(x_1, \dots, x_n) = \sum_{b_1, \ldots b_n \in Q} \left(
            \left(\prod_{k =1}^n v(x_k, b_k) \right)
            \cdot \left(
                \left[\bigoplus_{k =1}^n h(x_k) = \bigoplus_{k =1}^n b_k\right] - \frac{1}{m}
            \right)
        \right) \; .
    \]
    By the above rearrangement, we have
    $T^t = \sum_{(s_i)_{i \in [t]} \in S^t} \prod_{i=1\in [t]} f(s_i)$, such that,
    \begin{align*}
        \E{T^t} = \sum_{(s_i)_{i =1}^t \in S^t} \E{\prod_{i =1}^t f(s_i)} \; .
    \end{align*}
    Now, for a $t$-tuple $(s_i)_{i =1}^t\in S^t$, we overload notation by for a subset $T\subseteq\{1,\dots,t\}$ defining $\bigoplus_{i\in T}s_i=\bigoplus_{i\in T}\bigoplus_{j=1}^n(s_i)_j$, where we still think of the keys $(s_i)_j$ as sets of input characters, and where $\oplus$ is the symmetric difference.
    Let $(s_i)_{i =1}^t\in S^t$ and let $T_1, \dots, T_r\subseteq \{1, \dots, t\}$
    be all subsets of indices satisfying $\bigoplus_{i \in T_j}s_i = \emptyset, 1\leq j\leq t$.
     If for some $i \in \{1, \dots, t\}$, $i \not\in \bigcup_{j =1}^r T_i$ then by
    Lemma \ref{lem:LinIndep}, $h(s_i)$ is independent of the joint distribution
    $(h(s_j))_{j \neq i}$ and uniformly distributed in $[m]$. It follows that $f(s_i)$ is independent of the joint distribution
    $(f(s_j))_{j \neq i}$. Since it further holds that $\E{f(s_i)}=0$, this implies
    \begin{align*}
        \E{\prod_{j=1}^t f(s_j)} = \E{f(s_i)}\cdot \E{\prod_{j \neq i} f(s_j)} = 0 \; .
    \end{align*}
    Hence, we shall only sum over the $t$-tuples $(s_i)_{i=1}^t\in S^t$ satisfying that there exist
    subsets of indices $T_1, \dots, T_r \subseteq \{1, \dots, t\}$ such that
    $\bigoplus_{i \in T_j}s_i = \emptyset$ for every $j \in \{1, \dots, r\}$ and $\bigcup_{j =1}^r T_j = \{1, \dots, t\}$.

    Fix such subsets $T_1, \dots, T_r\subseteq \{1, \dots, t\}$ and for $i \in \{1, \dots, r\}$ let
    $B_i = T_i \setminus \left(\bigcup_{j < i} T_j\right)$. Then we can write
    \begin{align}
    \nonumber
    &\sum_{
            \substack{(s_i)_{i=1}^t \in S^t
            \\ \forall j \in \{1, \dots, r\}\colon \bigoplus_{i \in T_j} s_i = \emptyset}
        } \prod_{i=1}^t f(s_i)  \\
        \label{eq:firststep}
                =  &\sum_{
            \substack{(s_i)_{i\in \{1,\dots,t\}\setminus B_r} \in S^{t-|B_r|}
            \\ \forall j \in \{1, \dots, r-1\}\colon \bigoplus_{i \in T_j} s_i = \emptyset}
        } \prod_{i\in \{1,\dots,t\}\setminus B_r} f(s_i)
        \sum_{
            \substack{(s_i)_{i\in B_r} \in S^{|B_r|}
            \\ \bigoplus_{i \in B_r} s_i = \bigoplus_{i \in T_r\setminus B_r} s_i}
        } \prod_{i\in B_r} f(s_i)
   \end{align}
   Now fix $(s_i)_{i\in \{1,\dots,t\}\setminus B_r} \in S^{t-|B_r|}$ such that for all $j \in \{1, \dots, r-1\}$ it holds that $\bigoplus_{i \in T_j} s_i = \emptyset$. We wish to upper bound the inner sum in~\eqref{eq:firststep} for this choice of $(s_i)_{i\in \{1,\dots,t\}\setminus B_r}$.  In order to do this, observe that for $s = (x_1, \ldots, x_n)\in S$ we always have
    \begin{align*}
        \abs{f(s)}
            \le \sum_{b_1, \ldots b_n \in Q} \prod_{k \in [n]} \abs{v(x_k, b_k)}
            = \prod_{k =1}^n \abs{\sum_{b \in Q} v(x_k, b)}
            \le \prod_{k=1}^n \sqrt{\ell \sum_{b \in Q} v(x_k, b)^2} \; ,
    \end{align*}
by the QA-inequality.
We now wish to combine this bound with~\Cref{lem:CollBound} to obtain a bound on the inner sum in~\eqref{eq:firststep}. For this, we define let $\ell=|T_r|n$ and define sets of keys $F_1,\dots,F_\ell$ and weight functions $w_1,\dots,w_\ell:\Sigma^c \to \R$ as follows. Enumerate $T_r=\{i_1,\dots,i_{|T_r|}\}$ such that $\{i_1,\dots,i_{|B_r|}\}=B_r$. Now for $0\leq k<|B_r|$ and $1\leq j \leq n$ we define $F_{kn+j}=A_j$. We further define the weight function $w_{kn+j}:\Sigma^c\to \R$ by $w_{kn+j}(x)=\sqrt{\ell \sum_{b \in Q} v(x, b)^2}$ for $x\in \Sigma^c$. Observe that these weight functions are all identical. Secondly, for $|B_r|\leq k<|T_r|$ and $1\leq j \leq n$, we define $F_{kn+j}=\{s_{i_k}(j)\}$, and $w_{kn+j}:\Sigma^c \to \R$ by $w_{kn+j}(x)=1$ for all $x\in \Sigma^c$. Then,
$$
\abs{
\sum_{
            \substack{(s_i)_{i\in B_r} \in S^{|B_r|}
            \\ \bigoplus_{i \in B_r} s_i = \bigoplus_{i \in T_r\setminus B_r} s_i}
        } \prod_{i\in B_r} f(s_i)
        }
        \leq 
        \sum_{\substack{x_1\in B_1, \dots, x_\ell\in B_\ell \\ \bigoplus_{k=1}^\ell x_k = \emptyset}} \prod_{k=1}^\ell w_k(x_k)
        \leq
        (n|T_r|-1)!!)^c
        \prod_{k=1}^n \left(
                \ell\cdot \sum_{x \in A_k} \sum_{b \in Q} v(x, b)^2
            \right)^{\abs{B_r}/2},
$$
where the last inequality follows from~\Cref{lem:CollBound}. Note that this upper bound does not depend on the choice of $(s_i)_{i\in \{1,\dots,t\}\setminus B_r} \in S^{t-|B_r|}$ in the outer sum in~\eqref{eq:firststep}. Repeating this argument another $r-1$ times, and using that $\{1,\dots,t\}$ is the disjoint union of $B_1,\dots,B_r$, we obtain that 
$$
\abs{\sum_{
            \substack{(s_i)_{i=1}^t \in S^t
            \\ \forall j \in \{1, \dots, r\}\colon \bigoplus_{i \in T_j} s_i = \emptyset}
        } \prod_{i=1}^t f(s_i)} 
        \leq 
        ((nt - 1)!!)^{cr} \cdot \prod_{k=1}^n \left(
                \ell \sum_{x \in A_k} \sum_{b \in Q} v(x, b)^2
            \right)^{t/2}.
$$
Since there are at most $2^{2^t}$ ways of choosing $r$ and the subsets $T_1, \dots, T_r$ and since $r\leq 2^t$,
    summing over these choices yields
    \begin{align*}
        \abs{\E{T^t}}
            &\leq 2^{2^t} ((nt - 1)!!)^{cr} \cdot \prod_{k=1}^n \left(
                \ell \sum_{x \in A_k} \sum_{b \in Q} v(x, b)^2
            \right)^{t/2}
            \\&\leq O_{t, n, c}\left( \ell^{nt/2} \prod_{k=1}^n \left(
                \sum_{x \in A_k} \sum_{b \in Q} v(x, b)^2
            \right)^{t/2} \right) \; .
    \end{align*}

\end{proof}

We are now ready to prove the main theorem of the subsection, a bound on the sum of squared deviations of the value function from its deviation when shifting by every $j\in [m]$, the second part of~\Cref{thm:simpleConcentration}. As described in~\Cref{sec:combine_simple}, this bound is an important ingredient in the proof of the first part of~\Cref{thm:simpleConcentration}. Namely, in our inductive proof, it bounds the variance of the value obtained from the keys of one of the groups $G_i$ when the keys from this group are shifted by a uniformly random XOR with $h(\alpha_i)$.
\begin{theorem}\label{thm:sum-squares}
    Let $h\colon \Sigma^c \to [m]$ be a simple tabulation hash function and
    $S \subseteq \Sigma^c$ a set of keys. Let $v\colon \Sigma^c \times [m] \to [-1, 1]$
    be a value function such that the set
    $Q = \setbuilder{i \in [m]}{\exists x \in \Sigma^c : v(x, i) \neq 0}$ satisfies $\abs{Q} \le m^{\eps}$,
    where $\eps < \frac{1}{4}$ is a constant.
    For $j\in [m]$ define the random variable $V_j = \sum_{x\in S}v(x, h(x)\xor j)$ and let $\mu = \E{V_j}$, noting that this is independent of $j$. For any $\gamma\geq 1$,
    \begin{align}
    	\Pr \left[\sum_{j\in [m]}(V_j-\mu)^2>C_{\gamma}^c \sum_{x\in S}\sum_{k\in [m]} v(x, k)^2\right] = O_{\gamma, \eps, c}(n/m^\gamma)
    \end{align}
    where $C_{\gamma} = 3 \cdot 2^6\cdot \gamma^2$ and this bound is query invariant up to constant factors.
\end{theorem}
\begin{proof}
First, note that we may write 
\begin{align}\label{eq:Vj-rewrite}
	V_j-\mu = \sum_{x\in S} \sum_{k\in Q} v(x, k)[h(x)=j\xor k] - \frac{1}{m}\sum_{x\in S}\sum_{k\in Q}v(x, k) = \sum_{x\in S}\sum_{k\in Q}v(x, k)\left( [h(x)=j\xor k]-\frac{1}{m} \right) 
\end{align}

    Now, define $v'(x) = \sum_{k \in Q} v(x, k)^2$ and for $X \subseteq \Sigma^c$ we let
    $v'(X) = \sum_{x \in X} v'(x)$ and define $v'_\infty(X) = \max_{x \in X} v'(x)$.
    Now applying \Cref{lem:groups} with respect to $v'$ we get position characters
    $\alpha_1, \ldots, \alpha_r$ with corresponding groups $G_1, \ldots, G_r$,
    such that, $\cupdot_{i=1}^r G_i = S$ and for every $i \in \{1, \dots, r\}$,
    $v'(G_i) \le v'(S)^{1 - 1/c} v'_{\infty}(S)^{1/c}$. For $i \in \{1, \dots, r\}, j \in [m]$
    we define the random variables
    \begin{align*}
        X_i^{(j)}
            &=\sum_{x \in G_i}\sum_{k \in Q} v(x, k) \left([h(x \setminus \alpha_i) = j \oplus k] - \frac{1}{m}\right) ,
        & Y_i^{(j)} &= X_i^{\left(j\oplus h(\alpha_i)\right)},
    \end{align*}
where we recall that $x\setminus \alpha_i$ denotes the set containing the position characters of $x$ except $\alpha_i$. Notice that by \eqref{eq:Vj-rewrite}, $V_j-\mu = \sum_{i\in [r]}Y_i^{(j)}$. Writing $V=\sum_{j\in [m]}(V_j-\mu)^2 = \sum_{j\in [m]}\left( \sum_{i\in [r]}Y_i^{(j)} \right)^2$, the statement we wish to prove is
    \[
        \PR{V > C_{\gamma}^c v'(S)}
            \le O_{\gamma, c}\left(\abs{S}m^{-\gamma}\right) \; .
    \]
    We proceed by induction on $c$. The induction start, $c=1$, and the induction step are almost identical, so we carry them out in parallel. Note that when $c=1$ each group has size at most 1, i.e.\ $\abs {G_i} \leq 1$ for every $i\in \{1, \dots, r\}$.

    Let $\gamma \ge 1$ be fixed. We write
    \begin{align}\label{sqsum1}
        V = \underbrace{\sum_{j \in [m]} \sum_{i=1}^r\left(Y_i^{(j)}\right)^2}_{V_1} + \underbrace{\sum_{j \in [m]}\sum_{i=1}^rY_i^{(j)}Y_{<i}^{(j)}}_{V_2}
    \end{align}
    and bound $V_1$ and $V_2$ separately starting with $V_1$.

    Interchanging summations, $V_1 = \sum_{i=1}^r \sum_{j \in [m]}\left(Y_i^{(j)}\right)^2$. In the case $c=1$, let $i\in \{1, \dots, r\}$ be given. If $\abs{G_i} = 0$, $\sum_{j \in [m]}\left(Y_i^{(j)}\right)^2=0$. If on the other hand $G_i=\{x_i\}$ for some $x_i\in \Sigma^c$,
    \begin{align*}
        \sum_{j\in [m]}\left(Y_i^{(j)}\right)^2 
        &= \sum_{j\in [m]}\left( \sum_{k\in Q} v(x_i, k) \left(\left[ h(x_i)=j\oplus k \right]-\frac1m\right) \right)^2\\
        &= \sum_{j\in [m]}\left( \sum_{k\in [m]} v(x_i, k) \left(\left[ h(x_i)\oplus j= k \right]-\frac1m\right) \right)^2\\
        &= \sum_{j\in [m]}\left( \sum_{k\in [m]} v(x_i, k) \left(\left[ j= k \right]-\frac1m\right) \right)^2\\
        &= \sum_{j\in [m]}\left( v(x_i, j)-\frac1m\sum_{k\in [m]} v(x_i, k) \right)^2\\
        &\leq \sum_{j\in [m]}v(x_i, j)^2
    \end{align*}
    where the last inequality follows from the inequality $\E{(X-\E X)^2}\leq \E{X^2}$. Thus, we always have $V_1\leq v'(S)\leq \frac{C_{\gamma}^c}2v'(S)$. In the case
    $c > 1$ we observe that the keys of $G_i$ have a common position character.
    Hence, we can apply the induction hypothesis on the keys of $G_i$ with the remaining $c-1$ position characters to conclude that
    \[
        \PR{\sum_{j \in [m]}\left(Y_i^{(j)}\right)^2 > C_{\gamma}^{c-1}v'(G_i)} \le O_{\gamma,c}(|G_i|m^{-\gamma}) \; .
    \]
    By a union bound,
    \begin{align}\label{sqsum2}
        \PR{V_1 > \frac{C_{\gamma}^{c}}2 v'(S)} \le \PR{V_1 > C_{\gamma}^{c-1}v'(S)} \le \sum_{i\in [r]}O_{\gamma,c}(|G_i|m^{-\gamma}) = O_{\gamma,c}(|S|m^{-\gamma}) \; .
    \end{align}

    Next we proceed to bound $V_2$. For $0\leq i\leq r$ define
    $Z_i = \sum_{j \in [m]}Y_i^{(j)}Y_{<i}^{(j)}$ with $Z_0=0$ and
    $\F_i = \sigma((h(\alpha_j))_{j=1}^i)$ with $\F_0=\{\emptyset, \Omega\}$. As $Y_{<i}^{(j)}$ is
    $\F_{i-1}$ measurable for $j \in [m]$ it holds that
    \[
        \EC{Z_i}{\F_{i-1}} = \sum_{j \in [m]} \EC{Y_i^{(j)}}{\F_{i-1}}Y_{<i}^{(j)} = 0 \; ,
    \]
    and so $(Z_i, \F_i)_{i=0}^r$ is a martingale difference. We will define
    a modified martingale difference $(Z'_i, \F_i)_{i=0}^r$ recursively as
    follows: We define the events $A_i, B_i$ and
    $C_i$ for $i \in \{1, \dots, r\}$ as
    \begin{align*}
        A_i
            &= \bigcap_{k=1}^i\left(
                    \sum_{j \in [m]}\left(Y_k^{(j)}\right)^2 \leq C_{\gamma}^{c-1}v'(G_k)
                \right),
        \\
        B_i
            &= \bigcap_{k=1}^i \left(
                \VarC{Z_k}{\F_{k-1}} \leq m^{-1/2}v'(G_k)v'(G_{<k})
            \right),
        \\
        C_i
            &= \left(
                \max_{1\leq k\leq i}\set{Z'_{<k}} \leq \frac{C_{\gamma}^{c}}2 v'(S)
            \right). \;
    \end{align*}
    Finally, we let $Z'_i = [A_i \cap B_i \cap C_i]\cdot Z_i$. Clearly $B_i, C_i \in \F_{i-1}$.
    To see that this is also the case for $A_i$ we note that for $k \leq i$,
    \[
        \sum_{j \in [m]}(Y_k^{(j)})^2
            = \sum_{j \in [m]}(X_k^{(j\oplus h(\alpha_k))})^2
            = \sum_{j \in [m]}(X_k^{(j)})^2 \; ,
    \]
    and as each $X_k^{(j)}$ is $\F_{k-1}$-measurable it follows that $A_i \in \F_{i-1}$.
    Now, as $[A_i \cap B_i \cap C_i]$ is $\F_{i-1}$-measurable,
        $$\EC{Z'_i}{\F_{i-1}} = [A_i \cap B_i \cap C_i]\EC{Z_i}{\F_{i-1}} = 0,$$
     which implies that
    $(Z'_i, \F_i)_{i=0}^r$ is a martingale difference. 
    
    If $A_{r}, B_{r}$,
    and $C_{r}$ all occur then $\sum_{i=1}^r Z_i = \sum_{i=1}^rZ'_i$. In particular
    \[
        \PR{V_2> \frac{C_{\gamma}^{ c}}2v'(S)}=\PR{\sum_{i \in [r]}Z_i > \frac{C_{\gamma}^{ c}}{2} v'(S)}
            \le \PR{\sum_{i=1}^rZ'_i > \frac{C_{\gamma}^{ c}}{2} v'(S)}
                + \PR{A_{r}^c \cup B_{r}^c \cup C^c_{r}} \; .
    \]
    If $C_r$ does not occur then $\sum_{i \in [r]}Z'_i > \frac{C_{\gamma}^c}2 v'(S)$
    so a union bound yields
    \begin{align}\label{sqsum3}
        \PR{V_2\geq \frac{C_\gamma^c}2v'(S)}
        \le 2\PR{\sum_{i=1}^rZ'_i > \frac{C_{\gamma}^{ c}}{2} v'(S)}
                + \PR{A_{r}^c} + \PR{B_{r}^c}\; .
    \end{align}
    We now wish to apply~\Cref{martingalebennettcor} to the martingale difference
    $(Z'_i,\F_i)_{i=0}^r$. Thus, we have to bound $|Z'_i|$ as well as the conditional
    variances $\VarC{Z'_i}{\F_{i-1}}$. For the bound on $Z_i'$, observe that by the Cauchy-Schwarz inequality,
    \[
        \abs{Z_i'}
            = [A_i \cap B_i \cap C_i] \abs{\sum_{j \in [m]}Y_i^{(j)}Y_{<i}^{(j)}}
            \leq [A_i\cap B_i \cap C_i]\sqrt{\sum_{j\in [m]}\left(Y_i^{(j)}\right)^2}\sqrt{\sum_{j\in [m]}\left(Y_{<i}^{(j)}\right)^2} \; .
    \]
    If $A_i$ occurs we obtain
    \[
        \sum_{j \in [m]}\left(Y_i^{(j)}\right)^2
            \le C_\gamma^{c-1} v'(G_i)
            \le C_\gamma^{c-1} v'(S)^{1 - 1/c} v'_{\infty}(S)^{1/c} \; ,
    \]
    by~\Cref{lem:groups} and if $A_i$, $B_i$, and $C_i$ all occur then
    \[
        \sum_{j \in [m]}\left(Y_{<i}^{(j)}\right)^2
            = \sum_{j \in [m]}\sum_{k<i}\left(Y_{k}^{(j)}\right)^2 + 2Z'_{<i}
            \le C_{\gamma}^{c-1}v'(G_{<i}) + 2C_{\gamma}^{c}v'(G_{<i})
            \le 3C_{\gamma}^{c}v'(S) \; .
    \]
    In conclusion
    \[
        |Z_i'| \le C_\gamma^{c - 1/2}\sqrt{3}v'(S)^{1 - 1/(2c)}v'_{\infty}(S)^{1/(2c)} \; .
    \]
    For the bound on the conditional variance note that if $B_i$ occurs then
    $\VarC{Z_i}{\F_{i-1}} \le m^{-1/2}v'(G_k)v'(G_{<k})$ and thus,
    \[
        \VarC{Z'_i}{\F_{i-1}}
            = [A_i][B_i][C_i]\VarC{Z_i}{\F_{i-1}}
            \le m^{-1/2}v'(G_k)v'(G_{<k}) \; .
    \]
    It follows that $\sum_{i=1}^r \VarC{Z'_i}{\F_{i-1}} \le m^{-1/2}v'(S)^2$.
    Letting
    \[
        \sigma^2 = m^{-1/2} v'(S)^2
    \]
    and
    \[
        M = C_\gamma^{c - 1/2}\sqrt{3}v'(S)^{1 - 1/(2c)}v'_{\infty}(S)^{1/(2c)}
    \]
    in~\Cref{martingalebennettcor} we thus obtain
    \[
        \PR{\sum_{i=1}^rZ_i' > \frac{C_{\gamma}^{ c}}{2} v'(S)}
            \le \exp\left(-\frac{v'(S)^{1/c}}{3 C_\gamma^{2c - 1} \cdot \sqrt{m} \cdot v'_{\infty}(S)^{1/c}}
                \C\left(\frac{(C\gamma)^{2c - 1/2}\cdot \sqrt{3}\cdot \sqrt{m}\cdot v'_{\infty}(S)^{1/2c}}{2v'(S)^{1/2c}} \right)
            \right) \; .
    \]
    Applying \Cref{lem:Benn-func-consts} first with $b=\left(\frac{v'_\infty(S)}{v'(S)}\right)^{1/(2c)}\leq 1$ and then with $b=\sqrt m>1$ yields
    \begin{align*}
        \frac{v'(S)^{1/c}}{3 C_\gamma^{2c - 1} \sqrt{m} v'_{\infty}(S)^{1/c}}
                \C\left(\frac{C_\gamma^{2c - 1/2}\sqrt{3}\sqrt{m}v'_{\infty}(S)^{1/2c}}{2v'(S)^{1/2c}}\right)
            &\ge \frac{1}{3C_\gamma^{2c - 1}}\C\left(\frac{C_\gamma^{2c - 1/2}\sqrt{3}\sqrt{m}}{2}\right) \; .
    \end{align*}
    We then use \Cref{rateofgrowth} to get
    \begin{align*}
        \frac{1}{3C_\gamma^{2c - 1}}\C\left(\frac{C_\gamma^{2c - 1/2}\sqrt{3}\sqrt{m}}{2}\right)
            \ge \frac{\sqrt{C_\gamma}}{\sqrt{3} \cdot 4}\log\left(1 + \frac{(C\gamma)^{2c - 1/2}\sqrt{3}\sqrt{m}}{2}\right)
            \ge \frac{\sqrt{C_\gamma}}{\sqrt{3} \cdot 8}\log(m)
            = \gamma\log(m) \; ,
    \end{align*}
    where we have used that $C_\gamma = 3 \cdot 8\cdot \gamma^2$ and $\gamma \ge 1$. Combining this we get that
    \begin{align}\label{sqsum4}
        \PR{\sum_{i=1}^rZ_i' > \frac{C_\gamma^c}2v'(S)} \le m^{-\gamma} \; .
    \end{align}

    It thus suffices to bound the probabilities $\Pr[A_{r-1}^c]$ and $\Pr[B_{r-1}^c]$. For $A_{r-1}^c$, if $c=1$ the discussion from the bound on $V_1$ proves that $A_{r-1}^c$ never occurs. If $c>1$, 
    the inductive hypothesis on the groups $G_i$ and a union bound yields
    \begin{align}\label{sqsum5}
        \PR{A_{r-1}^c} = O_{\gamma, \epsilon, c}\left(\sum_{i=1}^r |G_i|m^{-\gamma}\right) = O(|S|m^{-\gamma}) \; .
    \end{align}
    For $B_{r-1}^c$, we can for each $i\in \{1, \dots, r\}$ write
    \begin{align*}
        \VarC{Z_i}{\F_{i-1}}
            &= \EC{\left(\sum_{j \in [m]}X_i^{(j\oplus h(\alpha_i))}Y_{<i}^{(j)} \right)^2}{\F_{i-1}}
            \\&= \frac{1}{m}\sum_{k \in [m]}\left( \sum_{j \in [m]}X_i^{(j\oplus k)}Y_{<i}^{(j)} \right)^2
            \\&= \frac{1}{m}\sum_{k \in [m]}\sum_{(j_1, j_2) \in [m]^2}Y_i^{(j_1\oplus k)}Y_i^{(j_2\oplus k)}Y_{<i}^{(j_1)} Y_{<i}^{(j_2)}
            \\&= \frac{1}{m}\sum_{\substack{(j_1, j_2, j_3, j_4) \in [m]^4 \\ j_1 \oplus j_2 \oplus j_3 \oplus j_4 = 0}}
                Y_i^{(j_1)}Y_i^{(j_2)}Y_{<i}^{(j_3)} Y_{<i}^{(j_4)}
    \end{align*}
    Call this quantity $T_i$. It follows from~\Cref{lem:moment} and Markov's inequality that
    \[
    \PR{T_i \geq m^{-1/2}v'(G_k)v'(G_{<k})}
        \le \frac{\E{T_i^{2\gamma/(1 - 4\eps)}}}{m^{\gamma/(1 - 4\eps)}\left(v'(G_k)v'(G_{<k}) \right)^{2\gamma/(1 - 4\eps)}}
        \le O_{\gamma, \eps, c}(m^{-\gamma}).
    \]
    Thus, $\Pr[B_{r-1}^c]=O(|S|m^{-\gamma})$ by a union bound.
    
    Combining equations~\eqref{sqsum1}-\eqref{sqsum5} we conclude that indeed
    $\PR{V \ge C_\gamma^c v'(S)} =O_{\gamma,\eps, c}(|S|m^{-\gamma})$.
\end{proof}

\subsection{Establishing the Concentration Bound}
With the results of the previous subsection at hand, we proceed to prove the first part of~\Cref{thm:simpleConcentration}. We show that for a value function of support bounded in size by $m^\eps$ for some $\eps<1/4$, simple tabulation supports Chernoff-style bounds with \errorterm inversely polynomial in $m$. For convenience, we restate the first part of~\Cref{thm:simpleConcentration} as~\Cref{thm:simple-concentration-only}. The statement is equivalent to \cref{thm:simpleConcentration} but for precision, we have chosen to write out the statement more explicitly.
\begin{theorem}\label{thm:simple-concentration-only}
    Let $h\colon \Sigma^c \to [m]$ be a simple tabulation hash function and
    $S \subseteq \Sigma^c$ be a set of keys. Let $v\colon \Sigma^c \times [m] \to [-1, 1]$
    be a value such that the set
    $Q = \setbuilder{i \in [m]}{\exists x \in \Sigma^c : v(x, i) \neq 0}$ satisfies $\abs{Q} \le m^{\eps}$,
    where $\eps < \frac{1}{4}$ is a constant.
    Define the random variable $W=\sum_{x\in S}v(x, h(x))$ and write $\mu = \E{W}$ and $\sigma^2=\Var{W}$. Then for any constant $\gamma \geq 1$,
    \[
        \PR{\abs{W-\mu} \ge C_{\gamma, c} t}
            \le  2\exp\left(-\sigma^2 \C\left(\frac{t}{\sigma^2}\right)\right)
                + O_{\gamma, \eps, c}\left(\abs{S} m^{-\gamma}\right) \; ,
    \]
    where $C_{\gamma, c} = \left(1 + \frac{1}{\gamma}\right)^{3\frac{c(c - 1)}{2}}\left(C c \gamma \right)^{3c}$
    for some large enough universal constant $C$.
\end{theorem}
%\begin{remark}
%Note that by 2-independence $\Var{W}=\sigma^2$ is the same as in the fully random setting. Thus the first term in the bound above is asymptotically as good as we could hope for.     
%\end{remark}

\begin{proof}
First, akin to the proof of \cref{thm:sum-squares}, we may write 
$$V = W-\mu = \sum_{x\in S}\sum_{k\in Q}v(x, k)\left([h(x)=j\xor k]-\frac{1}{m}\right),$$ and note that $$\Var V=\Var W = \sum_{x \in S}\left(\sum_{k \in Q} \frac{1}{m} v(x, k)^2 - \left(\sum_{k \in Q} \frac{1}{m} v(x, k) \right)^2\right).$$
   
    We proceed by induction on $c$. For $c = 1$ we have full randomness and it
    follows immediately from~\Cref{martingalebennettcor} that 
    \[
        \PR{\abs{V} \ge t} \le 2\exp\left(-\sigma^2\C\left(\frac{t}{\sigma^2}\right)\right) \; .
    \]
    Now assume that $c > 1$ and inductively that the result holds for smaller
    values of $c$. We define $v'(x) = \sum_{k \in Q} v(x, k)^2$ and for $X \subseteq \Sigma^c$ we let
    $v'(X) = \sum_{x \in X} v'(x)$ and define $v'_\infty(X) = \max_{x \in X} v'(x)$.
    Now, applying \Cref{lem:groups} with respect to $w=v'$ we get position characters
    $\alpha_1, \ldots, \alpha_r$ with corresponding groups $G_1, \ldots, G_r$,
    such that $\cupdot_{i=1}^r G_i = S$ and for every $i \in \{1, \dots, r\}$,
    $v'(G_i) \le v'(S)^{1 - 1/c} v'_{\infty}(S)^{1/c}$. For a bin $j \in [m]$
    and an $i\in \{1, \dots, r\}$ we again define
    \begin{align*}
        X_i^{(j)} &= \sum_{x \in G_i}\sum_{k \in Q} 
            v(x, k) \cdot \left(\left[h(x \setminus \alpha_i) = j \oplus k\right] - \frac{1}{m}\right) \; ,
        & Y_i &= X_i^{\left(h(\alpha_i)\right)} \; .
    \end{align*} 
    Note that $\sum_{i=1}^r Y_i = V$. For $i \in \{1, \dots, r\}$ we define the $\sigma$-algebra
    $\F_i= \sigma((h(\alpha_j))_{j=1}^i)$. We furthermore define $Y_0=0$ and $\F_0=\{\emptyset, \Omega\}$. As $Y_i$ is $\F_i$-measurable for $i\in [r+1]$ and 
    $\EC{Y_i}{\F_{i - 1}} = 0$ for $i\in \{1, \dots, r\}$, $(Y_i,\F_i)_{i=0}^r$ is a martingale
    difference. Furthermore, for $i\in \{1, \dots, r\}$,
    \[
        \VarC{Y_i}{F_{i - 1}} = \frac{1}{m} \sum_{j \in [m]}\left(X_i^{(j)}\right)^2.
    \]
    According to~\Cref{thm:sum-squares} there exists a constant $K = 3 \cdot 2^6\cdot \gamma^2$ such that 
    \begin{align}\label{sqsumsbound}
        \PR{\sum_{j\in [m]}\left(X_i^{(j)}\right)^2 > K^{c - 1} v'(G_i)}
            \le O_{\gamma, \eps, c}\left(\abs{G_i}m^{\gamma}\right) \; .
    \end{align}
    For $i \in \{1, \dots, r\}$ we define the events
    \begin{align*}
        A_i &= \bigcap_{k \leq i} \left(
            \sum_{j \in [m]}(X_k^{(j)})^2 \leq K^{c - 1} v'(G_k)
        \right) \; , \\
        B_i &= \left(\max_{\substack{k \leq i \\ j \in [m]}} |X_k^{(j)}| \leq C_{\gamma + 1, c - 1} M \right) \; ,
    \end{align*}
    for some $M$ to be specified later. We define $Z_i = [A_i \cap B_i]Y_i$
    for $i \in \{1, \dots, r\}$ and $Z_0=0$. As both $A_i, B_i \in \F_{i-1}$ we have that
    $\EC{Z_i}{\F_{i-1}} = [A_i \cap B_i]\EC{Y_i}{\F_{i - 1}} = 0$ for $\{1, \dots, r\}$ so
    $(Z_i, \F_i)_{i=0}^r$ is a martingale difference. By definition of
    $A_i$ and $B_i$ it moreover holds for $i \in \{1, \dots, r\}$ that 
    \[
        |Z_i| \le C_{\gamma + 1, c - 1}M
            \quad \text{and} \quad
        \VarC{Z_i}{\F_{i - 1}} \le \frac{K^{c - 1} v'(G_i)}{m} \; .
    \]
    Setting $\sigma_0^2= \frac{K^{c - 1} v'(S)}{m}$ and
    applying~\Cref{martingalebennettcor} we obtain 
    \begin{align}\label{applybennett}
        \PR{\abs{\sum_{i=1}^r Z_i} \ge t}
            \le 2\exp\left(-\frac{\sigma_0^2}{C_{\gamma + 1, c - 1}^2 M^2} \C\left(\frac{t C_{\gamma + 1, c - 1} M}{\sigma_0^2}\right)\right) \; .
    \end{align}
    If $A_{r - 1}$ and $B_{r - 1}$ both occur then $\sum_{i=1}^r Z_i = \sum_{i=1}^r Y_i$ so it must hold that
    \begin{align*}
        \PR{\abs{V} \ge t}
            \le \PR{\abs{\sum_{i=1}^r Z_i} \ge t} + \PR{A_{r-1}^c} + \PR{B_{r-1}^c} \; .
    \end{align*}
    We may assume that $m>1$, i.e., the number of bins exceeds one, and then by the Cauchy-Schwarz inequality,
    \begin{align*}
        \sigma^2 = \sum_{x \in S}\left(\sum_{k \in Q} \frac{1}{m} v(x, k)^2 - \left(\sum_{k \in Q} \frac{1}{m} v(x, k) \right)^2\right)
            &\ge \sum_{x \in S}\left(\sum_{k \in Q} \frac{1}{m} v(x, k)^2 - \frac{1}{m^{2(1 - \eps)}}\sum_{k \in Q} \frac{1}{m^{\eps}} v(x, k)^2\right)
            \\&= \frac{v'(S)}{m}\left(1 - \frac1{m^{1 - \eps}}\right)
            \\&\ge \frac{v'(S)}{3m}
            \\&\ge \frac{\sigma_0^2}{3 K^{c - 1}}
    \end{align*}
    so using~\eqref{applybennett} we obtain
    \begin{align}\label{intermediatebound}
        \PR{\abs{V} \ge C_{\gamma, c} t}
            \le 2\exp\left(-\frac{3K^{c - 1} \sigma^2}{C_{\gamma + 1, c - 1}^2  M^2}
                \C\left(\frac{C_{\gamma, c}\cdot t\cdot C_{\gamma + 1, c - 1} M}{3K^{c - 1} \sigma^2} \right)
            \right) + \PR{A_{r-1}^c} + \PR{B_{r-1}^c} \; .
    \end{align}
    By~\eqref{sqsumsbound} and a union bound $\PR{A_{r-1}^c} \le  O(|S|m^{-\gamma})$.
    For bounding $\PR{B_{r-1}^c}$ we use the induction hypothesis on the groups,
    concluding that for $i \in \{1, \dots, r\}$ and $j \in [m]$,
    \[
        \PR{\abs{X_i^{(j)}} > C_{\gamma + 1, c - 1} M}
            \le 2\exp\left(-\sigma_i^2 \C\left(\frac{M}{\sigma_i^2}\right) \right) + O(|G_i|m^{-\gamma - 1}) \; ,
    \]
    where $\sigma_i^2=\Var{Y_i^{(j)}} \le v'(G_i)/m$. By the initial assumption on the groups, this implies $\sigma_i^2 \leq v'(S)^{1-1/c}v'_{\infty}(S)/m$ and we denote the latter
    quantity $\tau^2$. Combining with \Cref{lem:Benn-func-var-and-max} we obtain
    by a union bound that
    \[
        \PR{B_{r-1}^c}
            \le 2\abs{S}m \exp\left(-\tau^2 \C\left(\frac{M}{\tau^2}\right) \right) + O(|S|m^{-\gamma}) \; .
    \]
    We fix $M$ to be the unique real number with
    $\tau^2 \C\left(\frac{M}{\tau^2}\right) = (\gamma + 1)\log(m)$.
    With this choice of $M$, $\PR{B_{r-1}^c} \le O(\abs{S}m^{-\gamma})$,
    so by~\eqref{intermediatebound} it suffices to show that 
    \begin{align}\label{finalstep}
        \frac{3K^{c - 1} \sigma^2}{C_{\gamma + 1, c - 1}^2  M^2} \C\left(
            \frac{C_{\gamma, c}\cdot t\cdot  C_{\gamma + 1, c - 1} M}{3K^{c - 1} \sigma^2}
        \right)
            \ge \min\set{\sigma^2 \C\left(\frac{t}{\sigma^2}\right), \gamma \log(m)} \; .
    \end{align}
    First since $\frac{C_{\gamma, c} C_{\gamma + 1, c - 1}}{3K^{c - 1}} \ge 1$
    \Cref{lem:Benn-func-consts} give us that
    \[
        \frac{3\left(K\gamma\right)^{c - 1} \sigma^2}{C_{\gamma + 1, c - 1}^2  M^2} \C\left(
            \frac{C_{\gamma, c}\cdot t\cdot  C_{\gamma + 1, c - 1} M}{3\left(K\gamma\right)^{c - 1} \sigma^2}
        \right)
        \ge \frac{C_{\gamma, c}}{C_{\gamma + 1, c - 1}}\frac{\sigma^2}{M^2}\C\left(\frac{t M}{\sigma^2}\right) \; .
    \]
    Now by definition of $C_{\gamma, c}$ and $C_{\gamma + 1, c - 1}$ we get that
    \[
        \frac{C_{\gamma, c}}{C_{\gamma + 1, c - 1}}
            = \frac{\left(1 + \frac{1}{\gamma}\right)^{3\frac{c(c - 1)}{2}}\left(C c \gamma \right)^{3c}}
                {\left(1 + \frac{1}{\gamma + 1}\right)^{3\frac{(c - 1)(c - 2)}{2}}\left(C c (\gamma + 1) \right)^{3(c - 1)}}
            \ge (C c \gamma)^3 \; .
    \]
    So we have reduced the problem to showing that
    \[
        (C c \gamma)^3 \frac{\sigma^2}{M^2}\C\left(\frac{t M}{\sigma^2}\right)
            \ge \min\set{\sigma^2 \C\left(\frac{t}{\sigma^2}\right), \gamma \log(m)} \; . 
    \]
    For that we have to check a couple of cases.

    \paragraph{Case 1. $\frac{t M}{\sigma^2} \le 1$:} Using \Cref{rateofgrowth}
    twice and the fact that $(C c \gamma)^3 \ge \frac{3}{2}$, we get that
    \[
        (C c \gamma)^3 \sigma^2\C\left(\frac{t M}{\sigma^2}\right)
            \ge \frac{(C c \gamma)^3}{3} \frac{t^2}{\sigma^2}
            \ge \sigma^2 \C\left(\frac{t}{\sigma^2}\right) \; .
    \]

    \paragraph{Case 2. $v'(S) \le m^{(1 - \eps/c)\left(1 + \frac{1}{2c - 1}\right)}$:}
    We then get that
    \[
        \tau^2
            \le \frac{v'(S)^{1-1/c}v'_{\infty}(S)}{m}
            \le \frac{m^{(1 - \eps/c)\left(1 - \frac{1}{2c - 1}\right)}}{m^{1 - \eps/c}}
            = m^{-\frac{1 - \eps/c}{2c - 1}} \; .
    \]
    Now we note that $M \le 12 \gamma c$ since
    \begin{align*}
        \tau^2 \C\left(\frac{12 \gamma c}{\tau^2}\right)
            \ge 12 \gamma c \log\left(1 + \frac{12 \gamma c}{\tau^2}\right) / 2
            \ge 6 \gamma c \log(1 / \tau^2)
            \ge 6 \gamma c \frac{1 - \eps/c}{2c - 1} \log(m)
            \ge (\gamma + 1) \log(m) \; ,
    \end{align*}
    where we have used that $\eps \le \frac{1}{4}$ and $\gamma \ge 1$.

    We then get that
    \begin{align*}
        (C c \gamma)^3 \frac{\sigma^2}{M^2}\C\left(\frac{t M}{\sigma^2}\right)
            \ge \frac{(C c \gamma)^3}{M} \sigma^2 \C\left(\frac{t}{\sigma^2}\right)
            \ge \frac{(C c \gamma)^3}{12 c \gamma} \sigma^2 \C\left(\frac{t}{\sigma^2}\right)
            \ge \sigma^2 \C\left(\frac{t}{\sigma^2}\right) \; .
    \end{align*}
    % \paragraph{Case 4. $\frac{t M}{\sigma^2} > 1$ and $\tau^2 \ge 2 \gamma \log(m)$:}
    % We see that
    % \begin{align*}
    %     2 \gamma \log(m)
    %         \le \tau^2
    %         \le \frac{3(\sigma^2)^{1 - 1/c}}{m^{1/c - \eps/c}}
    %         \le \frac{3(\sigma^2)^{1 - 1/c}}{m^{\frac{1}{2c}}} \; .
    % \end{align*}
    % Rewriting this we get that
    % \[
    %     \sigma^2 \ge m^{\frac{1}{2(c - 1)}} (\tau^2)^{1 + \frac{1}{c - 1}}/9 \; .
    % \]
    % We see that $M \le \sqrt{2\gamma \log(m)}\tau$ since
    % \[
    %     \tau^2\C\left(\frac{\sqrt{2\gamma \log(m)}\tau}{\tau^2}\right)
    %         \ge \sqrt{2\gamma \log(m)}^2/2
    %         = \gamma \log(m) \; .
    % \]
    % This implies that
    % \begin{align*}
    %     \frac{\sigma^2}{M^2}
    %         \ge \frac{m^{\frac{1}{2(c - 1)}} (\tau^2)^{1 + \frac{1}{c - 1}}}{9 \cdot 2\gamma \log(m) \tau^2}
    %         \ge \frac{m^{\frac{1}{2(c - 1)}}}{18 \gamma \log(m)}
    %         \ge \log(m) \frac{e}{18 \cdot \gamma \cdot 4(c - 1)^2}
    % \end{align*}

    \paragraph{Case 3. $\frac{t M}{\sigma^2} > 1$ and
    $v'(S) > m^{(1 - \eps/c)\left(1 + \frac{1}{2c - 1}\right)}$:}
    We see that $M \le \max\set{6\gamma \log(m), \sqrt{6\gamma \log(m)} \tau}$
    since
    \begin{align*}
        \tau^2\C\left(\frac{\max\set{6\gamma \log(m), \sqrt{6\gamma \log(m)} \tau}}{\tau^2}\right)
            \ge \max\set{ \frac{6\gamma \log(m)}{3}, \frac{6\gamma\log(m)}{3} }
            \ge (\gamma + 1) \log(m) \; ,
    \end{align*}
    where we have used \Cref{rateofgrowth} and that $\gamma \ge 1$.    
    Now we have that $\sigma^2 \ge \frac{v'(S)}{3m} > m^{\frac{1 - 2\eps}{2c - 1}}/3$
    and $\tau^2 \le \frac{v'(S)^{1 - 1/c} v'_{\infty}(S)^{1/c}}{m}$. Combining this
    we get that
    \begin{align*}
        \frac{\sigma^2}{M^2}
            &\ge \min\set{\frac{\sigma^2}{36 \gamma^2 \log(m)^2}, \frac{\sigma^2}{6\gamma \log(m) \tau^2}}
            \\&\ge \min\set{
                \frac{m^{\frac{1 - 2\eps}{2c - 1}}}{108\gamma^2\log(m)^2},
                \left(\frac{v'(S)}{v'_{\infty}(S)} \right)^{1/c} \cdot \frac{1}{3 \cdot 6 \gamma \log(m)}
            }
            \\&\ge \min\set{
                \frac{m^{\frac{1 - 2\eps}{2c - 1}}}{108\gamma^2\log(m)^2},
                m^{\frac{1 - 2\eps}{2c}} \cdot \frac{1}{18\gamma \log(m)}
            }
            \\&\ge \frac{m^{\frac{1}{4c}}}{108\gamma^2\log(m)^2}
            \\&\ge \frac{\log(m)}{108 \cdot 4^3 c^3 \gamma^2} \; ,
    \end{align*}
    where have used that $\eps < \frac{1}{4}$ and that
    $\frac{m^{\frac{1}{4c}}}{\log(m)^2} \ge \frac{\log(m)}{4^3 c^3}$. Now we
    get that
    \begin{align*}
        (C c \gamma)^3 \frac{\sigma^2}{M^2}\C\left(\frac{t M}{\sigma^2}\right)
            \ge (C c \gamma)^3 \frac{\log(m)}{108 \cdot 4^3 c^3 \gamma^2} / 3
            \ge \gamma \log(m) \; .
    \end{align*}
\end{proof}

\section{General Value Functions -- Arbitrary Bins}\label{sec:valuefunctions}
The goal of this section is to prove~\Cref{thm:valuefunctions}, the second step towards~\Cref{thm:intro-tab-perm}. Again, we postpone the argument that our concentration bounds are query invariant to~\Cref{sec:queryinvariance}.
Recall that~\Cref{thm:valuefunctions} is concerned with a hash function of the form $h=\tau \circ g$, where $g: \Sigma^c \to [m]$ is a simple tabulation hash function and $\tau$ is a uniformly random permutation. Our goal is to prove that for any value function $v\colon \Sigma^c\times [m]\to [-1, 1]$, the sum $\sum_{x\in \Sigma^c}v(x, h(x))$ is strongly concentrated with high probability in $m$.
This result follows by combining the distributional properties of $g$ with the randomness of $\tau$. 

We start out by proving a lemma. The lemma describes properties we need $g$ to possess for the final composition with $\tau$ to yield Chernoff-style concentration.
\begin{lemma}\label{lemma:helpful}
	Let $m\geq 2$ be an integer and $C, T\in \R^+$ positive reals. Furthermore, let $\mathcal V\colon [m]\times[m] \to \R$ be a value function satisfying $\sum_{i \in [m]} \mathcal{V}(i, j) = 0$ for every $j\in [m]$ and such that
	\[
        \max_{i, j \in [m]} \abs{\mathcal{V}(i, j)} \le M := \max\set{C, \frac{\sigma^2}{T}} \; ,
    \]
    where $\sigma^2 = \frac{1}{m}\sum_{i \in [m]} \sum_{j \in [m]}   \mathcal{V}(i, j)^2$. If $\tau\colon [m]\to[m]$ is a uniformly random permutation, then the random variable $Z=\sum_{i\in [m]}\mathcal V(\tau(i), i)$ satisfies
    \begin{align*}
        \PR{|Z| \ge D t}
            &\le 4 \left( \exp\left(-\sigma^2 \C\left( \frac{t}{\sigma^2} \right)\right)
                + \exp\left(-\frac{T^2}{2\sigma^2} \right)
            \right) \; ,
           \end{align*}
    where $D = \max\set{8C, 12}$ is a universal constant depending on $C$.
\end{lemma}
\begin{proof}
    We define $Y_1 = \sum_{i=0}^{\ceil{m/2}-1} \mathcal{V}(\tau(i), i)$ and
    $Y_2 = \sum_{i = \ceil{m/2}}^{m-1} \mathcal{V}(\tau(i), i)$. Since $Z = Y_1 + Y_2$ it follows that if $Z > D t$ then there exists $i \in \set{1, 2}$
    such that $Y_i \ge \frac{D}{2} t$. It suffices to show that
    \begin{align}\label{eq:Yi_bound}
    	\PR{Y_i \ge \frac{D}{2} t} \le
             \exp\left(-\sigma^2 \C\left( \frac{t}{\sigma^2} \right)\right)
                + \exp\left(-\sigma^2 \C\left( \frac{T}{\sigma^2} \right)\right)
             \; ,
    \end{align}
    for $i \in \set{1, 2}$. A union bound over $i$ then yields a bound on $\Pr[Z\geq Dt]$. Since we may instead consider the value function $-\mathcal V$, the same argument yields a bound on $\Pr[Z\leq -Dt]$, which concludes the proof.

    Thus, we shall prove \eqref{eq:Yi_bound} for $Y_1$ -- the proof is completely analogous for $Y_2$.
    Define the filtration $(\F_i)_{i=0}^{\ceil{m/2}}$ by
    $\F_i = \sigma\left( (\tau(j))_{j \in [i]} \right)$ and let
    $X_i = \EC{Y_1}{\F_i}$ such that $(X_i, \F_i)_{i=0}^{\ceil{m/2}}$
    is a martingale, $X_0=\E{Y_1}$, and $X_{\ceil{m/2}}=Y_1$. Towards applying \Cref{martingalebennettcor}, we
    bound $\abs{X_i - X_{i - 1}}$ and $\sum_{i=1}^{\ceil{m/2}}\VarC{X_i - X_{i - 1}}{\F_{i - 1}}$.

    First, we bound $\abs{X_i - X_{i - 1}}$. We start by writing
    \begin{align*}
        X_i - X_{i - 1}
            &= \EC{Y_1}{\F_i} - \EC{Y_1}{\F_{i - 1}}
            \\&= \mathcal{V}(\tau(i-1), i-1) - \EC{\mathcal{V}(\tau(i-1), i-1)}{\F_{i - 1}}
                + \sum_{k=i}^{\ceil{m/2}-1} \left( 
                    \EC{\mathcal{V}(\tau(k), k)}{\F_i} - \EC{\mathcal{V}(\tau(k), k)}{\F_{i - 1}}
                \right) \; .
            % \\&= \mathcal{V}(i, \tau(i)) - \EC{\mathcal{V}(i, \tau(i))}{\F_{i - 1}}
            % + \sum_{k \in [m] \setminus [i]} \left( 
            %     - \frac{1}{m - i}\sum_{j \in [i]}\mathcal{V}(k, \tau(j)) + \frac{1}{m - i}\left( \EC{\mathcal{V}(k, \tau(i)}{\F_{i - 1}} \sum_{j \in [i - 1]}\mathcal{V}(k, \tau(j))
            % \right)
    \end{align*}
    Now, note that for $k\geq i$,
    \begin{align*}
        \EC{\mathcal{V}(\tau(k), k)}{\F_i}
            = - \frac{1}{m - i}\sum_{j=0}^{i-1}\mathcal{V}(\tau(j), k) \; ,
    \end{align*}
    since $\sum_{\ell\in [m]}\mathcal V(\ell, k)=0$, and furthermore,
    \begin{align*}
        \EC{\mathcal{V}(\tau(k), k)}{\F_{i - 1}}
            = - \frac{1}{m - i}\left(
                    \EC{\mathcal{V}(\tau(i-1), k)}{\F_{i - 1}}
                  + \sum_{j=0}^{i-2}\mathcal{V}(\tau(j), k)
                \right) \; .        
    \end{align*}
    Hence, it follows that
    \begin{align*}
        X_i - X_{i - 1}
            &=  \mathcal{V}(\tau(i-1), i-1) - \EC{\mathcal{V}(\tau(i-1), i-1)}{\F_{i - 1}}
                \\ &- \frac{1}{m - i} \sum_{k=i}^{\ceil{m/2}-1} \left(
                    \mathcal{V}(\tau(i-1), k) - \EC{\mathcal{V}(\tau(i-1), k)}{\F_{i - 1}}
                \right) \; .
    \end{align*}
    Since $\abs{\mathcal{V}(i, j)} \le M$ for all $i,j\in [m]$, it follows that
        $\abs{X_i - X_{i - 1}} \le 4 M$.

    Second, we bound $\VarC{X_i - X_{i - 1}}{\F_{i - 1}}$. To this end, observe that
    \begin{align*}
        \VarC{X_i - X_{i - 1}}{\F_{i - 1}}
            % &= \VarC{V(i, \tau(i)) - \EC{\mathcal{V}(i, \tau(i))}{\F_{i - 1}}
            %         - \frac{1}{m - i} \sum_{k \in [m] \setminus [i]} \left(
            %         \mathcal{V}(k, \tau(i)) - \EC{\mathcal{V}(k, \tau(i))}{\F_{i - 1}}
            % \right)}{\F_{i - 1}}
            &= \VarC{\mathcal{V}(\tau(i-1), i-1)
                - \frac{1}{m - i} \sum_{k=i }^{\ceil{m/2}-1} \mathcal{V}(\tau(i-1), k)}{\F_{i - 1}}
            % \\&\le \left( \sqrt{\VarC{\mathcal{V}(\tau(i), i)}{\F_{i - 1}}}
            %     + \frac{1}{m - i}\sum_{k \in [m] \setminus [i]}
            %         \sqrt{\VarC{\mathcal{V}(\tau(i), k)}{\F_{i - 1}}}
            % \right)^2
            \\&\le 2\left( \VarC{\mathcal{V}(\tau(i-1), i-1)}{\F_{i - 1}}
                + \frac{1}{m - i}\sum_{k=i}^{\ceil{m/2}-1}
                    \VarC{\mathcal{V}(\tau(i-1), k)}{\F_{i - 1}}
            \right) \; ,
            % \\&\le 2\left( \EC{\mathcal{V}(i, \tau(i))^2}{\F_{i - 1}}
            %     + \frac{1}{m - i}\sum_{k \in [m] \setminus [i]}
            %         \EC{\mathcal{V}(k, \tau(i))^2}{\F_{i - 1}}
            % \right)
            % \\&= 
    \end{align*}
    where the inequality follows from the fact that
    $2\CovarC{A}{B}{\mathcal{H}} \le \VarC{A}{\mathcal{H}} + \VarC{B}{\mathcal{H}}$,
    for any random variables $A$ and $B$ and any sigma algebra $\mathcal{H}$.
    For any $k \in [m]$,
    \begin{align*}
        \VarC{\mathcal{V}(\tau(i-1), k)}{\F_{i - 1}}
            &\le \EC{\mathcal{V}(\tau(i-1), k)^2}{\F_{i - 1}}
            \\&= \frac{1}{m - i+1}\sum_{j \in [m] \setminus \tau([i-1])} \mathcal{V}(j, k)^2
            \\&\le \frac{1}{m - i+1}\sum_{j \in [m]} \mathcal{V}(j, k)^2
            \\&\le \frac{2}{m}\sum_{j \in [m]} \mathcal{V}(j, k)^2 \; ,
    \end{align*}
    where the last inequality follows from the fact that $i \le \ceil{m/2} $.
    Hence, 
    \begin{align*}
        \VarC{X_i - X_{i - 1}}{\F_{i - 1}}
            &\le 2\left( \VarC{\mathcal{V}(\tau(i-1), i-1)}{\F_{i - 1}}
                + \frac{1}{m - i}\sum_{k=i}^{\ceil{m/2}-1}
                    \VarC{\mathcal{V}(\tau(i-1), k)}{\F_{i - 1}}
            \right)
            \\&\le \frac{4}{m} \sum_{j \in [m]} \mathcal{V}(j, i)^2
                + \frac{2}{m - i} \cdot \frac{2}{m}\sum_{k=i}^{\ceil{m/2}-1}\sum_{j \in [m]} \mathcal{V}(j, k)^2
            \\&\le \frac{4}{m} \sum_{j \in [m]} \mathcal{V}(j, i)^2
                + \frac{16}{m^2}\sum_{k \in [m]}\sum_{j \in [m]} \mathcal{V}(j, k)^2 \; ,
    \end{align*}
    again using that $i \le \ceil{m/2}$. We now see that
    \begin{align*}
        \sum_{i=1}^{\ceil{m/2}} \VarC{X_i - X_{i - 1}}{\F_{i - 1}}
            &\le \sum_{i=1}^{\ceil{m/2}} \left(
                \frac{4}{m} \sum_{j \in [m]} \mathcal{V}(j, i)^2
                + \frac{16}{m^2}\sum_{k \in [m]}\sum_{j \in [m]} \mathcal{V}(j, k)^2
            \right)
            \\&\le \sum_{i\in [m]} \left(
                \frac{4}{m} \sum_{j \in [m]} \mathcal{V}(j, i)^2
                + \frac{16}{m^2}\sum_{k \in [m]}\sum_{j \in [m]} \mathcal{V}(j, k)^2
            \right)
            \\&\le 20\sigma^2 \; .
    \end{align*}

    The assumption on $\mathcal{V}$ implies that $\E{\mathcal{V}(\tau(i),i)}=0$ for each $i \in [m]$, so also $\E{Y_1}=0$. Applying~\Cref{martingalebennettcor} then yields,
    \begin{align*}
        \PR{Y_1 \ge \frac{D}{2}t}
            \le \exp\left(-\frac{20\sigma^2}{(4M)^2} \C\left( \frac{(D/2)t \cdot 4 M}{20 \sigma^2} \right) \right) = \exp\left(-\frac{5\sigma^2}{4M^2} \C\left( \frac{DMt}{10 \sigma^2} \right) \right).
    \end{align*}
    The goal is now to show that
	\begin{align}\label{eq:tech_to_check}
		        \frac{5\sigma^2}{4M^2} \C\left( \frac{DMt}{10 \sigma^2} \right)
            \ge \min\set{\sigma^2\C\left(\frac{t}{\sigma^2}\right), \frac{T^2}{2\sigma^2}} .
	\end{align}
    Because if this is the case, then as desired
    \[
        \PR{Y_1 \ge \frac{D}{2}t}
             \le \exp\left(-\min\set{\sigma^2\C\left(\frac{t}{\sigma^2}\right), \frac{T^2}{2\sigma^2}} \right)
             \le \exp\left(-\sigma^2\C\left(\frac{t}{\sigma^2}\right)\right) + \exp\left(-\frac{T^2}{2\sigma^2}\right).
    \]
    We check \eqref{eq:tech_to_check} by cases. This completes the proof.
    
    \paragraph{Case 1. $M \le \frac{10}{D}$:} In this case, $\frac{DM}{10} \le 1$. Thus, by
    \Cref{lem:Benn-func-consts},
    \begin{align*}
        \frac{5\sigma^2}{4M^2} \C\left( \frac{D M t}{10 \sigma^2} \right)
            \ge \frac{D^2}{80}\sigma^2\C\left(\frac{t}{\sigma^2}\right)
            \ge \sigma^2\C\left(\frac{t}{\sigma^2}\right) \; ,
    \end{align*}
    using that $D \ge 12 \ge \sqrt{80}$.

    \paragraph{Case 2. $\frac{10}{D} \le M \le C$:} In this case, $\frac{DM}{10} \ge 1$. Thus,
    by \Cref{lem:Benn-func-consts},
    \begin{align*}
        \frac{5\sigma^2}{4M^2} \C\left( \frac{D M t}{10 \sigma^2} \right)
            \ge \frac{D}{8M}\sigma^2\C\left(\frac{t}{\sigma^2}\right)
            \ge \frac{D}{8C}\sigma^2\C\left(\frac{t}{\sigma^2}\right)
            \ge \sigma^2\C\left(\frac{t}{\sigma^2}\right) \; ,
    \end{align*}
    using that $D \ge 8C$.

    \paragraph{Case 3. $M \le \frac{\sigma^2}{T}$:} In this case, recall that $D \ge 12$ such that
    $\frac{D}{10} \ge 1$ and we may apply \Cref{lem:Benn-func-consts}, yielding
    \begin{align*}
        \frac{5\sigma^2}{4M^2} \C\left( \frac{D M t}{10 \sigma^2} \right)
            \ge \frac{5}{4} \frac{T^2}{\sigma^2} \C\left( \frac{D}{10} \frac{t}{T} \right)
            \ge \frac{D}{8} \frac{T^2}{\sigma^2} \C\left( \frac{t}{T} \right) \; .
    \end{align*}
    By \Cref{rateofgrowth},
    \[
        \C\left( \frac{t}{T} \right)
            \ge \C\left( \min\set{\frac{t}{T}, 1} \right)
            \ge \min\set{\frac{t^2}{3T^2}, \frac{1}{3}} \; .
    \]
    So finally,
    \begin{align*}
        \frac{5\sigma^2}{4M^2} \C\left( \frac{D M t}{10 \sigma^2} \right)
            \ge \min\set{ \frac{D}{24} \frac{t^2}{\sigma^2}, \frac{D}{24}\frac{T^2}{\sigma^2}}
            \ge \min\set{ \frac{D}{12} \sigma^2 \C\left(\frac{t}{\sigma^2}\right), \frac{D}{24}\frac{T^2}{\sigma^2} }
            \ge \min\set{ \sigma^2 \C\left(\frac{t}{\sigma^2}\right), \frac{T^2}{2\sigma^2}} \; ,
    \end{align*}
    where we have used \Cref{rateofgrowth} and the fact that $D \ge 12$.

   \end{proof}
	With this result in hand we are ready to prove \cref{thm:valuefunctions}. We restate it here in a more technically explicit version. For a more intuitive understanding, please refer back to the original statement. Note that we only require the hash function $h$ of the theorem to be 2-independent, whereas \cref{thm:valuefunctions} requires the hash function to be 3-independent. The difference lies in that the statement of \cref{thm:valuefunctions} is slightly stronger, guaranteeing query invariance. Having deferred the treatment of query invariance until later, we only need 2-independence for now.
	\begin{theorem}\label{thm:tab-perm-full}

		Let $\eps\in (0, 1]$ and $m\geq 2$ be given. Let $h\colon A\to [m]$ be a 2-independent hash function satisfying the following. For every $\gamma>0$ and every value function $\tilde v\colon A\times [m]\to [-1, 1]$ such that $Q=\setbuilder{i\in [m]}{\exists x\in A\colon \tilde v(x, i)\neq 0}$ has size $\abs Q\leq m^\eps$, the random variables $W=\sum_{x\in A}\tilde v(x, h(x))$ and $W_j=\sum_{x\in A}\tilde v(x, h(x)\xor j), j\in [m]$ with mean $\mu_W = \E W = \E {W_j}$ and variance $\sigma_W^2 = \Var{W}$ satisfy the inequalities
		\begin{align}
				\Pr \left[ \abs {W-\mu_W} \geq C\cdot t \right]\leq 2\exp\left(-\sigma_W^2 \C\left( \frac{t}{\sigma_W^2} \right) \right) + O(\abs Am^{-\gamma}), \label{eq:many_bins}\\
				\Pr\left[ \sum_{j\in [m]}\left( W_j-\mu_W \right)^2 \geq D\cdot \sum_{x\in A}\sum_{k\in Q}\tilde v(x, k)^2\right] = O(\abs Am^{-\gamma}), \label{eq:variance}
			\end{align}
		for every $t>0$, where $C$ and $D$ are universal constants depending on $\gamma$ and $\eps$.
		
		Let $v\colon A\times [m]\to [-1, 1]$ be any value function, $\tau\colon [m]\to [m]$ a uniformly random permutation independent of $h$, and $\gamma>0$. The random variable $U = \sum_{x\in A}v(x, \tau(h(x)))$ with expectation $\mu=\E U$ and variance $\sigma^2=\Var U$ satisfies
		\begin{align}
			\Pr\left[ \abs{U-\mu}\geq E\cdot t \right]\leq 6\exp\left( -\sigma^2\C\left(\frac{t}{\sigma^2}\right) \right)+O(\abs Am^{-\gamma})
		\end{align}
		for every $t>0$, where $E$ is a universal constant depending on $\gamma$ and $\eps$.
	\end{theorem}
\begin{proof}
Define $v':A \times [m] \to [-1,1]$ by letting $v'(x,i)=\frac{1}{2}\left(v(x,i)-\frac{1}{m}\sum_{j\in [m]}v(x,j)\right)$ and write $V=U-\mu$. Since $\sum_{i\in [m]} \left( [\tau(h(x)) = i] - \frac{1}{m} \right)=0$, we may write
\begin{align*}
V=& \sum_{x\in A}  \sum_{i\in [m]} v(x, i)[\tau(h(x)) = i] - \frac1m \sum_{x\in A}\sum_{i\in [m]}v(x, i) + \sum_{x\in A}\left(\left(\sum_{i\in [m]}\left( [\tau(h(x)) = i] - \frac{1}{m} \right)\right)\cdot \left( \frac1m \sum_{j\in [m]}v(x, j) \right)\right)\\ 
=&\sum_{x\in A} \sum_{i\in [m]} \left(v(x,i)-\frac{1}{m}\sum_{j\in [m]}v(x,j)\right)\left( [\tau(h(x)) = i] - \frac{1}{m} \right) \\
=&2\sum_{x\in A} \sum_{i\in [m]} v'(x,i)\left( [\tau(h(x)) = i] - \frac{1}{m} \right).
\end{align*}
We write $V' = \sum_{x\in A} \sum_{i\in [m]} v'(x,i)\left( [\tau(h(x)) = i] - \frac{1}{m} \right)$ such that $V=2V'$.
We note that by $2$-independence 
$$
\sigma^2=\sum_{x\in A} \Var{v(x,\tau(h(x)))}=\sum_{x\in A}\E{\left(v(x,\tau(h(x)))-\frac{1}{m}\sum_{j\in [m]}v(x,j)\right)^2}=\frac{4}{m}\sum_{x\in A} \sum_{i \in [m]}v'(x,i)^2.
$$
Thus, we may write $\sigma'^2=\Var{V'}=\frac{1}{m}\sum_{x\in A} \sum_{i \in [m]}v'(x,i)^2$. We proceed to show that for some constant $E'$ depending on $\gamma$ and $\eps$,
\begin{align*}
        \PR{\abs{ V' } \ge E' \cdot t}
            \le 6 \exp\left(-\sigma'^2 h\left(\frac{t}{\sigma'^2}\right)\right) + O(|A|m^{-\gamma}) \; ,
    \end{align*}
As $\sigma'\leq \sigma$ and $V=2V'$ the theorem then follows with $E=2E'$ by applying~\Cref{lem:Benn-func-var-and-max}.

 For $i\in [m]$ we define $\sigma_i^2=\frac{1}{m}\sum_{x\in A}v'(x,i)^2$, so that $\sum_{i\in [m]}\sigma_i^2=\sigma'^2$. Assume without loss of generality that $\sigma_0^2\geq \cdots \geq \sigma_{m-1}^2$. Now define $\mathcal{V}:[m] \times [m] \to \R$ by
$$
\mathcal{V}(i,j)=\sum_{x\in A}v'(x,j)\left([h(x)=i]-\frac{1}{m} \right).
$$
Note that for any $j\in [m]$, $\sum_{i\in [m]}\mathcal{V}(i,j)=0$, regardless of the (random) choice of $h$. With this definition, $V'=\sum_{i\in [m]} \mathcal{V}(i,\tau (i))=\sum_{j\in [m]} \mathcal{V}(\tau^{-1}(j),j)$. Now let 
$$
V_1=\sum_{j\in [m^\eps]}\mathcal{V}(\tau^{-1}(j),j) \quad \text{and} \quad V_2 =\sum_{j\in [m]\setminus[m^\eps]}\mathcal{V}(\tau^{-1}(j),j),
$$
and note that $V_1+V_2=V'$. Defining value functions $v_1', v_2'\colon A\times[m]\to [-1, 1]$ by
$$
v_1'(x,i)=\begin{cases}
v'(x,i), & \text{if } i\in [m^\eps] \\
0, & \text{otherwise} 
\end{cases} 
\quad
\text{and}
\quad
v_2'(x,i)=\begin{cases}
v'(x,i), & \text{if } i\in[m] \setminus [m^\eps] \\
0, & \text{otherwise} 
\end{cases},
$$
we observe that 
$$V_1 = \sum_{x\in A} v_1'(x, \tau(h(x))) - \E{\sum_{x\in A} v_1'(x, \tau(h(x)))}
\quad\text{and}\quad
V_2 = \sum_{x\in A} v_2'(x, \tau(h(x))) - \E{\sum_{x\in A} v_2'(x, \tau(h(x)))}$$

Let $D\geq 1$ be such that~\cref{eq:variance} holds with \errorterm $O(|A|m^{-\gamma-1})$ and let $M=\max\left\{C, \frac{\sigma'}{\sqrt{2D\gamma\log m}}\right\}$ for some large constant $C$ to be fixed later. Define the two events
$$
\mathcal{A}=\bigcap_{j\in[m]\setminus [m^\eps]} \left( \max_{i \in m} |\mathcal{V}(i,j)| \leq M\right) \quad \text{and} \quad \mathcal{B}=\bigcap_{j\in [m]}\left(\sum_{i \in [m]} \mathcal{V}(i,j)^2<D\sigma_j^2m\right).
$$
By a union bound,
$$
\Pr[|V'|\geq E't]\leq \Pr[|V_1|\geq E't/2]+\Pr[|[\mathcal{A}]\cdot[\mathcal{B}] \cdot V_2|\geq E't/2]+\Pr[\mathcal{A}^c]+\Pr[\mathcal{B}^c],
$$
and we proceed to bound each of these terms individually. 

First, we bound $\Pr[|V_1|\geq E't/2]$. To do so, suppose we fix the permutation $\tau=\tau_0$. With this conditioning and by $2$-independence,
\begin{align*}
\VarC{V_1}{\tau=\tau_0}=&\Var{\sum_{x \in A} [\tau_0(h(x))\in [m^\eps]]\cdot v'(x,\tau_0(h(x)))}\leq \sum_{x \in A} \E{[\tau_0(h(x))\in [m^\eps]]\cdot v'(x,\tau_0(h(x)))^2} \\
=&\frac{1}{m}\sum_{x \in A} \sum_{j\in [m^\eps]}v'(x,j)^2\leq \sigma'^2.
\end{align*}
Defining 
$\overline{v}: A \times[m] \to[-1,1]$ by $\overline{v}(x,i)=v_1'(x,\tau_0(i))$
it holds that 
$$
V_1=\sum_{x \in A} \sum_{i \in \tau_0^{-1}([m^\eps])} \overline{v}(x, i) \left( [h(x) = i] - \frac{1}{m} \right).
$$
As $\overline{v}$ has support of size at most $m^\eps$ we can apply~\cref{eq:many_bins} to conclude that 
$$
\Pr[|V_1|\geq E't/2\,\vert\, \tau=\tau_0]\leq 2 \exp\left(-\sigma'^2 \C\left(\frac{t}{\sigma'^2}\right)\right) + O(|A|m^{-\gamma}),
$$
if the constant $E'$ is large enough. Since this holds for any fixed $\tau_0$, it also holds for the unconditioned probability.

We now bound $\Pr[|[\mathcal{A}]\cdot[\mathcal{B}] \cdot V_2|\geq E't/2]$. It suffices to condition on $h=h_0$ for some $h_0$ satisfying that $[\mathcal{A}]=[\mathcal{B}]=1$ and make the bound over the randomness of $\tau$.  For this we may use~\Cref{lemma:helpful}. Indeed if $h=h_0$ for some $h_0$ such that $[\mathcal{A}]=[\mathcal{B}]=1$, then $\sum_{i \in [m]} \sum_{j \in [m]} \frac{1}{m}  \mathcal{V}(i, j)^2\leq D\sigma'^2$. Here we used the conditioning on $\mathcal{A}$. Define the function $\mathcal{V}_0:[m]\times [m]\to \R$ by $\mathcal{V}_0(i,j)=\mathcal{V}(i,j)$ when $j\in [m]\setminus [m^\eps]$ and $\mathcal{V}_0(i,j)=0$ otherwise. Then also $\sum_{i \in [m]} \sum_{j \in [m]} \frac{1}{m}  \mathcal{V}_0(i, j)^2\leq D\sigma'^2$ and further, for each $j\in[m]$, $\sum_{i\in [m]}\mathcal{V}_0(i,j)=0$. Finally, the conditioning on $\mathcal{B}$ gives that $\max_{i,j\in [m]}\mathcal{V}_0(i,j)\leq M$. Note that $V_2=\sum_{j\in [m]}\mathcal{V}_0(\tau^{-1}(j),j)$. Applying~\cref{lemma:helpful} to $\mathcal{V}_0$, noting that the bound obtained in that lemma is increasing in $\sigma$, we obtain that 
$$
        \PR{|V_2| \ge E' t/2}
            \le 4 \left( \exp\left(-D\sigma'^2 \C\left( \frac{t}{D\sigma'^2} \right)\right)
                + \exp\left(-\gamma \log m \right)
            \right)=
            4\exp \left(- \Omega \left(\sigma'^2 \mathcal{C} \left(\frac{t}{\sigma'^2} \right)\right) \right)+O(m^{-\gamma}),
$$
if $E'$ is sufficiently large. From this it follows that,
$$
\Pr[[\mathcal{A}]\cdot [\mathcal{B}]\cdot |V_2|\geq E't/2]\leq 4 \exp\left(-\sigma'^2 \C\left(\frac{t}{\sigma'^2}\right)\right) + O(m^{-\gamma}).
$$

We finally need to bound $\Pr[\mathcal{A}^c]$ and $\Pr[\mathcal{B}^c]$. By the choice of $D$ and a union bound we obtain that $\Pr[\mathcal{B}^c]=O(|A|m^{-\gamma})$, so for completing the proof it suffices to bound $\Pr[\mathcal{A}^c]$ which we proceed to do now. More specifically we bound $\Pr[|\mathcal{V}(i,j)|\geq M]$ for each $(i,j)\in [m]\times ([m]\setminus [m^\eps])$, finishing with a union bound over the $m^2$ choices. So let $(i,j)\in [m]\times ([m]\setminus [m^\eps])$ be fixed and define $\tilde{v}:A \times [m]\to [-1,1]$ by $\tilde v (x,i)=v_2'(x,j)$ and $\tilde v (x,k)=0$ for $k\neq i$.  Then $\tilde v$ has support $A \times \{i\}$,
$$
\mathcal{V}(i,j)=\sum_{x \in A} \sum_{k \in [m]} \tilde v(x, k) \left( [\tau(h(x)) = i] - \frac{1}{m} \right),
$$
and $\Var{\mathcal{V}(i,j)}\leq \frac{1}{m}\sum_{x \in A} v_2'(x,j)^2=\sigma_j^2 \leq \sigma'^2/m^\eps$. The last inequality follows from our assumption that $\sigma_0^2\geq  \cdots \geq \sigma_{m-1}^2$ and $j\geq m^\eps$.

By the assumption of~\cref{eq:many_bins} with $\gamma$ replaced by $\gamma+2$ it follows that 
$$
\Pr[|\mathcal{V}(i,j)|\geq M] \leq 2 \exp\left(-\Omega\left(\sigma_j^2 \C\left(\frac{M}{\sigma_j^2}\right)\right)\right) + O(|A|m^{-\gamma-2})\le 2 \exp\left(-D'\frac{\sigma'^2}{m^\eps} \C\left(\frac{Mm^\eps}{\sigma'^2}\right)\right) + O(|A|m^{-\gamma-2}),
$$ 
for some constant $D'$. We finish the proof by showing that if the constant $C$ from the definition of $M$ is large enough, then 
$$
2 \exp\left(-D'\frac{\sigma'^2}{m^\eps} \C\left(\frac{Mm^\eps}{\sigma'^2}\right)\right)=O(m^{-\gamma-2}).
$$
For this it suffices to show that if $C$ is large enough and $m$ is greater than some constant, then 
$$
\frac{\sigma'^2}{m^\eps} \C\left(\frac{Mm^\eps}{\sigma'^2}\right) \geq \frac{(\gamma+2) \log m}{D'}.
$$
Suppose first that $\sigma'^2\leq m^{\eps/2}$. In that case we use~\Cref{rateofgrowth} to conclude that 
$$
\frac{\sigma'^2}{m^\eps} \C\left(\frac{Mm^\eps}{\sigma'^2}\right) \geq \frac{M}{2} \log \left(\frac{Mm^\eps}{\sigma'^2}+1 \right)\geq \frac{C}{2} \log \left( Cm^{\eps /2}+1 \right) \geq \frac{C\eps}{4} \log m,
$$
so if $C\geq 4\frac{\gamma+2}{D'\eps}$ this is at least $\frac{(\gamma+2) \log m}{D'}$. Now suppose $m^{\eps/2} < \sigma'^2 \leq m^{2\eps}/(2D  \gamma \log m)$. In that case we recall that $M=\max\left\{C, \frac{\sigma'}{\sqrt{2D\gamma\log m}}\right\}$ and use the bound
$$
\frac{\sigma'^2}{m^\eps} \C\left(\frac{Mm^\eps}{\sigma'^2}\right) \geq \frac{M}{2} \log \left(\frac{Mm^\eps}{\sigma'^2}+1 \right)\geq \frac{\sigma'}{\sqrt{8D\gamma \log m} } \log \left(\frac{m^\eps}{\sigma'\sqrt{2D  \gamma \log m}}+1 \right) = \Omega \left( \frac{m^{\eps/4}}{\sqrt{\log m}} \right).
$$
If $m$ is larger than some constant, this is certainly at least $\frac{(\gamma+2) \log m}{D'}$. Finally suppose that $\sigma'^2 >m^{2\eps}/(2D  \gamma \log m)$. Using the inequality $\log(1+x)\geq \frac{x}{2}$ holding for $0\leq x \leq 1$ we find that 
$$
\frac{\sigma'^2}{m^\eps} \C\left(\frac{Mm^\eps}{\sigma'^2}\right) \geq \frac{\sigma'}{\sqrt{8D\gamma \log m} } \log \left(\frac{m^\eps}{\sigma'\sqrt{2D  \gamma \log m}}+1 \right) \geq \frac{m^\eps}{8D \gamma \log m}.
$$
Again it holds that if $m$ is greater than some constant, this is at least $\frac{(\gamma+2) \log m}{D'}$. It follows that if $C$ is large enough, then $\Pr[|\mathcal{V}(i,j)|\geq M]=O(|A|m^{-\gamma-2})$. Union bounding over $(i,j)\in [m] \times ([m]\setminus [m^\eps])$ we find that $\Pr[\mathcal{A}^c]=O(|A|m^{-\gamma})$. Combining the bounds we find that 
$$
\Pr[|V'|\geq E't] \leq 6 \exp\left(-\sigma'^2 h\left(\frac{t}{\sigma'^2}\right)\right) + O(|A|m^{-\gamma}),
$$
which completes the proof.

\end{proof}

\section{Extending the Hash Range}\label{sec:codomainext}
This section is dedicated to proving~\Cref{thm:extendingCodomain}, which we will restate shortly. Again, we will postpone the argument that our concentration bounds are query invariant to~\Cref{sec:queryinvariance}.
First, we prove the following technical lemma.
\begin{lemma}\label{lem:increase-var}
    Let $\sigma^2 > 0$ and $t > 0$. Writing $s=\max\set{\sigma^2, \sqrt{t \sigma^2}}$, 
    \[
        s \cdot \chernoff{\frac{t}{s}}
            \ge \sigma^2 \chernoff{\frac{t}{\sigma^2}}/4
        \; .
    \]
\end{lemma}
\begin{proof}
    For $t \le \sigma^2$ the inequality is trivial, so suppose $t>\sigma^2$. We note that for $x \ge 0$,
    $1 + \sqrt{x} \ge \sqrt{1 + x}$, such that $\lg(1 + \sqrt{x}) \ge \lg(1 + x)/2$
    for every $x \ge 0$. Using this fact in between two applications of \Cref{rateofgrowth}, we find that
    \[
        \sqrt{t \sigma^2} \chernoff{\frac{t}{\sqrt{t \sigma^2}}}
            \ge t \lg\left(1 + \sqrt{\frac{t}{\sigma^2}} \right)/2
            \ge t \lg\left(1 + \frac{t}{\sigma^2} \right)/4
            \ge \sigma^2 \chernoff{\frac{t}{\sigma^2}}/4
        \; .
    \]
\end{proof}
Next, we recall the law of total variance.
\begin{lemma}[Law of Total Variance]
    For every pair of random variables $X, Y$, 
    $$ \Var{Y} = \E{\VarC{Y}{X}} + \Var{\EC{Y}{X}}.$$
    In particular, $\Var{Y} \geq  \Var{\EC{Y}{X}}$.
\end{lemma}
We are now ready to prove the main theorem of the section, which informally states that concatenating the output values of hash functions preserves the property of having Chernoff-style bounds. 
	Note that the following is a much more explicit and elaborate statement of \cref{thm:extendingCodomain}. The purpose of this restatement is to make a formal proof more readable. The reader is encouraged to refer back to \cref{thm:extendingCodomain} for intuition regarding the theorem statement. Again, we highlight that we have left out the part of~\Cref{thm:extendingCodomain} concerning query independence. How query independence is obtained will be discussed in~\Cref{sec:queryinvariance}
\begin{theorem}
    Let $A$ be a finite set. Let $(X_a)_{a \in A}$ and $(Y_a)_{a \in A}$ be pairwise independent families
    of random variables taking values in $B_X$ and $B_Y$, respectively, and satisfying that the
    distributions of $(X_a)_{a \in A}$ and $(Y_a)_{a\in A}$ are independent. Suppose that there exist
    universal constants $D_X, D_Y \geq 1$, $K_X, K_Y>0$, and $R_X, R_Y\geq 0$ such that for every choice
    of value functions $v_X \colon A \times B_X \to [0, 1]$ and $v_Y \colon A \times B_Y \to [0, 1]$
    and for every $t > 0$,
    \begin{align}\label{eq:X-chernoff}
        \PR{ \abs{\sum_{a \in A} v_X(a, X_a) - \mu_X} > t}
            &< K_X\exp\left( - \sigma_X^2 \chernoff{\frac{t}{\sigma_X^2}} / D_X \right) + R_X
        \; , \\ \label{eq:Y-chernoff}
        \PR{ \abs{\sum_{a \in A} v_Y(a, Y_a) - \mu_Y} > t}
            &< K_Y\exp\left( - \sigma_Y^2 \chernoff{\frac{t}{\sigma_Y^2}} / D_Y \right) + R_Y
        \; .    
    \end{align}
    where $\mu_X = \E{\sum_{a \in A} v_X(a, X_a)}$, $\mu_Y = \E{\sum_{a \in A} v_Y(a, Y_a)}$,
    $\sigma_X^2 = \Var{\sum_{a \in A} v_X(a, X_a)}$, and $\sigma_Y^2 = \Var{\sum_{a \in A} v_Y(a, Y_a)}$.
    Then for every value function $\overline{v} \colon A \times B_X \times B_Y \to [0, 1]$ and
    every $t > 0$,
    \[
        \PR{\abs{\sum_{a \in A} \overline{v}(a, X_a, Y_a) - \mu_{XY}} > t}
            < K_{KY}\exp\left(-\sigma_{XY}^2 \chernoff{\frac{t}{\sigma_{XY}^2}} / D_{XY} \right) + R_{XY}
        \; ,
    \]
    where $\mu_{XY}=\E{\sum_{a \in A} \overline{v}(a, X_a, Y_a)}$,
    $\sigma_{XY}^2 = \Var{\sum_{a \in A} \overline{v}(a, X_a, Y_a)}$,
    $D_{XY} = \max\set{144D_X, 72D_Y }$, $K_{XY} = 3K_X + K_Y$,
    and $R_{XY} = 3R_X + R_Y$.    
\end{theorem}
\begin{proof}
    Let a value function, $\overline{v} \colon A \times B_X \times B_Y \to [0, 1]$, and a positive real, $t>0$, be given.
    Define $V_a = \overline{v}(a, X_a, Y_a)$, $\mu_a = \E{V_a}$,
    and $\sigma_a^2 = \Var{V_a}$. We shall be concerned with the variance of $V_a$ when conditioned on $X_a$. Hence, we define 
    $$L_a =  \indicator{\VarC{V_a}{X_a} > \sqrt{\frac{6 \sigma_{XY}^2}{t}}}\,\,
    \text{ and }\,\,
    S_a = \indicator{\VarC{V_a}{X_a} \leq \sqrt{\frac{6 \sigma_{XY}^2}{t}}}$$
    to be the indicators on the conditional variance of $V_a$ given $X_a$ being larger or smaller, respectively, than the threshold $\sqrt{\frac{6 \sigma_{XY}^2}{t}}$.    
    Noting that $L_a+S_a=1$, we split the sum $\sum_{a \in A} (V_a - \mu_a)$ into three parts.
    \[\begin{split}
        \sum_{a \in A} (V_a - \mu_a)
            &= \underbrace{\sum_{a \in A} (\EC{V_a}{X_a} - \mu_a)}_{T_1}
            \\&+ \underbrace{\sum_{a \in A} L_a(V_a - \EC{V_a}{X_a})}_{T_2}
            \\&+ \underbrace{\sum_{a \in A} S_a(V_a - \EC{V_a}{X_a})}_{T_3}
    \end{split}\]
    Now, the triangle inequality and a union bound yields
    \begin{align*}
        \PR{\abs{\sum_{a \in A} \overline{v}(a, X_a, Y_a) - \mu_{XY}} > t} = \PR{\abs{\sum_{a \in A} (V_a-\mu_a)} > t}
        \leq \sum_{i=1}^3\PR{\abs{T_i}>t/3}.
    \end{align*}
    We shall bound each of the three terms $T_1, T_2$, and $T_3$ individually.

    For bounding $\PR{\abs{T_1}>t/3}$, define the value function $v_X^{(1)}\colon A\times B_X\to [0, 1]$ by $v_X^{(1)}(a, x) = \EC{V_a}{X_a=x}$. Note that $\E{\EC{V_a}{X_a}} = \mu_a$ and $\Var{\EC{V_a}{X_a}} \le \sigma_a^2$,
    by the law of total variance, such that $\Var{\sum_{a\in A}v_X^{(1)}(a, X_a)}\leq \sigma_{XY}^2$.  
    Thus, by \Cref{eq:X-chernoff} and \Cref{lem:Benn-func-consts},
    \begin{align*}
        \PR{\abs{T_1} > t/3} & = \PR{\sum_{a\in A} \left(v_X^{(1)}(a, X_a)- \mu_a\right)> t/3}\\
            &< K_X\exp\left( -\sigma_{XY}^2 \chernoff{\frac{t/3}{\sigma_{XY}^2}} / D_X \right) + R_X\\
            &\le K_X\exp\left( - \sigma_{XY}^2 \chernoff{\frac{t}{\sigma_{XY}^2}} / (9 D_X) \right) + R_X
        \; .
    \end{align*}
        
    For bounding $\PR{\abs{T_2}>t/3}$, we may assume that $t> 6\sigma_{XY}^2$ since otherwise $T_2=0$ almost surely. Now,
    recall that $L_a =  \indicator{\VarC{V_a}{X_a} > \sqrt{6 \sigma_{XY}^2/t}}$ and write $Z=\sum_{a \in A} L_a$. We observe that since $V_a \in [0, 1]$ almost surely, $Z\geq \abs{T_2}$
    almost surely. By the law of total variance, $\E{\VarC{V_a}{X_a}} \le \sigma_a^2$, so
    by Markov's inequality,
    \[
        \E{L_a}
            = \PR{\VarC{V_a}{X_a} > \sqrt{\frac{6\sigma_{XY}^2}{t}}}
            \le \sigma_a^2 \sqrt{\frac{t}{6\sigma_{XY}^2}}.
    \]
    Now, $\Var{L_a}\leq \E{L_a}\leq \sigma_a^2 \sqrt{t/(6\sigma_{XY}^2)}$ as $L_a\in [0, 1]$ almost surely. Thus, $\E{Z} \le \sqrt{t \sigma_{XY}^2 / 6}$ and $\Var{Z} \le \sqrt{t \sigma_{XY}^2 / 6}$. Combining this with $t> 6\sigma_{XY}^2$, we may write
    \begin{align*}
        \PR{\abs{T_2} > t/3}
            \le \PR{Z - \E{Z} > t/3 - \sqrt{t \sigma_{XY}^2 / 6}}
            \le \PR{\abs{Z - \E{Z}} > t/6}.
    \end{align*}
    Applying \Cref{eq:X-chernoff} with the value function $v_X^{(2)}\colon A\times B_X\to [0,1]$ given by $v_X^{(2)}(a, X_a) = L_a$ to $\PR{\abs{Z - \E{Z}} > t/6}$ yields
    \[\begin{split}
        \PR{\abs{T_2} > t/3}
            &< K_X\exp\left(- \sqrt{t \sigma_{XY}^2 / 6}\cdot  \chernoff{\frac{t/6}{\sqrt{t \sigma_{XY}^2 / 6}}} / D_X \right) + R_X
            \\&\le K_X\exp\left(- \sigma_{XY}^2 \chernoff{\frac{t/6}{\sigma_{XY}^2}} / (4 D_X) \right) + R_X
            \\&\le K_X\exp\left(- \sigma_{XY}^2 \chernoff{\frac{t}{\sigma_{XY}^2}} / (144 \cdot D_X) \right) + R_X,
    \end{split}\]
    where the second follows from \Cref{lem:increase-var} and the third inequality follows from \Cref{lem:Benn-func-consts}.
    
    Lastly, we shall bound $\PR{\abs{T_3} > t/3}$. By a union bound,
    \[\begin{split}
        \PR{\abs{T_3} > t/3}
            &\le \underbrace{\PR{\left(\abs{T_3} > t/3\right) \, \wedge\, \left(\VarC{T_3}{(X_a)_{a \in A}} \le 2\max\set{\sigma_{XY}^2, \sqrt{t \sigma_{XY}^2}}\right)}}_{R_1}
            \\&+ \underbrace{\PR{\VarC{T_3}{(X_a)_{a \in A}} > 2\max\set{\sigma_{XY}^2, \sqrt{t \sigma_{XY}^2}}}}_{R_2}
        \; . 
    \end{split}\]
    First, we bound $R_1$. For each $a\in A$, let $x_a\in B_X$ be given such that $P(\forall a\in A\colon X_a=x_a)>0$. We bound the probability of $R_1$ conditioned on $(X_a=x_a)_{a\in A}$ and since our bound will be the same for every choice of $(x_a)_{a\in A}$, the bound will hold unconditionally. 
    Now, if $\VarC{T_3}{(X_a=x_a)_{a \in A}} > 2\max\set{\sigma_{XY}^2, \sqrt{t \sigma_{XY}^2}}$, then $R_1=0$. So assume otherwise and
define the value function $v_Y^{(1)}\colon A\times B_Y \to [0, 1]$ by 
    $v_Y^{(1)}(a, y) = S_a\cdot \overline v(a,x_a, y)$,
    where $S_a=\left[\VarC{V_a}{X_a=x_a}\leq\sqrt{6\sigma_{XY}^2/t}\right]$. Then $T_3 = \sum_{a\in A}\left(v_Y^{(1)}(Y_a) - \E{v_Y^{(1)}(Y_a)} \right)$ and by assumption, $\Var{\sum_{a\in A}v_Y^{(1)}(a, Y_a)}\leq 2\max\set{\sigma_{XY}^2, \sqrt{t \sigma_{XY}^2}}$. Thus, we may apply \Cref{eq:Y-chernoff} with $v_Y^{(1)}$ to obtain
    \[\begin{split}
        &\PRC{\left(\abs{T_3} > t/3\right) \, \wedge\, \left(\VarC{T_3}{(X_a)_{a \in A}} \le 2\max\set{\sigma_{XY}^2, \sqrt{t \sigma_{XY}^2}}\right)}{(X_a=x_a)_{a\in A}}
        \\&\phantom{==} \le K_Y\exp\left(- 2\max\set{\sigma_{XY}^2, \sqrt{t \sigma_{XY}^2}} \chernoff{\frac{t/3}{2\max\set{\sigma_{XY}^2, \sqrt{t \sigma_{XY}^2}}}} / D_Y \right) + R_Y
            \\&\phantom{==} \le K_Y\exp\left(- \max\set{\sigma_{XY}^2, \sqrt{t \sigma_{XY}^2}} \chernoff{\frac{t}{\max\set{\sigma_{XY}^2, \sqrt{t \sigma_{XY}^2}}}} / (18 D_Y) \right) + R_Y
            \\&\phantom{==} \le K_Y\exp\left(- \sigma_{XY}^2 \chernoff{\frac{t}{\sigma_{XY}^2}} / (72 D_Y) \right) + R_Y,
    \end{split}\]
    where the second follows from \Cref{lem:Benn-func-consts}
    and the third inequality follows from \Cref{lem:increase-var}. In conclusion,
    \begin{align*}
        R_1\leq K_Y\exp\left(- \sigma_{XY}^2 \chernoff{\frac{t}{\sigma_{XY}^2}} / (72 D_Y) \right) + R_Y.
    \end{align*}

    Second, we bound $R_2$. Define the value function $v_X^{(3)}\colon A\times B_X\to [0, 1]$ by 
    $$v_X^{(3)}(a, x_a) = \left[ \VarC{V_a}{X_a=x_a}\leq \sqrt{\frac{6\sigma_{XY}^2}t} \right]\cdot \VarC{V_a}{X_a=x_a}.$$
    Then $\VarC{T_3}{(X_a)_{a\in A}} = \sum_{a\in A} v_X^{(3)}(a, X_a)$. Now, 
    by the law of total variance, 
    $$\E{\VarC{T_3}{(X_a)_{a \in A}}} \le \Var{T_3} \leq \sigma_{XY}^2,$$
    and since $\sqrt{\frac{t}{6 \sigma_{XY}^2}} v_X^{(3)}(X_a) \in [0, 1]$ almost surely for every $a\in A$, pairwise independence yields
    \[
        \Var{\sqrt{\frac{t}{6 \sigma_{XY}^2}} \VarC{T_3}{(X_a)_{a \in A}}}
            \le \E{\sqrt{\frac{t}{6 \sigma_{XY}^2}} \VarC{T_3}{(X_a)_{a \in A}}}
            \le \sqrt{t \sigma_{XY}^2 / 6}
        \; .
    \]
    Applying \Cref{eq:X-chernoff} with $v_X^{(3)}$, \Cref{lem:increase-var}, and \Cref{lem:Benn-func-consts}, we obtain
    \[\begin{split}
        R_2& \le \PR{\abs{\VarC{T_3}{(X_a)_{a \in A}} - \E{\VarC{T_3}{(X_a)_{a \in A}} } }
                > \max\set{\sigma_{XY}^2, \sqrt{t \sigma_{XY}^2}}
            }
            \\& = \PR{\sqrt{\frac{t}{6 \sigma_{XY}^2}} \abs{\VarC{T_3}{(X_a)_{a \in A}} - \E{\VarC{T_3}{(X_a)_{a \in A}} } }
                > \max\set{\sqrt{t \sigma_{XY}^2 / 6}, t/6}
            }
            \\& \le \PR{\sqrt{\frac{t}{6 \sigma_{XY}^2}} \abs{\VarC{T_3}{(X_a)_{a \in A}} - \E{\VarC{T_3}{(X_a)_{a \in A}} } }
                > t/6
            }
            \\& < K_X\exp\left( - \sqrt{t \sigma_{XY}^2 / 6} \chernoff{\frac{t/6}{\sqrt{t \sigma_{XY}^2 / 6}}} / D_X \right) + R_X
            \\& \le K_X\exp\left( - \sigma_{XY}^2 \chernoff{\frac{t/6}{\sigma_{XY}^2}} / (4 D_X) \right) + R_X
            \\& \le K_X\exp\left( - \sigma_{XY}^2 \chernoff{\frac{t}{\sigma_{XY}^2}} / (144  D_X) \right) + R_X.
    \end{split}\]
Combining the bounds on $\PR{\abs {T_i}>t/3}$ for $i \in \{1, 2, 3\}$ completes the proof.
\end{proof}

\section{Query invariance}\label{sec:queryinvariance}

In the following, we will briefly explain for each of the main sections of the paper, why all theorems still hold when adding the condition of query invariance of \cref{def:query-invariance}. Recall that query invariance comes into play when we have a hash function and a concentration bound in the following manner. The concentration bound is query invariant if for any hash key $q$, a \emph{query key}, the concentration bound still holds whenever we condition the hash function on the hash value of $q$. 

\paragraph{Simple Tabulation Hashing.} In \cite{patrascu12charhash} it is observed that ordering the position characters $\alpha_1\prec \dots\prec \alpha_r$ such that $\alpha_1, \dots, \alpha_c$ are the position characters of the query key $q$ only worsens the bound on the groups, $G_i$, by a factor of 2. We consider a slightly more general case, but exactly the same argument still applies. Always imposing this ordering in our proofs lets us condition on the hash value of $q$ and only causes some of the constants to increase by a small factor. 

\paragraph{Tabulation-Permutation} In the proof of \cref{thm:valuefunctions} we consider some specific value function $w$. We proceed by considering separately the $m^\eps$ bins $S\subset [m]$ of largest contribution to the variance, $\sigma^2$, and then the remaining bins, $[m]\setminus S$. The contribution of each subset of bins is then individually bounded. In the first case, we simply use the assumption on the hash function $h$ that we received in a black box manner and use no properties of the permutation. Now, say towards query invariance that we require that $\tau\circ h(q) = i$. To support this, we instead chose $S$ to have have size $\abs S=m^\eps/2$. This does not change the proof by more than constant factors and simply adding $i$ to $S$ yields a set $S'=S\cup \{i\}$ of size $S'<m^\eps$, such that the assumption on $h$ directly yields the result. In conclusion, the proof goes through exactly as before.

\paragraph{Extending the Codomain} In this section nothing in the proof requires us to take into special consideration the conditioning on a query key. We simply consider families of hash functions in a black box manner and thus, we may as well consider families that have already been condition on the hash value of the query key $q$.
\section{Tightness of Concentration: Simple Tabulation into Few Bins}\label{sec:few-bin-problem}
	Recall the result of \cref{thm:intro-simple-tab}. If $h\colon [u]\to [m]$ is a simple tabulation hash function with $[u]=\Sigma^c$ and $c=O(1)$, and $S\subseteq [u]$ is a set of hash keys of size $n=\abs S$ where each key $x\in S$ is given a weight $w_x\in [0, 1]$. Then for arbitrary $y\in [m]$ and a constant $\gamma>0$ the total weight of the balls landing in bin $y$, given by the random variable $X=\sum_{x\in S}w_x[h(x)=y]$, satisfies the concentration bound
	\begin{align}\label{eq:thm1bound}
		\Pr\left[ \abs {X-\mu} \geq t \right]\leq 2\exp(-\Omega( \sigma^2\C(t/\sigma^2) )) + n/m^\gamma,
	\end{align} 
	where $\mu=\E X$ and $\sigma^2=\Var X$ are the expectation and variance, respectively, of $X$, and the constant in the $\Omega$-notation depends on $\gamma$. As mentioned in the introduction, the \errorterm $n/m^\gamma$ renders the theorem nearly useless for small $m$, the prime example being the tossing of an unbiased coin corresponding to $m=2$. The purpose of this section is to show that the bound of \eqref{eq:thm1bound} is optimal in the sense that an \errorterm of at least $m^{-\gamma}$ for some constant $\gamma$ is inevitable so long as we insist on strong concentration according to \cref{def:strongly-concentrated}. In other words, we must accept an \errorterm of $m^{-\gamma}$ to have Chernoff-style bounds on the sum $X$. In fact, it will turn out that unless allowing an error term of the form $m^{-\gamma}$, the deviation from the case of a fully random hash function can be quite significant.
	
The example where simple tabulation does not concentrate well, which we shall use in the formal proof below, is the following. For some $k	<\abs \Sigma$, we consider the key set $S=[k]^{c-1}\times \Sigma\subset \Sigma^c$ with weights $w_x=1$ for every $x\in S$. We shall think of $k$ as slightly superconstant and mutually dependent on $\gamma$. Recall that $h$ is defined by $c$ fully random functions $h_0, \dots, h_{c-1}\colon \Sigma\to [m]$ and that $h(x) = \bigoplus_{i=0}^c h_i(x_i)$. With probability $m^{-(k-1)(c-1)}$, $h_i$ is constant on $[k]$ for each $0\leq i\leq c-2$. Under such a \emph{collapse} it holds for every $\alpha \in \Sigma$ that every key from the set $[k]^{c-1}\times \{\alpha\}$ hashes to the same value in $[m]$ under $h$. Hence, each entry of $h_{c-1}$ decides where $k^{c-1}$ keys hash to. Thus, during such a collapse, we may view the hashing of $S$ into $[m]$ as throwing $\abs\Sigma$ balls each of weight $k^{c-1}$ into $m$ bins. This increases the variance by a factor of $k^{c-1}$ affecting the Chernoff bounds accordingly. 
	
Without further ado, let us present the formal statement of the above. Essentially, it states that there is a \emph{delay} of the exponential decrease which depends on $\gamma$. If $\gamma$ is superconstant, so is the delay, and hence, we do not have strong concentration according to \cref{def:strongly-concentrated}.
\begin{theorem}
	Let $m\leq\abs \Sigma^{1-\eps}$ for some constant $\eps>0$ and $h\colon [u]\to [m]$ be a simple tabulation hash function. Let $0<\eps'<\eps$ be a constant and suppose that $C\colon \R^+\to \R^+$ is a function such that for all $0\leq \gamma \leq \abs \Sigma^{\eps'/c}$, all sets $S\subseteq [u]$, all choices of weights $w_x\in [0, 1], x\in S$, and every $y\in [m]$, the random variable $X=\sum_{x\in S} w_x[h(x)=y]$ satisfies
	\begin{align}\label{eq:hypotheticalbound}
		\Pr[|X-\mu|\geq t] \le 2\exp\left(\frac{-\sigma^2\cC(t/\sigma^2)}{C(\gamma)}\right)+m^{-\gamma}
	\end{align}
	for all $t>0$. Then $C(\gamma) = \Omega(\gamma^{c-2})$.
\end{theorem}
\begin{proof}
	Assume the existence of the function $C$. As suggested above, consider the set of keys $S=[k]^{c-1}\times \Sigma$ for some $k$ to be determined. Denote by $\mathcal E$ the event that the first $c-1$ position characters of $S$ collapse, i.e., that each $h_i, 0\leq i\leq c-2$ is constant on $[k]$. 
	We easily calculate $\Pr[\mathcal E] = m^{-(k-1)(c-1)}$. Now, conditioning on $\mathcal E$, we may consider the situation as follows. Let $y'$ be the random variable satisfying $\bigoplus_{i=0}^{c-2}h_i(x_i)=y'$ for all $x_0, \dots, x_{c-2}\in [k]$. The last positional hash function $h_{c-1}$ is a fully random hash function $\Sigma\to m$ such that the conditioned variable $(X|\mathcal E)$ satisfies 
	$$(X|\mathcal E)= \sum_{\alpha\in \pi_{c-1}(S)}k^{c-1}[h_{c-1}(\alpha) = y\xor y'] \stackrel d= \sum_{\alpha\in \Sigma}k^{c-1}[h_{c-1}(\alpha) = 0] =: X',$$ 
	where $\stackrel d=$ denotes equality of distribution.
	We write $\sigma'^2=\Var {X'} = k^{2(c-1)}\abs \Sigma \frac{m-1}{m^2}$ and note that $\E{X'}=\mu$. Now, since $h_{c-1}$ is a uniformly random hash function, tightness of the Bennet inequality, \cref{eq:var-chernoff}, implies that for $t=O(\sigma'^2)$,
	\begin{align}\label{eq:conditioned_lower_bound_on_cher}
		\Pr\left[ \abs{X'- \mu} \geq t\right] =\Omega\left( \exp\left( -\sigma'^2 \C (t/\sigma'^2) \right) \right) = \Omega\left(\exp\left(-\frac{t^2}{3\sigma'^2}\right)\right)
	\end{align}
	where we have applied \cref{rateofgrowth}. 
	
	Towards our main conclusion, observe that $\sigma^2 = k^{c-1}\abs \Sigma \frac{m-1}{m^2} = \sigma'^2/k^{c-1}$. Letting $t = \sigma'\sqrt{\log m}$,  $t\leq \sigma'^2$ since $\sigma'>\sqrt{\log m}$ by the assumption on the size of $m$, so we may apply \eqref{eq:conditioned_lower_bound_on_cher} to get
	\begin{align*}
		\Pr\left[\abs{X-\mu}\geq t\right] \geq \Pr\left[ \mathcal E \right]\cdot \Pr\left[ \abs{X'- \mu} \geq t\right] \geq \Omega\left(m^{-ck}\right).
	\end{align*} 
	On the other hand, \eqref{eq:hypotheticalbound} demands that whenever $k\leq \gamma$,
	\begin{align*}
		\Pr\left[\abs{X-\mu}\geq t\right] \leq 2\exp\left( -\frac{\sigma^2\C\left(\sqrt{\log(m)k^{c-1}}/\sigma\right)}{C(\gamma)} \right)+m^{-\gamma}\leq 2m^{-\frac{k^{c-1}}{C(\gamma)}} + m^{-\gamma},
	\end{align*}
	where the last inequality used \cref{rateofgrowth} and $\sqrt{\log(m)k^{c-1}}<\sigma$. Let $k = \frac{\gamma}{2c}$ and combine the above inequalities to conclude that 
		$2m^{-\frac{(\gamma/(2c))^{c-1}}{C(\gamma)}} + m^{-\gamma} = \Omega(m^{-\gamma/2})$. It follows that indeed, $C(\gamma) = \Omega(\gamma^{c-2})$.
	
\end{proof}

\paragraph{Twisted Tabulation and ``permutation-tabulation''}
Variations upon the example above can also be used to show that the analysis of twisted tabulation hashing is tight in the sense that the \errorterm cannot be improved while maintaining strong concentration. In twisted tabulation we \emph{twist} the last position character of the input before applying simple tabulation. The twist is a Feistel permutation that for the key set
$S=[k]^b\times \Sigma^{c-b}$, will only permute the keys within
the set. Since the set of twisted keys is the same as the original set $S$, this has no effect on the filling of bins.
For almost the same reason, a reversal of the order of operations in our new tabulation-permutation hashing, i.e., if is first permuted each position character and the applied simple tabulation, would not improve the analysis, since the set $S$ while not invariant under the operation, would retain the same structure.

\section*{Acknowledgement}
Anders Aamand, Jakob B.\ T.\ Knudsen, Peter M.\ R.\ Rasmussen, and Mikkel Thorup are partly supported by Thorup's Investigator Grant 16582, Basic Algorithms Research Copenhagen (BARC), from the VILLUM Foundation.

\bibliographystyle{acm}
\bibliography{bib}
%\clearpage
%\appendix
%\centerline{\huge\bf Appendix}
%\input{sections/streaming}
%\newpage
\appendix
\section{Experiments}\label{sec:experiments}
This section is dedicated to provide further details regarding the timing experiments presented in the introduction in \cref{sec:intro-exp}. Furthermore, we present experiments which demonstrate concrete bad input sets for several hash functions that do not guarantee strong concentration bounds.

As explained in~\Cref{sec:intro-exp}, we ran experiments on various basic hash functions. 
More precisely, we compared our new hashing schemes tabulation-permutation and tabulation-1permutation with the following hashing schemes: $k$-independent
PolyHash~\cite{carter77universal},
Multiply-Shift~\cite{dietzfel96universal}, simple
tabulation~\cite{zobrist70hashing}, twisted
tabulation~\cite{PT13:twist},
mixed tabulation~\cite{DKRT15:k-part},
 and double
tabulation~\cite{Tho13:simple-simple}.
We were interested in both the speed of the hash functions involved, and the quality of the output. For our timing experiments we studied the hashing of 32-bit keys to 32-bit hash values, and 64-bit keys to 64-bit hash values.
%We evaluated both their speed and  the quality of the output. 
Aside from having strong theoretical guarantees, our experiments show that tabulation-permutation and tabulation-1permutation are very fast in practice. 

 \begin{table}
\begin{center}
    \begin{tabular}{|l | r | r | r | r |}
        \cline{2-5}
        \multicolumn{1}{l|}{} & \multicolumn{4}{|c|}{Running time (ms)} \\
        \cline{2-5}
        \multicolumn{1}{l|}{} \multirow{2}{*}{}&\multicolumn{2}{|c|}{Computer 1} & \multicolumn{2}{|c|}{Computer 2} \\
        \hline Hash function                  & 32 bits & 64 bits & 32 bits & 64 bits \\
        \hline
        \emph{Multiply-Shift}                 & 4.2     &  7.5   & 23.0     & 36.5    \\
        \emph{2-Independent PolyHash}         & 14.8    & 20.0   & 72.2     & 107.3   \\
        \emph{Simple tabulation}              & 13.7    & 17.8   & 53.1     & 55.9    \\
        \emph{Twisted tabulation}             & 17.2    & 26.1   & 65.6     & 92.5    \\ 
        \emph{Mixed tabulation}               & 28.6    & 68.1   & 120.1    & 236.6   \\
        \hline
        \textbf{Tabulation-1permutation}      & 16.0    & 19.3   & 63.8     & 67.7    \\
        \textbf{Tabulation-permutation}       & 27.3    & 43.2   & 118.1    & 123.6   \\ \hline
        Double tabulation                     & 1130.1  & --     & 3704.1   & --      \\
        ``Random'' (100-Independent PolyHash) & 2436.9  & 3356.8 & 7416.8   & 11352.6 \\
        \hline
    \end{tabular}
    \caption{The time for different hash functions to hash $10^7$ keys of  
    length 32 bits and 64 bits, respectively, to ranges of size 32 bits and 64 bits. The experiment was carried out on two computers. The hash functions written in italics are those without general Chernoff-style bounds. Hash functions written in bold are the contributions of this paper. The hash functions in regular font are known to provide Chernoff-style bounds. Note that we were unable to implement double tabulation from 64 bits to 64 bits since the hash tables were too large to fit in memory. \label{tab:speed2}}
\end{center}
\end{table}

All experiments are implemented in C++11 using a random seed from
\url{https://www.random.org}. The seed for the tabulation based hashing methods
uses a random 100-independent PolyHash function. PolyHash is implemented
using the Mersenne primes $p = 2^{61} - 1$ for 32 bits and $p = 2^{89} - 1$ for 64 bits.
Furthermore, it has been implemented using Horner's rule, and GCC's 128-bit integers to ensure
an efficient implementation. Double tabulation is implemented as described
in \cite{Tho13:simple-simple} with $\Sigma = [2^{16}], c = 2, d = 20$.

\paragraph{Timing}
We timed the speed of the hash functions on two different computers. The first computer (Computer 1 in~\Cref{tab:speed2})
has a 2.4 GHz Quad-Core Intel Core i5 processor and 8 GB RAM, and it is running
macOS Catalina. The second computer (Computer 2 in~\Cref{tab:speed2})
has 1.5 GHz Intel Core i3 processor and 4 GB RAM, and it is running Windows 10. We restate the results of our experiments in~\Cref{tab:speed2} and refer the reader to~\Cref{sec:intro-exp} for a discussion of these results and of the choice of parameters used in the various hashing schemes. %We hash the same $10^7$ randomly chosen integers with each hash function. We consider the case where the hash functions output 32 bits and when the hash functions output 8 bits. 

%First we compare with the hash functions that do not have good concentration bounds, that is, Multiply-Shift, 2-independent PolyHash, simple tabulation, and twisted tabulation. We see that when the hash functions output 32 bits then tabulation-permutation is slightly slower than simple tabulation and twisted tabulation which again are slower than Multiply-Shift and 2-independent PolyHash. Next we consider the case where we only have a single 8-bit character output. Recall that this is the case we needed when using tabulation-permutation to hash into a small number $m$ of bins, e.g. $m=2$. Having a small output is only an advantage for simple tabulation, twisted tabulation, and tabulation-permutation. For the latter it means that we only need to perform one more table-lookup than simple tabulation which is also reflected in the running time where tabulation-permutation is only marginally slower than simple tabulation and faster than twisted tabulation 8-bit output.

%Next, comparing with the hash functions that do strong concentration bounds, tabulation-permutation is approximately 20 times faster than double tabulation and approximately 125 times faster than 100-independent PolyHash. 

\paragraph{Quality}
We will now present experiments with concrete bad instances for the schemes
without general concentration bounds, that is, Multiply-Shift, 2-independent 
PolyHash, simple tabulation, and twisted tabulation. In each case, we compare
with our new tabulation-permutation scheme as well as 100-independent PolyHash,
which is our approximation to an ideal fully random hash function.
We note that all schemes considered are 2-independent, so they all have
exactly the same variance as fully-random hashing. From $2$-independence, it also follows that the schemes work perfectly on sufficiently random input~\cite{mitzenmacher08hash}. Our concern is therefore concrete inputs making them fail in the tail.
%the frequency of large deviations in the tails.

%We also note that all schemes work perfectly on sufficiently random input  \cite{mitzenmacher08hash}, so the question is if there exists concrete inputs making them fail in the tail. 

First, we consider simple bad instances for Multiply-Shift and 2-independent
PolyHash.  These are analyzed in detail in~\cite[Appendix
B]{patrascu16kwise-lb}.  The specific instance we consider is that of hashing the arithmetic
progression $A=\setbuilder{a \cdot i}{i \in [50000]}$ into $16$ bins, where we are interested in the number of keys from $A$ that hashes to a specific bin. We performed this experiment $5000$ times, with independently chosen hash functions. The cumulative distribution functions on the number of keys from $A$ hashing to a specific bin is presented in~\Cref{fig:polynomials}. We see that most of the
time 2-independent PolyHash and Multiply-Shift distribute the keys perfectly with exactly $1/16$ of
the keys in our bin. Since the variance is the same as with fully random hashing, this should suggest a  much heavier tail,
which is indeed what our experiments show. For contrast, we see that the cumulative distribution function with our
tabulation-permutation hash function is almost indistinguishable from that of 100-independent Poly-Hash. We note that our experiments with tabulation-permutation is only a sanity check: No experiment can prove good performance on all possible inputs. 

\begin{figure}
    \centering
    \def\svgwidth{\textwidth}
    \input{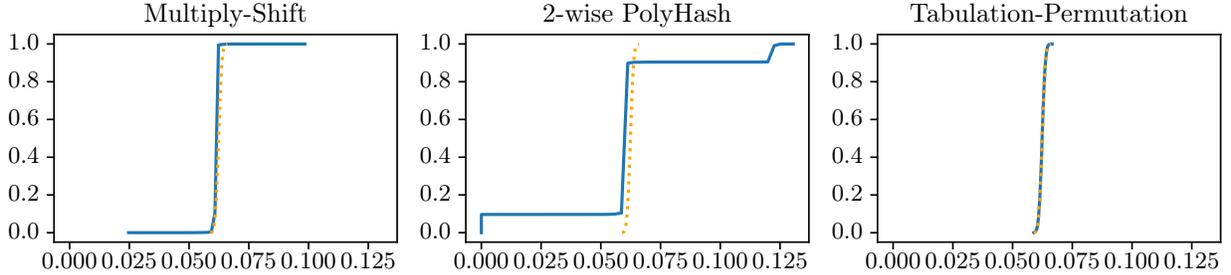}
    \caption{Hashing the arithmetic progression $\setbuilder{a \cdot i}{i \in [50000]}$
    to $16$ bins for a random integer $a$. The dotted line is a 100-independent PolyHash.}
    \label{fig:polynomials}
\end{figure}

Our second set of experiments shows bad instances for simple tabulation and twisted tabulation. We already know theoretically from~\Cref{sec:few-bin-problem} that these bad instances exist, but we shall now see that, in a sense, things can be even worse than described in~\Cref{sec:few-bin-problem} for certain sets of keys. The specific instance we consider is hashing the discrete cube $Q=[2]^7 \times [2^6]$ to $m=2$ bins using simple
tabulation, twisted tabulation, and tabulation-permutation. We performed this experiment $5000$ times, with independently chosen hash functions, and again we were interested in the number of keys from $Q$ hashing to one of the bins. The cumulative distribution functions of the number of such keys is presented in~\Cref{fig:tabulation}.  
Let us explain the appearance of the curves for simple and twisted tabulation.
In general, if we hash
the set of keys $[2] \times R$ to $[2]$ with simple tabulation, then if
$h_1(0) \neq h_1(1)$, each bin will get exactly the same amount of
keys.  When we hash the set of keys $[2]^7 \times [2^6]$ this happens with
probability $1 - 2^{-7}$. If on the other hand $h_i(0)=h_i(1)$ for each $i=1,\dots,7$, which happens with probability $2^{-7}$, the distribution of the balls in the bins is the same as that when $2^6$ balls, each of weight $2^7$, are distributed independently and uniformly at random into the two bins. If this happens, the variance of the number of balls in a bin is a factor of $2^7$ higher, so we expect a much heavier tail than in the
completely independent case. These observations agree with the results in~\Cref{fig:tabulation}. Most of the time, the distribution is perfect, but the tail is very heavy.
We believe that this instance is also one of the
worst instances for tabulation-permutation hashing. We would therefore expect to see that on this instance it performs slightly worse than 100-independent
PolyHash, which is indeed what our experiments show. We note that
that no amount of experimentation can prove that tabulation-permutation
always works well for all inputs. We do, however, have mathematical concentration guarantees, and the experiments performed here give
us some idea of the impact of the constant delay hidden in the exponential decrease in the bounds of~\Cref{thm:intro-tab-perm}. For completeness, we note that the situation with mixed tabulation is unresolved. Neither do we have strong concentration bounds, nor any bad instances showing that such bounds do not hold. Running experiments is not expected to resolve this issue since mixed tabulation, as any other $2$-independent hashing scheme, performs well on almost all inputs~\cite{mitzenmacher08hash}.   

\begin{figure}
    \centering
    \def\svgwidth{\textwidth}
    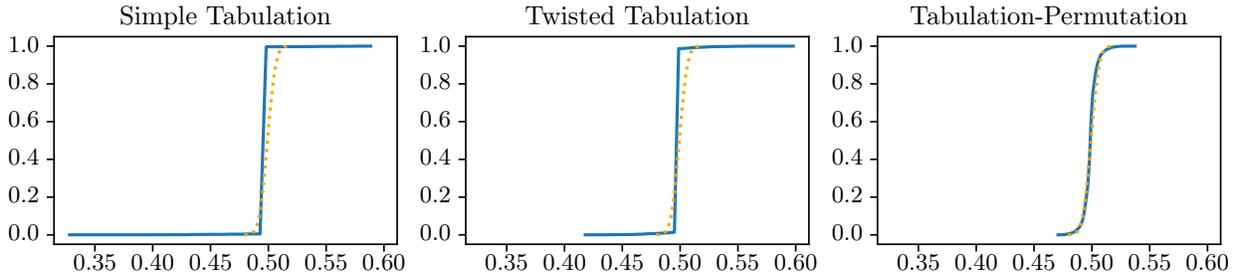
    \caption{Hashing the discrete cube $[2]^7  \times [2^6]$ to $2$ bins. The dotted line
    is a 100-independent PolyHash.}
    \label{fig:tabulation}
\end{figure}

\end{document}